\documentclass[11pt]{article}
\usepackage{amsmath,amssymb,amsthm,bm,subfig,fullpage,xspace,graphicx}
\usepackage[colorlinks=true]{hyperref}
\usepackage{algorithm,algorithmicx}
\usepackage[noend]{algpseudocode}
\newtheorem{theorem}{Theorem}[section]
\newtheorem{lemma}[theorem]{Lemma}

\newtheorem{proposition}[theorem]{Proposition}

\newtheorem{corollary}[theorem]{Corollary}

\theoremstyle{definition}
\newtheorem{definition}[theorem]{Definition}
\newtheorem{remark}[theorem]{Remark}
\newtheorem{example}[theorem]{Example}

\newcommand{\caI}{\mathcal{I}}
\newcommand{\caL}{\mathcal{L}}

\newcommand{\caN}{\mathcal{N}}

\newcommand{\caR}{\mathcal{R}}

\newcommand{\bbN}{\mathbb{N}}
\newcommand{\bbR}{\mathbb{R}}
\newcommand{\bmb}{{\bm b}}
\newcommand{\bmd}{{\bm d}}
\newcommand{\bmg}{{\bm g}}
\newcommand{\bmp}{{\bm p}}
\newcommand{\bmv}{{\bm v}}
\newcommand{\bmw}{{\bm w}}
\newcommand{\bmx}{{\bm x}}
\newcommand{\bmy}{{\bm y}}
\newcommand{\bmz}{{\bm z}}
\newcommand{\bmone}{{\bm 1}}
\newcommand{\bmzero}{{\bm 0}}

\newcommand{\bmeta}{{\bm \eta}}
\newcommand{\dr}{\mathrm{d}r}
\newcommand{\dt}{\mathrm{d}t}
\newcommand{\dx}{\mathrm{d}\bmx}
\newcommand{\dtau}{\mathrm{d}\tau}
\newcommand{\cut}{\mathrm{cut}}
\newcommand{\vol}{\mathrm{vol}}
\newcommand{\supp}{\mathrm{supp}}
\newcommand{\set}[1]{\{#1\}}
\newcommand{\bigset}[1]{\bigl\{#1\bigr\}}
\newcommand{\Bigset}[1]{\Bigl\{#1\Bigr\}}

\newcommand{\thr}{\mathrm{thr}}

\newcommand{\E}{\mathop{\mathbf{E}}}
\newcommand{\argmax}{\mathop{\mathrm{argmax}}}
\newcommand{\argmin}{\mathop{\mathrm{argmin}}}
\newcommand{\Lovasz}{Lov{\'a}sz\xspace}

\newcommand{\SDP}{\mathrm{SDP}}
\newcommand{\poly}{\mathrm{poly}}

\allowdisplaybreaks

\title{Cheeger Inequalities for Submodular Transformations}
\author{Yuichi Yoshida\thanks{Supported by JST ERATO Grant Number JPMJER1305 and JSPS KAKENHI Grant Number JP17H04676.}\\
National Institute of Informatics\\
\texttt{yyoshida@nii.ac.jp}}

\begin{document}
\maketitle
\begin{abstract}
  The Cheeger inequality for undirected graphs, which relates the conductance of an undirected graph and the second smallest eigenvalue of its normalized Laplacian, is a cornerstone of spectral graph theory.
  The Cheeger inequality has been extended to directed graphs and hypergraphs using normalized Laplacians for those, that are no longer linear but piecewise linear transformations.

  In this paper, we introduce the notion of a submodular transformation $F:\set{0,1}^n \to \bbR^m$, which applies $m$ submodular functions to the $n$-dimensional input vector, and then introduce the notions of its Laplacian and normalized Laplacian.
  With these notions, we unify and generalize the existing Cheeger inequalities by showing a Cheeger inequality for submodular transformations, which relates the conductance of a submodular transformation and the smallest non-trivial eigenvalue of its normalized Laplacian.
  This result recovers the Cheeger inequalities for undirected graphs, directed graphs, and hypergraphs, and derives novel Cheeger inequalities for mutual information and directed information.

  Computing the smallest non-trivial eigenvalue of a normalized Laplacian of a submodular transformation is NP-hard under the small set expansion hypothesis.
  In this paper, we present a polynomial-time $O(\log n)$-approximation algorithm for the symmetric case, which is tight, and a polynomial-time $O(\log^2n+\log n \cdot \log m)$-approximation algorithm for the general case.

  We expect the algebra concerned with submodular transformations, or \emph{submodular algebra}, to be useful in the future not only for generalizing spectral graph theory but also for analyzing other problems that involve piecewise linear transformations, e.g., deep learning.
\end{abstract}

\thispagestyle{empty}
\newpage

\tableofcontents
\thispagestyle{empty}
\setcounter{page}{0}
\newpage

%!TEX root=./main.tex

\section{Introduction}

\subsection{Background}

Spectral graph theory is concerned with the relations between the properties of a graph and the eigenvalue/vectors of matrices associated with the graph (refer to~\cite{chung1997spectral} for a book).
One of the most seminal results in spectral graph theory is the Cheeger inequality~\cite{Alon:1986gz,Alon:1985jg}, which we briefly review below.
Let $G = (V,E)$ be an undirected graph.
The \emph{conductance} of a vertex set $\emptyset \subsetneq S \subsetneq V$ is defined as
\[
  \phi_G(S) = \frac{\cut_G(S)}{\min\set{\vol_G(S),\vol_G(V\setminus S)}},
\]
where the \emph{cut size} of $S$, denoted by $\cut_G(S)$, is the number of edges between $S$ and $V \setminus S$, and the \emph{volume} of $S$, denoted by $\vol_G(S)$, is the sum of degrees of the vertices in $S$.
The \emph{conductance} $\phi_G$ of $G$ is the minimum conductance of a vertex set $\emptyset \subsetneq S \subsetneq V$.
The problem of finding a vertex set of a small conductance has been intensively studied because such a set can be regarded as a tight community~\cite{Gleich:2012kq,Leskovec:2010fp}.
Although computing $\phi_G$ is an NP-hard problem, we can well approximate it using the Cheeger inequality, which relates $\phi_G$ and an eigenvalue of a matrix constructed from $G$ known as the normalized Laplacian.
Here, the \emph{Laplacian} of $G$ is the matrix $L_G = D_G - A_G$, where $D_G \in \bbR^{V \times V}$ is the diagonal matrix consisting of the degrees of vertices and $A_G \in \bbR^{V \times V}$ is the adjacency matrix, and the \emph{normalized Laplacian} of $G$ is the matrix $\caL_G= D_G^{-1/2}L_G D_G^{-1/2} = I - D_G^{-1/2}A_G D_G^{-1/2}$.
Then, the Cheeger inequality~\cite{Alon:1986gz,Alon:1985jg} states that
\begin{align}
  \frac{\lambda_G}{2} \leq \phi_G \leq \sqrt{2\lambda_G}, \label{eq:undirected-cheeger}
\end{align}
where $\lambda_G$ is the second smallest eigenvalue of $\caL_G$ (note that the smallest eigenvalue is zero with the corresponding trivial eigenvector $D_G^{1/2}\bmone$, where $\bmone$ is the all-one vector).
Indeed, the second inequality of~\eqref{eq:undirected-cheeger} yields an algorithm, which computes a set $\emptyset \subsetneq S \subsetneq V$ of conductance at most $\sqrt{2\lambda_G} = O(\sqrt{\phi_G})$ from an eigenvector corresponding to $\lambda_G$.
Moreover, the Cheeger inequality is tight in the sense that computing a set with a conductance $o(\sqrt{\phi_G})$ is NP-hard~\cite{Raghavendra:2012fc} assuming the small set expansion hypothesis (SSEH)~\cite{Raghavendra:2010jg}.

Extensions of the Cheeger inequality were recently proposed for directed graphs~\cite{Yoshida:2016ig} and hypergraphs~\cite{Chan:2018eu,Louis:2014tg} by using modified notions of conductance and a normalized Laplacian.
We note that normalized Laplacians for directed graphs and hypergraphs are no longer linear but piecewise linear transformations.
We can show that those normalized Laplacians always have the eigenvalue of zero associated with a trivial eigenvector, and that they also have a non-trivial eigenvalue in the sense that the corresponding eigenvector is orthogonal to the trivial eigenvector.
Then, the extended Cheeger inequalities~\cite{Chan:2018eu,Louis:2014tg,Yoshida:2016ig} relate the conductance of a directed graph or a hypergraph with the smallest non-trivial eigenvalue of its normalized Laplacian.
However, as those normalized Laplacians are no longer linear transformations, computing its smallest non-trivial eigenvalue becomes NP-hard under the SSEH~\cite{Chan:2018eu,Louis:2014tg}.
Although a polynomial-time $O(\log n)$-approximation algorithm is known for hypergraphs on $n$ vertices, no non-trivial polynomial-time approximation algorithm is known for directed graphs.

\subsection{Our contributions}

In this paper, we unify and extend the existing Cheeger inequalities discussed above by introducing the notions of a submodular transformation and its normalized Laplacian.
A set function $F\colon\set{0,1}^V \to \bbR$ is called \emph{submodular} if $F(S) + F(T) \geq F(S \cap T) + F(S \cup T)$ for every $S,T\subseteq V$.
We note that the cut function $\cut_G\colon\set{0,1}^V \to \bbR$ associated with an undirected graph, a directed graph, or a hypergraph $G$ is submodular, where $\cut_G(S)$ for a vertex set $S$ represents the number of edges, arcs, or hyperedges leaving $S$ and entering $V \setminus S$.
We say that a function $F\colon\set{0,1}^V \to \bbR^E$ is a \emph{submodular transformation} if $F_e\colon S \mapsto F(S)(e)$ is a submodular function for every $e \in E$.

To derive a Cheeger inequality for a submodular transformation $F\colon\set{0,1}^V \to \bbR^E$ with $F(\emptyset)=F(V)=\bmzero$, we need to define the conductance of a set with respect to $F$ and the normalized Laplacian associated with $F$.
First, we define the \emph{degree} $\bmd_F(v)$ of $v \in V$ as the number of $F_e$'s to which $v$ is relevant. (See Section~\ref{sec:pre} for the formal definition.)
For a set $S \subseteq V$, we define the \emph{volume} of $S$ as $\vol_F(S) = \sum_{v \in S}\bmd_F(v)$ and the \emph{cut size} of $S$ as $\cut_F(S) = \sum_{e \in E}F_e(S)$.
Then, we define the \emph{conductance} $\phi_F(S)$ of a set $\emptyset \subsetneq S \subsetneq V$ as
\[
  \phi_F(S) = \frac{\min \set{\cut_F(S),\cut_F(V \setminus S)}}{\min\set{\vol_F(S),\vol_F(V\setminus S)}}.
\]
We define the \emph{conductance} of $F$ as $\phi_F = \min_{\emptyset \subsetneq S \subsetneq V}\phi_F(S)$.

\begin{example}\label{ex:undirected-conductance}
  Let $G=(V, E)$ be an undirected graph.
  Now, we consider a submodular transformation $F\colon\set{0,1}^V \to \bbR^E$, where $F_e$ is the cut function of the undirected graph with a single edge $e$.
%  $f_e:\bbR^V \to \bbR$ for an edge $e = \set{u,v} \in E$ is defined as $f_e(\bmx) = |\bmx(u) - \bmx(v)|$, which is the \Lovasz extension of the
  % Then, for each edge $e = \set{u,v} \in E$, we define $F_e:\set{0,1}^V \to \set{0,1}$ as the cut function of the undirected graph with a single edge $e$.
  % Then, its \Lovasz extension $f_e:\bbR^V \to \bbR$ becomes
  Then, $\bmd_F(v)$ for a vertex $v \in V$ coincides with the usual degree of $v$, and $\cut_F(S)$ for a vertex set $S \subseteq V$ coincides with the usual cut size of $S$.
  As $\cut_F$ is symmetric, that is, $\cut_F(S) = \cut_F(V\setminus S)$ holds for every vertex set $S \subseteq V$, $\phi_F(S)$ coincides with the conductance of $S$ in the graph sense.
\end{example}

% For a set function $F:\set{0,1}^V \to \bbR$, we can define its \emph{\Lovasz extension} $f:\bbR^V \to \bbR$ as $f(\bmx) = \max_{\bmw \in B(F)}\langle \bmw, \bmx\rangle$, where $B(F) \subseteq \bbR^V$ is the \emph{base polytope} of $F$ (see Section~\ref{sec:pre} for the definition).
% Here, we note that $f(\bmone_S) = F(S) \;(S \subseteq V)$ holds, where $\bmone_S \in \bbR^V$ is the characteristic vector of $S$.
% Also, note that $f$ is a piecewise linear function of $\bmx$.
%It is known that $f$ is convex if and only if $F$ is submodular~\cite{Lovasz:1983jc}.
% We say that a function $f:\bbR^V \to \bbR^E$ is a \emph{submodular transformation} if $f_e:\bmx \mapsto f(\bmx)(e)$ is the \Lovasz extension of a submodular function for every $e \in E$.

For a set function $F\colon\set{0,1}^V \to \bbR$, we define its \emph{\Lovasz extension} $f\colon\bbR^V \to \bbR$ as $f(\bmx) = \max_{\bmw \in B(F)}\langle \bmw, \bmx\rangle$, where $B(F) \subseteq \bbR^V$ is the \emph{base polytope} of $F$ (see Section~\ref{sec:pre} for the definition).
Then, using a submodular transformation $F\colon\set{0,1}^V \to \bbR^E$ with $F(\emptyset)=F(V)=\bmzero$, we define its \emph{Laplacian} $L_F\colon\bbR^V \to \set{0,1}^{\bbR^V}$ as
\[
  L_F(\bmx) := \left\{\sum_{e \in E}\bmw_e f_e(\bmx) \mid \bmw_e \in \partial f_e(\bmx) \; (e \in E)\right\} = \left\{\sum_{e \in E}\bmw_e \langle \bmw_e, \bmx\rangle \mid \bmw_e \in \partial f_e(\bmx)\; (e \in E)\right\},
\]
where $f_e\colon \bbR^V \to \bbR$ is the \Lovasz extension of $F_e$ and $\partial f_e(\bmx) := \argmax_{\bmw \in B(F_e)}\langle \bmw,\bmx\rangle \subseteq \bbR^V$ is the subdifferential of $f_e$ at $\bmx$.
See Section~\ref{sec:Laplacian} for more detailed explanation.
We note that $L_F$ is set-valued and $L_F(\bmx)$ forms a convex polytope in $\bbR^V$.
However, $L_F(\bmx)$ consists of a single point almost everywhere (as so does $\partial f_e(\bmx)$), and hence we can almost always regard $L_F$ as a function that maps a vector in $\bbR^V$ to another vector in $\bbR^V$.
Moreover, around $\bmx \in \bbR^V$ with $L_F(\bmx)$ consisting of a single point, $L_F$ acts as a linear transformation.
Hence, we can basically regard $L_F$ as a piecewise linear function.

Next, we define the \emph{normalized Laplacian} $\caL_F\colon\bbR^V \to \bbR^V$ as $\caL_F(\bmx) = D_F^{-1/2}L_F(D_F^{-1/2}\bmx)$, where $D_F\in \bbR^{V \times V}$ is a diagonal matrix with ${(D_F)}_{vv} = \bmd_F(v)\;(v \in V)$.
We say that $\lambda \in \bbR$ is an \emph{eigenvalue} of $\caL_F$ if there exists a non-zero vector $\bmv \in \bbR^{V}$ such that $\caL_F(\bmv) \ni \lambda\bmv$.
As with the normalized Laplacian for an undirected graph, using the assumption $F(\emptyset) = F(V) = \bmzero$, we can show that $\caL_F$ is positive-semidefinite, that is, all the eigenvalues are non-negative, and that $\caL_F(D_F^{1/2}\bmone) \ni \bmzero$, that is, $0$ is the smallest eigenvalue of $\caL_F$ with the corresponding trivial eigenvector $D_F^{1/2}\bmone$.
Then, we can also show that there exists a non-trivial eigenvalue in the sense that the corresponding eigenvector is orthogonal to $D_F^{1/2}\bmone$.
%We note that the condition $f(\bmzero) = f(\bmone)=\bmzero$ is equivalent to $F_e(\emptyset) = F_e(V)=0$ for every $e \in E$, which is satisfied for cut functions of undirected graphs, directed graphs, and hypergraphs.
We denote by $\lambda_F$ the smallest non-trivial eigenvalue of $\caL_F$.

\begin{example}\label{ex:undirected-Laplacian}
  For an undirected graph $G=(V,E)$, we define a submodular transformation $F\colon\set{0,1}^V \to \bbR^E$ as in Example~\ref{ex:undirected-conductance}.
  Then, $\caL_F$ essentially equals to the usual normalized Laplacian $\caL_G$ for $G$ because $\caL_F(\bmx)$ consists of a single vector $L_G \bmx$. (See Example~\ref{ex:undirected-Laplacian-detailed} for details.)
  Moreover, $\lambda_F$ is equal to the second smallest eigenvalue of $\caL_G$.
\end{example}

We show the following Cheeger inequality that relates $\phi_F$ and $\lambda_F$:
%We say that a submodular function $F$ is normalized if $F(\emptyset) = 0$.
%In this paper, we only consider normalized submodular functions.
\begin{theorem}\label{the:intro-cheeger}
  Let $F\colon\set{0,1}^V \to \bbR^E$ be a submodular transformation with $F(\emptyset) = F(V) = \bmzero$ and $F(S) \in [0,1]$ for every $S \subseteq V$.
%  $f(\bmx) \in [0,1]^E$ for $\bmx \in [0,1]^V$ and $f(\bmone) = \bmzero$.
%  Let $\lambda_F$ be the smallest non-trivial eigenvalue of the normalized Laplacian $\caL_F:\bbR^V \to \set{0,1}^{\bbR^V}$.
  Then, we have
  \[
    \frac{\lambda_F}{2} \leq \phi_F \leq 2\sqrt{\lambda_F}.
  \]
\end{theorem}
%The condition $f(\bmx) \leq \bmone$ for $\bmx \leq \bmone$ reduces to $F_e(S) \leq 1$ for every $e \in E$ and $S \subseteq V$, where $F_e$ is the submodular function associated with $f_e:\bmx \to f(\bmx)(e)$.
%In this paper, we mainly consider submodular transformations $f:\bbR^V \to \bbR^E$ with $f(\bmx) \in [0,1]^E$ for $\bmx \in [0,1]^V$.
%This condition is equivalent to $F_e(S) \in [0,1]$ for every $e \in E$ and $S \subseteq V$, and it follows that $\phi_F \in [0,1]$.
%Although we can rephrase the requirement on $F_e$ in Theorem~\ref{the:eigenvalues} using $f$ only, we believe the current form is more intuitive.

We now see several instantiations of Theorem~\ref{the:intro-cheeger}.
\begin{example}[Undirected graphs]\label{ex:undirected-graph}
  For an undirected graph $G=(V,E)$, we define a submodular transformation $F\colon\set{0,1}^V \to \bbR^E$ as in Example~\ref{ex:undirected-conductance}.
  Then, Theorem~\ref{the:intro-cheeger} reduces to the Cheeger inequality for undirected graphs (with a slightly worse coefficient in the right inequality, that is, $2$ instead of $\sqrt{2}$).
\end{example}

\begin{example}[Directed graphs]\label{ex:directed-graph}
  Let $G=(V,E)$ be a directed graph.
  Then, we define a submodular transformation $F\colon\set{0,1}^V \to \bbR^E$ so that, for each arc $e \in E$, $F_e\colon\set{0,1}^V \to \bbR$ is the cut function of the directed graph with a single arc $e$.
%   $f_e:\bbR^V \to \bbR$ for an arc $e = (u,v) \in E$ is defined as $f_e(\bmx) = \max\set{\bmx(u)-\bmx(v),0}$, which is the \Lovasz extension of the
  Then, $\bmd_F(v)$ for a vertex $v \in V$ is the number of arcs to which $v$ is incident as a head or a tail, and $\cut_F(S)$ for a vertex set $S \subseteq V$ is the number of arcs leaving $S$ and entering $V \setminus S$.
  Then, the Cheeger inequality derived from Theorem~\ref{the:intro-cheeger} coincides with that in~\cite{Yoshida:2016ig}.
\end{example}

\begin{example}[Hypergraphs]\label{ex:hypergraph}
  Let $G=(V,E)$ be a hypergraph.
  Then, we define a submodular transformation $F\colon\set{0,1}^V \to \bbR^E$ so that, for each hyperedge $e \in E$, $F_e\colon\set{0,1}^V \to \bbR$ is the cut function of the hypergraph with a single hyperedge $e$.
  % $f_e:\bbR^V \to \bbR$ for a hyperedge $e \in E$ is defined as $f_e(\bmx) = \max_{u,v \in e}|\bmx(u)-\bmx(v)|$, which is the \Lovasz extension of the cut function of the hypergraph with a single hyperedge $e$.
  Then, $\bmd_F(v)$ for a vertex $v \in V$ is the number of hyperedges incident to $v$, and $\cut_F(S)$ for a vertex set $S \subseteq V$ is the number of hyperedges containing a vertex in $S$ and another vertex in $V \setminus S$.
  Then, the Cheeger inequality derived from Theorem~\ref{the:intro-cheeger} coincides with that in~\cite{Chan:2018eu,Louis:2014tg}.
\end{example}

\begin{example}[Hypergraphs with each hyperedge being associated with a submodular function]\label{ex:map}
  A slightly more general setting than Example~\ref{ex:hypergraph} is that each hyperedge $e \in E$ is associated with a submodular function $F_e:\set{0,1}^e \to \bbR$ and a submodular transformation $F\colon\set{0,1}^V \to \bbR^E$ is defined so that $F(S)(e) = F_e(S \cap e)$ for each $e \in E$.
  Finding a large set $S \subseteq V$ achieving small $\cut_F(S)$ has numerous applications, including image segmentation and denoising via MAP inference on Markov random fields~\cite{Ishikawa:2009jl,Ishikawa:2011gt,Kolmogorov:2004kp}, clustering based on network motifs~\cite{Li:2017up}, and learning ranking data~\cite{Li:2017up}.
  We can formulate such a problem as finding a set $S \subseteq V$ with small $\phi_F$, and we can bound it from below and above by Theorem~\ref{the:intro-cheeger} using $\lambda_F$.
%  some heuristics have been proposed to compute $\phi_F$~\cite{Li:2017up}.
\end{example}
% \begin{example}[Vertex cut]
%   Let $G=(V,E)$ be an undirected graph and $\caM = (E, \caI)$ be the graphic matroid associated with $G$, that is, a set of edges $F \subseteq E$ belongs to $\caI$ if and only if $F$ is a forest.
%   For the rank function $r_\caM:2^E \to \bbR_+$ of $\caM$, we define the connectivity function $\kappa_M:2^E \to \bbR_+$

% \end{example}

Theorem~\ref{the:intro-cheeger} also derives some novel Cheeger inequalities for joint distributions.
\begin{example}[Mutual information]\label{ex:mutual-information}
  Let $V$ be a set of discrete random variables with $|V| = n$.
  Then, it is known that the mutual information $\caI(S;V \setminus S)$ as a function of $S$ satisfies submodularity.
%  From the fact that the random variables are Boolean, $\caI$ is bounded by $n$.
  Now, we define a submodular transformation (or, function) $F\colon\set{0,1}^V \to \bbR$ as $\caI$, divided by $H(V)$ for normalization, where $H(V)$ is the entropy of $V$.
  Then, $\bmd_F(v) = 1$ for $v \in V$, and $\cut_F(S) = \caI(S;V\setminus S)/H(V)$.
  Since $\caI(S;V \setminus S)$ is symmetric, we have $\phi_F = \min_{\emptyset \subsetneq S \subsetneq V}\caI(S;V \setminus S)/\min\set{|S|,n-|S|}H(V)$.
  Intuitively speaking, $\phi_F$ is small when there is a partition of $V$ into large sets $S$ and $V-S$ such that we obtain little information on $V-S$ by observing $S$, and vice versa.
%   used in structure learning~\cite{}.
  The problem of finding a large set $S$ with a small conductance frequently appears in machine learning~\cite{Narasimhan:2004tm,Narasimhan:2005uy} and application domains~\cite{Alhoniemi:2007ec,Cardoso:2003ug,Zhou:2003ip,Zhou:2004ij}, and it can be formalized as computing $\phi_F$.
  We can bound $\phi_F$ from below and above by Theorem~\ref{the:intro-cheeger} using $\lambda_F$.
  The argument can be generalized to multivariate Gaussians by defining the mutual information via differential entropy, which is also submodular.
%  However, algorithms presented there for computing $\phi_F$ are merely heuristics and, to the best of our knowledge, no theoretically guaranteed algorithms have been proposed.
\end{example}

\begin{example}[Directed information]\label{ex:directed-information}
  Let $V$ be a finite set with $|V|=n$ and for each $v \in V$, we consider a sequence $(v_1,\ldots,v_\tau)$ of Boolean random variables, where we regard $v_t$ as the random variable associated with $v \in V$ at time $t \in \set{1,\ldots,\tau}$.
%  To avoid triviality, we assume every $v_t\;(v \in V, t \in \set{1,\ldots,T})$ is dependent on $v'_t$ for some other $v' \in V$.
  Then, for a set $S \subseteq V$ and $t \in \set{1,\ldots,\tau}$, we define $S_t = \set{v_t \mid v \in S}$ as the set of random variables associated with $S$ available at time $t$, and define $S_{\leq t} = \set{S_1,\ldots,S_t}$.
  For two sets $S,T \subseteq V$, the \emph{directed information} from $S$ to $T$, denoted by $\caI(S \to T)$, is defined as $\sum_{t=1}^\tau\caI(S_{\leq t}; T_t \mid T_{\leq t-1})$, which measures the amount of information that flows from $S_{\leq \tau}$ to $T_{\leq \tau}$.
  Directed information has many applications in causality analysis~\cite{Massey:1990vy,Permuter:2011jr,Permuter:2009fg}.
  The directed information $\caI(S \to V \setminus S)$ as a function of $S$ is known to be submodular but is unnecessarily symmetric~\cite{Zhou:2016vc}.

  As in Example~\ref{ex:mutual-information}, we define a submodular transformation (or, function) $F\colon\set{0,1}^V \to \bbR$ as $\caI$, divided by $n\tau$ for normalization.
  Then, we can bound $\phi_F$ from below and above by Theorem~\ref{the:intro-cheeger} using $\lambda_F$.
\end{example}
For Examples~\ref{ex:directed-graph},~\ref{ex:map},~\ref{ex:mutual-information}, and~\ref{ex:directed-information}, although several algorithms for computing $\phi_F$ have been proposed in the literature~\cite{Alhoniemi:2007ec,Cardoso:2003ug,Li:2017up,Yoshida:2016ig,Zhou:2003ip,Zhou:2004ij}, to the best of our knowledge, no theoretically guaranteed algorithms have been known.

The right inequality in Theorem~\ref{the:intro-cheeger} is algorithmic in the following sense:
Given a vector $\bmx \in \bbR^V$ orthogonal to $D_F^{1/2}\bmone$, we can compute in polynomial time a set $\emptyset\subsetneq S\subsetneq V$ such that $\phi_F(S) \leq 2\sqrt{\caR_F(\bmx)}$, where $\caR_F(\bmx)$ is the \emph{Rayleigh quotient} of $\caL_F$ defined as
\[
  \caR_F(\bmx) = \frac{\bigl\langle \bmx, \caL_F(\bmx)\bigr\rangle}{\|\bmx\|_2^2}.
\]
Here, we can show that $\langle \bmx,\bmy\rangle$ has the same value for any $\bmy \in \caL_F(\bmx)$, and hence we denote it by $\langle \bmx,\caL_F(\bmx)\rangle$ by abusing the notation.
We can show that $\lambda_F$ is the minimum of $\caR_F(\bmx)$ subject to $\bmx \neq \bmzero$ and $\bmx$ being orthogonal to the trivial eigenvector, that is, $D_F^{1/2}\bmone$.

\begin{example}
  For a submodular transformation $F\colon\set{0,1}^V \to \bbR^E$ associated with a undirected graph $G=(V,E)$ (see Example~\ref{ex:undirected-conductance}), we have $\bigl\langle \bmx, \caL_F(\bmx)\bigr\rangle = \sum\limits_{\set{u,v}\in E}{\Bigl(\frac{\bmx(u)}{\sqrt{\bmd_F(u)}}-\frac{\bmx(v)}{\sqrt{\bmd_F(v)}}\Bigr)}^2$.
  For a submodular transformation $F\colon\set{0,1}^V \to \bbR^E$ associated with a directed graph $G=(V,E)$ (see Example~\ref{ex:directed-graph}), we have $\bigl\langle \bmx, \caL_F(\bmx)\bigr\rangle = \sum\limits_{(u,v)\in E}\max\Bigset{\frac{\bmx(u)}{\sqrt{\bmd_F(u)}}-\frac{\bmx(v)}{\sqrt{\bmd_F(v)}},0}^2$.
  For a submodular transformation $F\colon\set{0,1}^V \to \bbR$ associated with a hypergraph $G=(V,E)$ (see Example~\ref{ex:hypergraph}), we have $\bigl\langle \bmx, \caL_F(\bmx)\bigr\rangle = \sum\limits_{e\in E}\max\limits_{u,v \in e}{\Bigl(\frac{\bmx(u)}{\sqrt{\bmd_F(u)}}-\frac{\bmx(v)}{\sqrt{\bmd_F(v)}}\Bigr)}^2$.
\end{example}

%In this paper, we show approximation algorithms for $\lambda_F$.
%We say that a submodular transformation $f:\bbR^V \to \bbR^E$ is \emph{symmetric} if $f(\bmx) = f(\bmone-\bmx)$, which is equivalent to having each $F_e:\set{0,1}^V \to \bbR$ symmetric.
As opposed to the matrix case, it is NP-hard to compute $\lambda_F$ under the SSEH\@.
Hence, we consider approximating $\lambda_F$.
First, we provide the following approximation algorithm for  symmetric submodular transformations.
Here, we say that a submodular transformation $F\colon\set{0,1}^V \to \bbR^E$ is \emph{symmetric} if $F(S) = F(V\setminus S)$ for every $S \subseteq V$.
\begin{theorem}\label{the:symmetric-eigenvalue}
  There is an algorithm that, given $\epsilon > 0$ and (a value oracle of) a non-negative symmetric submodular transformation $F\colon\set{0,1}^V \to \bbR^E$ with $F(\emptyset) = \bmzero$, computes a non-zero vector $\bmx \in \bbR^V$ such that $\langle \bmx, D_F^{1/2}\bmone\rangle = 0$ and
  \[
    \lambda_F \leq \caR_F(\bmx) \leq O\Bigl(\frac{\log n}{\epsilon^2}\lambda_F +\epsilon B^2\Bigr),
  \]
  with a probability of at least $9/10$ in ${\poly(nm)}^{\poly(1/\epsilon)}$ time, where $n=|V|$, $m=|E|$, and $B$ is the maximum $\ell_2$-norm of a point in the base polytopes of $F_e$'s.

  If the number of extreme points of the base polytope of each $F_e$ is bounded by $N$, the upper bound and time complexity can be improved to $O(\log N \cdot \lambda_F)$ and $\poly(nmN)$, respectively.
\end{theorem}
The definition of the base polytope is deferred to Section~\ref{sec:pre}.
We do not need the condition $F(V) = \bmzero$ because it follows from $F(\emptyset)=\bmzero$ and the symmetry of $F$.
The left inequality is trivial because $\lambda_F$ is the minimum of $\caR_F(\bmx)$ subject to $\bmx \neq \bmzero$ and $\langle \bmx,D_F^{1/2}\bmone\rangle = 0$.

In general, if a submodular function is relevant to $r$ variables, then the number of extreme points in its base polytope is $r!$.
However, when the submodular transformation $F$ is constructed from an $r$-uniform hypergraph as in Example~\ref{ex:hypergraph}, the number of extreme points can be bounded by $O(r^2)$, and hence we get an approximation ratio of $O(\log r)$.
This approximation matches the one given in~\cite{Chan:2018eu,Louis:2014tg} and is known to be tight under the SSEH~\cite{Chan:2018eu,Louis:2014tg}.
%Combined with Theorem~\ref{the:intro-cheeger}, the approximated eigenvalue $\tilde{\lambda}_F$ satisfies $\Omega(\tilde{\lambda}_F/\log r) \leq \phi_F \leq O( \sqrt{\tilde{\lambda}_F})$.
%We cannot significantly improve this guarantee because it is known that we cannot distinguish the case $\phi_F = \epsilon$ from the case $\phi_F = \Omega(\sqrt{\epsilon \log r /r})$ under the SSEH~\cite{Chan:2018eu,Louis:2014tg}.

For general submodular transformations, we give the following algorithm:
\begin{theorem}\label{the:general-eigenvalue}
  There is an algorithm that, given $\epsilon > 0$ and (a value oracle of) a non-negative submodular transformation $F\colon\set{0,1}^V \to \bbR^E$ with $F(\emptyset) = F(V) = \bmzero$, computes a non-zero vector $\bmx \in \bbR^V$ such that $\langle \bmx, D_F^{1/2}\bmone\rangle = 0$ and
  \[
    \lambda_F \leq \caR_F(\bmx) \leq O\Bigl(\frac{\log n \log (n^{1/\epsilon^2}m)}{\epsilon^2}  \lambda_F+\epsilon B^2\Bigr) = O\Bigl(\Bigl(\frac{\log^2 n}{\epsilon^4}+\frac{\log n \log m}{\epsilon^2} \Bigr)   \lambda_F+\epsilon B^2\Bigr),
  \]
  with a probability of at least $9/10$ in ${\poly(nm)}^{\poly(1/\epsilon)}$ time, where $n=|V|$, $m=|E|$, and $B$ is the maximum $\ell_2$-norm of a point in the base polytopes of $F_e$'s.

  If the number of extreme points of the base polytope of each $F_e$ is bounded by $N$, the upper bound and time complexity can be improved to $O\bigl((\log^2 N + \log m \log N) \cdot \lambda_F\bigr)$ and $\poly(nmN)$, respectively.
\end{theorem}
Again, the left inequality is trivial.
Although the approximation ratio for the general case is slightly worse than that for the symmetric case, it remains polylogarithmic in $n$ and $m$.

We can obtain approximation algorithms for $\phi_F$ for any submodular transformation $F\colon\bbR^V \to \bbR^E$ by combining Theorem~\ref{the:intro-cheeger} and either Theorem~\ref{the:symmetric-eigenvalue} or Theorem~\ref{the:general-eigenvalue}.
Below, we see some representative examples for which no theoretically guaranteed approximation algorithms have been known.

If the submodular transformation $F\colon\set{0,1}^V \to \bbR^E$ is constructed from a directed graph as in Example~\ref{ex:directed-graph}, then the number of extreme points of each base polytope is two.
Noticing $\log m = O(\log n)$, we obtain the following:
\begin{corollary}
  Let $F\colon\set{0,1}^V \to \bbR$ be as in Example~\ref{ex:directed-graph}.
  Then, there exists a polynomial-time algorithm that outputs $\tilde{\lambda}_F$ such that $\Omega\bigl(\tilde{\lambda}_F/\log n\bigr) \leq \phi_F \leq O\bigl(\sqrt{\tilde{\lambda}_F}\bigr)$.
\end{corollary}
%As we have mentioned, no non-trivial polynomial-time approximation algorithm was known for directed graphs.

If the submodular transformation $F\colon\set{0,1}^V \to \bbR^E$ is constructed as in Example~\ref{ex:map}, where the underlying hypergraph is $r$-uniform, then the number of extreme points of each base polytope is $O(r!)$ and $|E| \leq \binom{n}{r}$.
Thus, we obtain the following:
\begin{corollary}
  Let $F\colon\set{0,1}^V \to \bbR$ be as in Example~\ref{ex:map}, where the underlying hypergraph is $r$-uniform for some constant $r$.
  Then, there exists a polynomial-time algorithm that outputs $\tilde{\lambda}_F$ such that $\Omega\bigl(\tilde{\lambda}_F/(r^2\log r \log n)\bigr) \leq \phi_F \leq O\bigl(\sqrt{\tilde{\lambda}_F}\bigr)$.
\end{corollary}

For the mutual information explained in Example~\ref{ex:mutual-information}, we have $B^2 \leq 1/H(V)$, and hence the approximated eigenvalue $\tilde{\lambda}_F$ satisfies $\lambda_F \leq \tilde{\lambda}_F \leq O(\log n /\epsilon^2 \cdot \lambda_F  + \epsilon/H(V))$.
Then, we obtain the following:
\begin{corollary}
  Let $F\colon\set{0,1}^V \to \bbR$  be as in Example~\ref{ex:mutual-information}.
  Then, for any $\epsilon > 0$, there exists a polynomial-time algorithm that outputs $\tilde{\lambda}_F$ such that $\Omega\bigl(\epsilon^2(\tilde{\lambda}_F-\epsilon/H(V))/\log n\bigr) \leq \phi_F \leq O\bigl(\sqrt{\tilde{\lambda}_F}\bigr)$.
\end{corollary}
% Now, we provide concrete bounds on $B$ for
%For the cut functions explained in Examples~\ref{ex:undirected-graph},~\ref{ex:directed-graph}, and~\ref{ex:hypergraph}, we have $B \leq 1$, and hence the approximated eigenvalue $\tilde{\lambda}_F$ satisfies $\lambda_F \leq \tilde{\lambda}_F \leq O(\log n /\epsilon^2 \cdot \lambda_F  + \epsilon)$.
%Then, we have $\Omega(\epsilon^2(\tilde{\lambda}_F-\epsilon)/\log n) \leq \phi_F \leq O(\sqrt{\tilde{\lambda}_F})$ by Theorem~\ref{the:intro-cheeger}.
%Hence, the lower bound is meaningful when $\tilde{\lambda}_F = \Omega(1)$, which always holds when $\phi_F = \Omega(1)$.

%  and hence the approximated eigenvalue $\tilde{\lambda}_f$
% additive error becomes $\epsilon /n$.
% As $\phi_F = O(1/n)$ in this case, we cannot obtain a meaningful upper bound via Theorem~\ref{the:intro-cheeger} in general.
% However, we can get a non-trivial lower bound on $\phi_F$ when $\lambda_F = \Omega(1/n)$.

\begin{remark}
  After a preprint of this work was posted on arXiv, Li and Milenkovic~\cite{Li:2018we} independently proposed and considered Laplacians for symmetric submodular transformations.
  They derived Cheeger's inequality (with a slight difference in the defininition of conductance) and showed that the SDP-based algorithm we use in Theorem~\ref{the:symmetric-eigenvalue} gives $O(\sqrt{r})$-approximation (with no additive error) to $\caR_F(\bmx)$.
\end{remark}

\subsection{Proof sketch}
The proof of our Cheeger inequality for submodular transformations (Theorem~\ref{the:intro-cheeger}) is similar to those of the existing Cheeger inequalities~\cite{Alon:1986gz,Alon:1985jg,Chan:2018eu,Louis:2014tg,Yoshida:2016ig}, although we have to use some specific properties of submodular functions.

In order to prove Theorem~\ref{the:symmetric-eigenvalue} and~\ref{the:general-eigenvalue}, that is, to approximate the smallest non-trivial eigenvalue of the normalized Laplacian of a submodular transformation, we use semidefinite programming (SDP).
To this end, we first rephrase its Rayleigh quotient using \Lovasz extensions.
More specifically, for a submodular transformation $F\colon\set{0,1}^V \to \bbR^E$, the numerator of $\caR_F(\bmx)$ can be written as
\begin{align}
  \bigl\langle \bmx,\caL_F(\bmx) \bigr\rangle = \sum_{e \in E}{f_e(\bmx)}^2= \sum_{e \in E}{\bigl(\max_{\bmw \in B(F_e)}\langle \bmw,\bmx\rangle\bigr)}^2, \label{eq:intro-numerator}
\end{align}
where $f_e\colon\bbR^V \to \bbR$ is the \Lovasz extension of $F_e$.
Now the goal is to minimize this numerator~\eqref{eq:intro-numerator} subject to $\|\bmx\|_2^2 = 1$ and $\langle \bmx, D_F^{1/2}\bmone\rangle=0$.

In the symmetric case, we can show that it is possible to further rephrase the numerator of $\caR_F(\bmx)$ as
\[
  \bigl\langle \bmx,\caL_F(\bmx) \bigr\rangle = \sum_{e \in E}{f_e(\bmx)}^2 = \sum_{e \in E}\max_{\bmw \in B(F_e)}\langle \bmw,\bmx\rangle^2.
\]
A problem here is that $B(F_e)$ is a polytope and we cannot express the maximum over $B(F_e)$ in an SDP\@.
Although it is not difficult to show that we only have to take the maximum over extreme points of $B(F_e)$, the number of extreme points can be $n!$ in general, which is prohibitively large.
(We can bypass this issue when the number of extreme points in each $B(F_e)$ is small.)
To address this issue, we replace $B(F_e)$ with an $\epsilon B$-cover $C_e \subseteq B(F_e)$ (see Theorem~\ref{the:symmetric-eigenvalue} for the definition of $B$), which is a set of points such that for any $\bmw \in B(F_e)$, there exists a point $\bmp \in C_e$ with $\|\bmp - \bmw\|_2 \leq \epsilon B$.
Using the properties of submodular functions, we can show that there is an $\epsilon B$-cover of size roughly $O(n^{1/\epsilon^2})$ (instead of being exponential in $n$), and we can efficiently compute it by exploiting Wolfe's algorithm~\cite{Wolfe:1976dg}, which is useful for judging whether a given point is close to a base polytope.
Then, we can solve the resulting SDP in polynomial time in $n$ and $m$.
The additive error of $\epsilon B$ in Theorem~\ref{the:symmetric-eigenvalue} (and Theorem~\ref{the:general-eigenvalue} as well) occurs when replacing $B(F_e)$ by its $\epsilon B$-cover $C_e$.
Also, we show that the bound of $O(n^{1/\epsilon^2})$ is almost that in Appendix~\ref{apx:covering}.

For each variable $\bmx(v)$ in the Rayleigh quotient, we introduce an SDP variable $\bmx_v \in \bbR^N$ for a large $N \geq n$.
Then after solving the SDP, we round the obtained solution $\set{\bmx_v}_{v \in V}$ using the Gaussian rounding, that is, $\bmx_v \mapsto \langle \bmx_v, \bmg\rangle =: \bmz(v)$, where $\bmg \in \bbR^N$ is sampled from a standard normal distribution $\caN(\bmzero,I_N)$.
Then, we can show that the value of $\sum_{e \in E}{f_e(\bmz)}^2 = \sum_{e \in E}\max_{\bmw \in B(F_e)}\langle \bmw,\bmz\rangle^2$ is roughly equal to $\sum_{e \in E}\max_{\bmw \in C_e}\langle \bmw,\bmz\rangle^2$.
Note that, as each $\bmz(v)\;(v\in V)$ is normally distributed, $\langle \bmw,\bmz\rangle$ for each $\bmw \in C_e$ acts as a normal random variable.
Then, the value  $\sum_{e \in E}\max_{\bmw \in C_e}\langle \bmw,\bmz\rangle^2$ is larger than the SDP value by a factor of $O(\max_{e \in E}\log |C_e|) = O((\log n)/\epsilon^2)$, caused when taking the maximum of $|C_e|$ many squared normal variables for each $e \in E$.
We can also show that the denominator is at least half and the constraint $\langle \bmz,D_F^{1/2}\rangle=0$ is satisfied with high probability, and hence we establish Theorem~\ref{the:symmetric-eigenvalue}.

The general case is more involved as we should stick to the numerator of the form~\eqref{eq:intro-numerator}.
To see the difficulty, suppose that the numerator of the Rayleigh quotient is zero in the SDP relaxation, that is, we obtained an SDP solution $\set{\bmx_v}_{v \in V}$ satisfying $\max_{\bmw \in C_e}\sum_{v \in V}\bmw(v) \langle \bmx_v,\bmv_1\rangle \leq 0$ for every $e \in E$, where $\bmv_1 \in \bbR^N$ is a unit vector representing the value of one.
Here, this value is supposed to represent $f_e(\bmx) = \max_{\bmw \in C_e}\langle \bmw,\bmx\rangle$.
Hence for the vector $\bmz \in \bbR^V$ obtained by rounding $\set{\bmx_v}_{v \in V}$, we expect that $f_e(\bmz) \leq 0$.
However, if we adopt the Gaussian rounding as with the symmetric case, then $\langle \bmw,\bmz\rangle$ for each $\bmw \in C_e$ acts as a normal random variable.
This means that, with a high probability, we have $f(\bmz) = \max_{\bmw \in C_e}\langle \bmw,\bmz\rangle > 0$, and hence the approximation ratio can be arbitrarily large.

The above-mentioned problem is avoided by decomposing $\bmx_v$ as $\langle \bmx_v,\bmv_1\rangle \bmv_1 + P_{\bmv_1^\bot}\bmx$, where $P_{\bmv_1^\bot} \in \bbR^{N \times N}$ is the projection matrix to the subspace orthogonal to $\bmv_1$.
Then, we construct two vectors $\bmz_+ \in \bbR^V$ and $\bmz_- \in \bbR^V$ such that $\bmz_+(v) = \langle\bmx_v,\bmv_1\rangle  + \delta \langle P_{\bmv_1^\bot} \bmx_v, \bmg\rangle$ and $\bmz_-(v) = \langle\bmx_v,\bmv_1\rangle - \delta \langle P_{\bmv_1^\bot} \bmx_v, \bmg\rangle$ for each $v \in V$, where $\delta = 1/\log(mn^{1/\epsilon^2})$ and $\bmg \in \bbR^N$ is sampled from the standard normal distribution $\caN(\bmzero,I_N)$.
This rounding procedure places more importance on the direction $\bmv_1$ than on other directions.
Then, with an additional constraint in the SDP, we can show that the
Rayleigh quotient of at least one of them achieves $O(\log n\log(mn^{1/\epsilon^2})/\epsilon^2)$-approximation.

We have mentioned that the smallest non-trivial eigenvalue $\lambda_F \geq 0$ of the normalized Laplacian $\caL_F\colon\bbR^V \to \set{0,1}^{\bbR^V}$ of a submodular transformation $F\colon\set{0,1}^V \to \bbR^E$ is obtained as the minimum of the Rayleigh quotient $\caR_F(\bmx)$ subject to $\bmx \neq \bmzero$ and $\langle \bmx, D_F^{1/2}\bmone \rangle= \bmzero$.
As opposed to symmetric matrices, the relation between the eigenvalues of $\caL_F$ and the Rayleigh quotient $\caR_F$ is not immediate because $\caL_F$ is not a linear transformation.
Indeed, it is not clear whether $\caL_F$ has a non-trivial eigenvalue at all.
To show this, we consider the following diffusion process associated with $\caL_F$: $\frac{\dx}{\dt} \in -\caL_F(\bmx)$, that is, at each moment we move the current vector $\bmx \in \bbR^V$ to a direction chosen from $-\caL_F(\bmx)$.
The idea of using such a diffusion process was already mentioned in~\cite{Chan:2018eu,Louis:2014tg,Yoshida:2016ig}.
We note that we can show that there is a (unique) solution to this difussion process using the theory of diffusion inclusion~\cite{Aubin:2012wi} or the theory of monotone operators and evolution equations~\cite{miyadera1992nonlinear}.
See~\cite{Ikeda:2018tu} for more details.

% The fact that $\caL_F(\bmx)$ is not continuous in $\bmx$ means that $\bmx$ may not proceed beyond a certain point.
% For example, after moving $\bmx$ to $\bmx'$ along the direction $\bmv \in -\caL_F(\bmx)$ for an infinitesimal time, it could be that the direction $-\bmv$ is in $-\caL_F(\bmx')$ and $\bmx'$ returns to $\bmx$ by moving along the direction $-\bmv$ for an infinitesimal time.
% In order to avoid this problem, we need to choose a direction $\bmv \in -\caL_F(\bmx)$ at each moment so that the direction $\bmv$ also exists in $-\caL_F(\bmx')$, where $\bmx'$ is the vector obtained by moving $\bmx$ along the direction $\bmv$ for an infinitesimal time.
% Fortunately, we can show that such a direction always exists by using  Kakutani's fixed point theorem~\cite{Kakutani:1941gf}.
% Then, analyzing the point at which $\bmx$ converges in the diffusion process, we can guarantee that there exists a small non-trivial eigenvalue of $\caL_F$ and it is achieved by the minimum of the Rayleigh quotient $\caR_F(\bmx)$ subject to $\bmx \neq \bmzero$ and $\langle \bmx, D_F^{1/2}\bmone\rangle = \bmzero$.

% \begin{remark}
%   The diffusion process $\frac{\dx}{\dt} \in -\caL_F(\bmx)$ can be cast as a differential inclusion and the conditions for the existence and the uniqueness of its solution is studied intensively (See, e.g.,~\cite{Aubin:2012wi} for a book).
%   The solution concept we pursue here is called absolute continuity and several sufficient conditions for it is known.
%   However, they
%   are known for absolute continuity, it is not applicable to our case.

% \end{remark}

\subsection{Discussions}
For undirected graphs, several extensions of the Cheeger inequality have been proposed.
For a graph $G=(V,E)$, the \emph{order-$k$ conductance} of $k$ disjoint vertex sets $S_1,\ldots,S_k\subseteq V$ is defined as their maximum conductance, and the \emph{order-$k$ conductance} of a graph is the minimum order-$k$ conductance of $k$ disjoint vertex sets taken from the graph.
Then, the higher order Cheeger inequality~\cite{Lee:2014hi,Louis:2012iva} bounds the order-$k$ conductance of a graph from below and above by the $k$-th smallest eigenvalue of its normalized Laplacian.
The standard conductance is also analyzed using the $k$-th smallest eigenvalue~\cite{Kwok:2017ga,Kwok:2013jd}.
In~\cite{Trevisan:2009jy}, it is argued that the largest eigenvalue of a normalized Laplacian can be used to bound from below and above the bipartiteness ratio, which measures the extent to which the graph is approximated by a bipartite graph.
Its higher order version is also studied~\cite{Liu:2015kc}.
It would be interesting to generalize these extended Cheeger inequalities for submodular transformations.

We believe that the spectral theory on submodular transformations will have many applications in theory and practice beyond Cheeger inequalities studied here.
For example, in a follow-up work, Fujii~\emph{et~al.}~\cite{Fujii:2018vi} studied solving Laplacian systems of the form $L_F(\bmx) = \bmb$, where $F\colon\bbR^V \to \bbR^E$ is a submodular transformation and $\bmb\in \bbR^V$ is a vector, and showed applications in semi-supervised learning and network science.

We believe that the notion of a submodular transformation will be useful not only for generalizing spectral graph theory but also for analyzing various problems that involve piecewise linear functions.
To see this, we introduce the notion of a \emph{\Lovasz transformation}, which is a function of the form $f\colon\bbR^V \to \bbR^E$ such that $f_e\colon\bmx \mapsto f(\bmx)(e)$ is the \Lovasz extension of some submodular function for each $e \in E$.
\Lovasz transformations are piecewise linear in general, and their compositions can express various functions including deep neural networks (see~Section~\ref{sec:dnn} for more details).
We believe that this connection sheds new light on piecewise lienar functions and the submodularity behind \Lovasz transformations will be useful to analyze piecewise linear functions.

\subsection{Organization}
In Section~\ref{sec:pre}, we review basic properties of submodular functions.
In Section~\ref{sec:Laplacian}, we formally define submodular transformation and its Laplacian, and observe their basic properties.
We prove the Cheeger inequality for submodular transformations in Section~\ref{sec:cheeger}.
We consider the covering number of the base polytope of a submodular function in Section~\ref{sec:covering}.
Then, we provide polynomial-time approximation algorithms for the smallest non-trivial eigenvalue of a normalized submodular Laplacian for the symmetric and general cases in Sections~\ref{sec:symmetric} and~\ref{sec:general}, respectively.
In Section~\ref{sec:eigen}, we show that the (normalized) Laplacian of a submodular transformation has a non-trivial eigenvalue and it can be obtained by minimizing the Rayleigh quotient.

%!TEX root=./main.tex

\section{Preliminaries}\label{sec:pre}
For an integer $n \in \bbN$, we define $[n]$ as the set $\set{1,2,\ldots,n}$.
For a subset $S \subseteq V$, we define $\bmone_S \in \bbR^n$ as the indicator vector of $S$, that is, $\bmone_S(v) = 1$ if $v \in S$ and $\bmone_S(v) = 0$ otherwise.
When $S = V$, we simply write $\bmone$.
For a vector $\bmx \in \bbR^V$ and a subset $S \subseteq V$, we define $\bmx|_S \in \bbR^V$ as the vector such that $\bmx|_S(v) = \bmx(v)$ for every $v \in S$ and $\bmx|_S(v) = 0$ for every $v \in V \setminus S$.
The \emph{support} of a vector $\bmx \in \bbR^n$, denoted by $\supp(\bmx)$, is defined as the set $\set{v \in V \mid \bmx(v) \neq 0}$.
For a polyhedron $P \subseteq \bbR^V$, $\bmp \in \bbR^V$, and $r>0$, we define $\bmp + P = \set{\bmp + \bmx \mid \bmp \in P}$ and $rP = \set{r\bmx \mid \bmx \in P}$.
For a polytope $P$, we define $\|P\|_H = \max_{\bmp \in P}\|\bmp\|_2$ as the maximum $\ell_2$-norm of a point in $P$.
% Hausdorff distance  between $P$ and the set consisting only of the origin with respect to the $\ell_2$-norm, which is equal to $$.

For a set function $F\colon\set{0,1}^V\to \bbR$, we define $\|F\|_\infty = \max_{S \subseteq V}F(S)$.
For a set function $F\colon\set{0,1}^V \to \bbR$, a set $S \subseteq V$, and an element $v \in V \setminus S$, we define $f(v \mid S) $ as the marginal gain $f(S \cup \set{v}) - f(S)$.

% \subsection{Submodular functions}

A function $F\colon\set{0,1}^V \to \bbR$ is referred to as \emph{submodular} if
\[
  f(S) + f(T) \geq f(S \cup T) + f(S \cap T)
\]
for every $S,T \subseteq V$.
We say that a function $F\colon\set{0,1}^V \to \bbR$ is \emph{symmetric} if $F(S) = F(V\setminus S)$ for every $S \subseteq V$.
A submodular function $F\colon\set{0,1}^V\to \bbR$ is referred to as \emph{normalized} if $F(\emptyset) = 0$.
In this work, we only consider normalized submodular functions.

We consider a variable $v \in V$ \emph{relevant} in $F\colon\set{0,1}^V \to \bbR$ if adding (or removing) $v$ from the input set may change the value of $F$, that is, there exists some $S \subseteq V \setminus \set{v}$ such that $F(S) \neq F(S \cup \set{v})$.
We consider $v$ \emph{irrelevant} otherwise.
The \emph{support} of $F$, denoted by $\supp(F)$, is the set of relevant variables of $F$.

% For a vector $\bmx\in \bbR^V$ and a subset $S \subseteq V$, we define $\bmx(S)=\sum_{v \in S}\bmx(v)$.

Let $F\colon\set{0,1}^V \to \bbR$ be a submodular function.
The \emph{submodular polyhedron} $P(F)$ and the \emph{base polytope} $B(F)$ of $F$ are defined as
\[
  P(F) = \Bigl\{\bmx \in \bbR^V \mid  \sum_{v \in S}\bmx(v) \leq F(S)\; \forall S \subseteq V\Bigr\}
  \quad \text{and} \quad
  B(F) = \Bigl\{\bmx \in P(F) \mid \sum_{v \in V}\bmx(v) =F(V)\Bigr\}.
\]
As the name suggests, it is known that the base polytope is bounded (Theorem~3.12 of~\cite{Fujishige:2005vx}).

The \emph{\Lovasz extension} $f\colon\bbR^V \to \bbR$ of a submodular function $F\colon\set{0,1}^V \to \bbR$ is defined as
\[
  f(\bmx) = \max_{\bmw \in B(F)} \langle \bmw, \bmx\rangle.
\]
We note that $f(\bmone_S) = F(S)$ for every $S \subseteq V$ and hence we can uniquely recover a submodular function from its \Lovasz extension.

We define $\partial f(\bmx) = \argmax_{\bmw \in B(F)}\langle \bmw,\bmx\rangle$\footnote{We adopted the notation $\partial f(\bmx)$ because each vector in $\partial f(\bmx)$ is a subgradient of $f$ at $\bmx$~\cite{Fujishige:2005vx}.  However, we do not use this property in the work presented in this paper.} as the set of vectors $\bmw \in B(F)$ that attains $f(\bmx)$.
The following is well known:
\begin{lemma}[Theorem~3.22 of~\cite{Fujishige:2005vx}]\label{lem:extreme-point}
  Let $f\colon\bbR^V \to \bbR$ be the \Lovasz extension of a submodular function.
  Then, every extreme point $\bmw$ of $\partial f(\bmx)$ is obtained as follows:
  Let $v_1,\ldots,v_n$ be an ordering of $V$ with $|V| = n$ such that $\bmx(v_1) \geq \cdots \geq \bmx(v_n)$.
  Then, $\bmw(v_i) = f(v_i\mid \set{v_1,\ldots,v_{i-1}} )$ for every $i \in [n]$.

  In particular, every extreme point of $\partial f(\bmzero) = B(F)$ can be obtained by following this approach by setting $\bmx = \bmzero$.
\end{lemma}
The algorithm for computing $\bmw \in \partial f(\bmx)$ based on the ordering of values in $\bmx$ is known as Edmonds' algorithm in the literature.
By Lemma~\ref{lem:extreme-point}, as long as the ordering of values $\bmx(v)\;(v\in V)$ does not change, we can use the same $\bmw \in \bbR^V$  for computing $f(\bmx)$.
% This implies that $f$ is a piecewise linear function.

%Then, we define the \emph{support} of $f$ as $\supp(f)=\supp(F)$.
%Also, we define $B(f) = B(F)$.

%!TEX root=./main.tex

\section{Submodular Transformations and their Laplacians}\label{sec:Laplacian}
In this section, we introduce the notion of a submodular transformation and its Laplacian and normalized Laplacian.

For a function $F\colon \set{0,1}^V \to \bbR^E$ and $e \in E$, let $F_e\colon \bbR^V \to \bbR$ be the $e$-th component of $F$, that is, $F_e\colon \bmx \mapsto F(\bmx)(e)$.
Then, we define a submodular transformation as follows:
\begin{definition}[Submodular transformation]\label{def:submodular-transformation}
  We say that $F\colon \set{0,1}^V \to \bbR^E$ is a submodular transformation if the function $F_e : \bbR^V \to \bbR$ is a submodular function for every $e \in E$.
\end{definition}
For a submodular transformation $F\colon \set{0,1}^V \to \bbR^E$, we always use the symbols $n$ and $m$ to denote $|V|$ and $|E|$.
We say that a submodular transformation $F\colon \set{0,1}^V \to \bbR^E$ is \emph{symmetric} if $F(S) = F(V \setminus S)$ for every $S \subseteq V$.
The \emph{\Lovasz extension} $f:\bbR^V \to \bbR^E$ of a submodular transformation $F\colon \set{0,1}^V \to \bbR^E$ is such that $f_e:\bmx \mapsto f(\bmx)(e)$ is the \Lovasz extension of $F_e$ for each $e \in E$.
The \Lovasz extensions of submodular transformations are collectively referred to as \emph{\Lovasz transformations}.
For a submodular transformation $F\colon \set{0,1}^V \to \bbR$, we will use symbols $f$ and $f_e\;(e\in E)$ to denote those functions.

%We can verify that $f$ is symmetric if and only if $F_e\colon 2^V \to \bbR$ is symmetric for every $e \in E$, where $F_e$ is the submodular function corresponding to $f_e$.
%As each \Lovasz extension is a piecewise linear function, a submodular transformation is also a piecewise linear function.

In Section~\ref{subsec:definition}, we define the Laplacian of a submodular transformation, which we collectively refer to as a submodular Laplacian, and study its basic spectral properties.
In Section~\ref{subsec:normalized}, we discuss the normalized version of a submodular Laplacian.

\subsection{Submodular Laplacians}\label{subsec:definition}

We define the Laplacian associated with a submodular transformation as follows:
\begin{definition}[Submodular Laplacian]\label{def:submodular-Laplacian}
  Let $F\colon\set{0,1}^V \to \bbR^E$ be a submodular transformation.
%  For each $e \in E$, let $F_e\colon 2^{[n]} \to \bbR$ be the submodular function corresponding to the $j$-th component of $f$.
  Then, the Laplacian $L_F\colon\bbR^V \to \set{0,1}^{\bbR^V}$ of $F$ is defined as
  \[
    L_F(\bmx) = \Bigset{\sum_{e \in E} \bmw_e \langle \bmw_e, \bmx\rangle \mid \bmw_e \in \partial f_e(\bmx)\; (e \in E)} = \Bigset{W W^\top \bmx \mid W \in \prod_{e \in E}\partial f_e(\bmx)},
  \]
  where $f_e$ is the \Lovasz extension of $F_e$ for each $e \in E$.
\end{definition}
We can verify that, for every $\bmz \in L_F(\bmx)$, we have $\langle \bmx, \bmz\rangle = \sum_{e \in E}{f_e(\bmx)}^2$, and hence we write $\bigl\langle \bmx, L_F(\bmx)\bigr\rangle$ to denote $\sum_{e \in E}{f_e(\bmx)}^2$ by abusing the notation.
Let $f:\bbR^V \to \bbR^E$ be the \Lovasz extension of $F$.
Then, we have $f(\bmx) = W^T \bmx$ for any $W \in \prod_{e \in E}\partial f_e(\bmx)$.
Hence, we can symbolically understand $L_F$ as $f^\top f$ because $\langle \bmx, L_f(\bmx)\rangle = \sum_{e\in E}{f_e(\bmx)}^2 = \|f(\bmx)\|_2^2$, and this is the intuition behind the definition of $L_F$.

\begin{example}\label{ex:undirected-Laplacian-detailed}
  For an undirected graph $G=(V,E)$, we define a submodular transformation $F\colon\set{0,1}^V \to \bbR^E$ as in Example~\ref{ex:undirected-conductance}.
  Then for an edge $e = \set{u,v} \in E$, we have $\bmw_e = (-1,1)$ if $\bmx(u) < \bmx(v)$, $\bmw_e=(1,-1)$ if $\bmx(u) > \bmx(v)$, and $\bmw_e$ is of the form $(a,-a)$ for $a \in [-1,1]$ if $\bmx(u)=\bmx(v)$.
  Then, we can verify that $L_F(\bmx) = (D_G - A_G)\bmx = L_G\bmx$, where $L_G \in \bbR^{V \times V}$ is the usual Laplacian of $G$.
\end{example}

% In Section~\ref{subsec:heat-equation}, we will consider a diffusion process using $L_F$ and the process naturally imposes a constraint on the choice of $\bmw_e$.

A pair $(\lambda,\bmx) \in \bbR \times \bbR^V$ is called an \emph{eigenpair} of a submodular Laplacian $L_F\colon\bbR^V \to \set{0,1}^{\bbR^V}$ if $L_F(\bmx) \ni \lambda \bmx$.
Such $\lambda$ and $\bmx$ are called \emph{eigenvalue} and \emph{eigenvector} of $L_F$, respectively.
When a submodular transformation $F\colon\set{0,1}^V \to \bbR^E$ satisfies $F(V) = \bmzero$, its Laplacian satisfies the following elegant spectral properties:
\begin{lemma}\label{lem:trivial-eigenpair}
  Let $F\colon\set{0,1}^V \to \bbR^E$ be a submodular transformation with $F(V)=\bmzero$.
  Then, $(0,\bmone)$ is an eigenpair of $L_F$.
\end{lemma}
\begin{proof}
  We have $L_F(\bmone)  = \bigset{\sum_{e \in E}\bmw_e f_e(\bmone) \mid \bmw_e \in \partial f_e(\bmone)\;(e \in E)} = \set{\bmzero} \ni 0 \cdot \bmone$.
\end{proof}

\begin{lemma}\label{lem:positive-semidefiniteness}
  Let $F\colon\set{0,1}^V \to \bbR^E$ be a submodular transformation.
  Then, $L_F$ is positive-semidefinite, that is, all the eigenvalues of $L_F$ are non-negative.
\end{lemma}
\begin{proof}
  Let $(\lambda,\bmx)$ be an eigenpair of $L_F$.
  Then, we have $\bigl\langle \bmx, L_F(\bmx)\bigr\rangle = \lambda \|\bmx\|_2^2$.
  On the other hand, we have $\bigl\langle \bmx, L_F(\bmx)\bigr\rangle = \sum_{e \in E}{f_e(\bmx)}^2 \geq 0$.
  Hence, $\lambda$ should be non-negative.
\end{proof}
From Lemmas~\ref{lem:trivial-eigenpair} and~\ref{lem:positive-semidefiniteness}, the value $0$ is the smallest eigenvalue of $L_F$ with the corresponding eigenvector $\bmone$.
Hence, we call $\bmone$ the \emph{trivial eigenvector} of $L_F$ and call $(0,\bmone)$ the \emph{trivial eigenpair} of $L_F$.

The \emph{Rayleigh quotient} $R_F\colon\bbR^V \to \bbR$ of the Laplacian of a submodular transformation $F\colon\set{0,1}^V \to \bbR^E$ is defined as
\[
  R_F(\bmx) = \frac{\bigl\langle \bmx, L_F(\bmx)\rangle}{\langle \bmx,\bmx\bigr\rangle}
  =\frac{\sum_{e \in E}{f_e(\bmx)}^2}{\|\bmx\|_2^2} = \frac{\|f(\bmx)\|_2^2}{\|\bmx\|_2^2}.
\]
%As we can observe, the value of the Rayleigh quotient does not depend on the choice of the choice of $W \in \prod_{e \in E}\partial f_e(\bmx)$ in Definition~\ref{def:submodular-Laplacian}.

When $L_F$ is a matrix, the minimum of $R_F(\bmx)$ subject to $\bmx \neq \bmzero$ and $\bmx \bot \bmone$ provides the smallest non-trivial eigenvalue and the minimizer is the corresponding eigenvector of $L_F$.
In Section~\ref{sec:eigen}, we show the following relation for general submodular transformations:
\begin{theorem}\label{the:eigenvalues}
  For a submodular transformation $F\colon\set{0,1}^V \to \bbR^E$ with $F(V) = \bmzero$, the Laplacian $L_F$ has a non-trivial eigenpair, that is, there exist $\gamma \in \bbR_+$ and a non-zero vector $\bmz \in \bbR^V$ such that $\bmz \bot \bmone$ and $L_F(\bmz) \ni \gamma \bmz$.
  Furthermore, each such $\gamma$ and $\bmz$ satisfies $\gamma = R_F(\bmz)$.
\end{theorem}

\subsection{Normalized submodular Laplacians}\label{subsec:normalized}
Let $F\colon\set{0,1}^V \to \bbR^E$ be a submodular transformation.
We define the \emph{degree vector} $\bmd_F \in \bbR^V$ of $F$ as $\bmd_F(v) = \#\set{e \in E \mid v \in \supp(F_e)}$.
We say that $\bmd_F(v)$ is the \emph{degree} of $v \in V$ with respect to $F$.
Let $D_F \in \bbR^{V \times V}$ be the diagonal matrix with ${(D_F)}_{vv} = \bmd_F(v)$.
Then, we define the \emph{normalized Laplacian} $\caL_F\colon\bbR^V \to \set{0,1}^{\bbR^V}$ of $f$ as $\caL_F(\bmx) = D_F^{-1/2}L_F(D_F^{-1/2}\bmx)$, or more formally, $\caL_F(\bmx) = \set{D_F^{-1/2} \bmz \mid \bmz \in L_F(D_F^{-1/2}\bmx)}$.
When we consider normalized Laplacians, we always assume that every element of $\bmd_F$ is positive as otherwise we cannot define $D_F^{-1/2}$.

We define an eigenpair/value/vector of the normalized Laplacian of a submodular transformation as with the Laplacian of a submodular transformation.
Then, using the same argument as in Lemmas~\ref{lem:trivial-eigenpair} and~\ref{lem:positive-semidefiniteness}, we can show that, for any  submodular transformation $F\colon\set{0,1}^V\to \bbR^E$ with $F(V)=\bmzero$, its normalized Laplacian $\caL_F$ has an eigenpair $(0,D_F^{1/2}\bmone)$ and that $\caL_F$ is positive-semidefinite.
We call $D_F^{1/2}\bmone$ the \emph{trivial eigenvector} of $\caL_F$ and call $(0,D_F^{1/2}\bmone)$ the \emph{trivial eigenpair} of $\caL_F$.
We define $\caR_F\colon\bbR^V \to \bbR$ as the Rayleigh quotient of the normalized Laplacian of $f$, that is,
\[
  \caR_F(\bmx) = \frac{\bigl\langle \bmx, \caL_F(\bmx)\bigr\rangle}{\langle \bmx,\bmx\rangle}
  =\frac{\sum_{e \in E}{f_e(D_F^{-1/2}\bmx)}^2}{\|\bmx\|_2^2} = \frac{\|{f(D_F^{-1/2}\bmx)}^2\|_2^2}{\|\bmx\|_2^2}.
\]
We have the following, which is a counterpart of Theorem~\ref{the:eigenvalues} for normalized Laplacians.
\begin{theorem}\label{the:normalized-eigenvalues}
  For a submodular transformation $F\colon\set{0,1}^V \to \bbR^E$, the normalized Laplacian $\caL_F$ has a non-trivial eigenvector, that is, there exist $\gamma \in \bbR_+$ and a non-zero vector $\bmz \in \bbR^V$ such that $\bmx \bot D_F^{1/2}\bmone$ and $\caL_F(\bmz) \ni \gamma \bmz$.
  Furthermore, each such $\gamma$ and $\bmz$ satisfies $\gamma = \caR_F(\bmz)$.
\end{theorem}

% \begin{proposition}\label{lem:reflection}
%   Let $F:\set{0,1}^V \to \bbR^E$ be a submodular transformation.
%   Then, $f(-\bmx) = f(\bmone- \bmx) - f(\bmone)$.
% \end{proposition}
% \begin{proof}
%   Fix $e \in E$.
%   Then, we have
%   \[
%     f_e(-\bmx) = \max_{\bmw \in B(f_e)}\langle \bmw, -\bmx\rangle = \max_{\bmw \in -B(f_e)}\langle \bmw, \bmx \rangle = \max_{\bmw \in B(f^\sharp_j)}\langle \bmw, \bmx \rangle = f^\sharp_j(\bmx).
%   \]
%   where $f^\sharp_j:\bbR^V \to \bbR$ is defined as $f^\sharp_j(\bmx) = f_e(\bmone - \bmx) - f(\bmone)$.
% \end{proof}

%!TEX root=./main.tex

\section{Cheeger Inequalities for Submodular Transformations}\label{sec:cheeger}

%Let $f:\bbR^V \to \bbR^E$ be a nonnegative submodular transformation with $f(\bmone)=0$.
%Note that $f$ is nonnegative because $\bmzero \in B(F)$ and $f(\bmx) = \max_{\bmw \in B(F)}\langle \bmw, \bmx\rangle \geq 0$.

%Let $\lambda_F$ be its second-smallest eigenvalue. \ynote{show that it uniquely exists}
In this section, we prove our Cheeger inequality for submodular transformations, that is, Theorem~\ref{the:intro-cheeger}.
We prove the left and right inequalities of Theorem~\ref{the:intro-cheeger} in Sections~\ref{subsec:cheeger-left} and~\ref{subsec:cheeger-right}, respectively.

The following fact is useful in this section.
\begin{proposition}\label{pro:shift-invariant}
  Let $F:\set{0,1}^V \to \bbR^E$ be a submodular transformation with $F(V) = \bmzero$ and let $f:\bbR^V \to \bbR^E$ be its \Lovasz extension.
  Then, we have $f(\bmx + c\bmone) = f(\bmx)$ for any $c \in \bbR^V$.
\end{proposition}
\begin{proof}
  Fix $e \in E$.
  Note that any $\bmw \in B(F_e)$ satisfies $\bmw(V)=0$ because $F_e(V) = 0$.
  Then, we have $f_e(\bmx+ c\bmone) = \max_{\bmw \in B(F_e)}\langle \bmw, \bmx + c \bmone \rangle =  \max_{\bmw \in B(F_e)}\langle \bmw, \bmx \rangle = f_e(\bmx)$.
\end{proof}

\subsection{Lower bound on conductance}\label{subsec:cheeger-left}

\begin{proof}[Proof of the left inequality of Theorem~\ref{the:intro-cheeger}]
  Let $\emptyset \subsetneq S \subsetneq V$ be a subset that achieves $\phi_F = \phi_F(S)$ with $\vol_F(S) \leq \vol_F(V \setminus S)$.
  Let $\bmx \in \bbR^V$ be the vector obtained from $D_F^{1/2}\bmone_S$  by projecting it to the subspace orthogonal to $D_F^{1/2}\bmone$.
  Then, we can write $\bmx = D_F^{1/2}\bmone_S + c D_F^{1/2}\bmone/\|D_F^{1/2}\bmone\|_2$, where
  \[
    c^2 = \frac{\bigl\langle D_F^{1/2}\bmone_S,D_F^{1/2}\bmone\bigr\rangle^2}{\|D_F^{1/2}\bmone\|_2^2}
    = \frac{{\bigl(\sum_{v \in S}\bmd_F(v)\bigr)}^2}{ \sum_{v \in V}\bmd_F(v)}
    = \frac{{\vol_F(S)}^2}{\vol_F(V)}
    \leq \frac{1}{2}\vol_F(S).
  \]
  Then by the Pythagorean theorem, we have
  \[
    \|\bmx\|_2^2 = \|D_F^{1/2}\bmone_S\|_2^2 - c^2 \geq \vol_F(S) - \frac{1}{2}\vol_F(S) = \frac{1}{2}\vol_F(S).
  \]
  Further, we have
  \[
    \bmx^\top \caL(\bmx)
    = \sum_{e \in E}{f_e(D_F^{-1/2}\bmx)}^2
    = \sum_{e \in E}{f_e(\bmone_S + c \bmone/\|D_F^{1/2}\bmone\|)}^2
    = \sum_{e \in E}{f_e(\bmone_S)}^2
    = \sum_{e \in E}{F_e(S)}^2,
  \]
  where we used Proposition~\ref{pro:shift-invariant} in the third equality.

  As $F_e(S) \in [0,1]$ holds for every $e \in E$ and $S \subseteq V$,  we have
  \[
    \lambda_F
    \leq
    \caR_F(\bmx)
    =
    \frac{\bmx^\top \caL(\bmx)}{\|\bmx\|_2^2}
    =
    \frac{2\sum_{e \in E}{F_e(S)}^2}{\vol_F(S)}
    \leq
    \frac{2\sum_{e \in E}F_e(S)}{\vol_F(S)}.
  \]
  Similarly, by considering $-\bmx$, we can show that $\lambda_F \leq 2\sum_{e \in E}F_e(V \setminus S)/\vol_F(S)$, and hence we obtain $\lambda_F \leq 2\phi_F$.
\end{proof}

\subsection{Upper bound on conductance}\label{subsec:cheeger-right}

In this section, we first provide an extension of the rounding known as \emph{sweep rounding}, which is used in the proof of the Cheeger inequality for undirected graphs (Section~\ref{subsubsec:rounding}).
Then, we prove the right inequality of Theorem~\ref{the:intro-cheeger} (Section~\ref{subsubsec:proof-cheeger}).

\subsubsection{Rounding}\label{subsubsec:rounding}

We start with the following equivalent definition of \Lovasz extension:
\begin{lemma}[See, e.g., Definition~3.1 of~\cite{Bach:2013wf}]\label{lem:lovasz-extension-integral}
  Let $F:\set{0,1}^V \to \bbR$ be a submodular function and $f:\bbR^V \to \bbR$ be its \Lovasz extension.
  Then, we have
  \[
    f(\bmx) = \int_0^\infty F\bigl(\set{v \in V \mid \bmx(v) \geq r}\bigr)\dr + \int_{-\infty}^0 \Bigl(F\bigl(\set{v \in V \mid \bmx(v) \geq r}\bigr)- F(V)\Bigr) \dr.
  \]
  % In particular when $F(V)=0$ (or equivalently $f(\bmone)=0$), we have
  % \[
  %   f(\bmx) = \int_{-\infty}^\infty F\bigl(\set{v \in V \mid \bmx(v) \geq r}\bigr)\dr.
  % \]
\end{lemma}

For $\tau \in [0,1]$, we define the threshold function $\thr_\tau: [0,1] \to \set{0,1}$ as $\thr_\tau(x) = 1$ if $x \geq \tau$ and $\thr_\tau(x) = 0$ otherwise.
For a vector $\bmx \in {[0,1]}^V$, we define $\thr_\tau(\bmx) \in \set{0,1}^V$ as the vector obtained from $\bmx$ by applying $\thr_\tau(\cdot)$ coordinate-wise.
% such that $\thr_\tau(\bmx)(v) = \thr_\tau(\bmx(v))$ for each $v \in V$.
Then, we can rephrase $f(\bmx)$ using the threshold function as follows:
\begin{lemma}\label{lem:integral-form}
  Let $F:\set{0,1}^V \to \bbR$ be a submodular function.
  Then, we have
  \[
    f(\bmx) = \int_0^1 f(\thr_\tau(\bmx))\dtau
  \]
  for any $\bmx \in {[0,1]}^V$.
\end{lemma}
\begin{proof}
  By Lemma~\ref{lem:lovasz-extension-integral}, we have
  \begin{align*}
    & f(\bmx)
    = \int_0^\infty F\bigl(\set{v \in V \mid \bmx(v) \geq \tau }\bigr) \dtau+ \int_{-\infty}^0 \Bigl(F\bigl(\set{v \in V \mid \bmx(v) \geq \tau }\bigr) -F(V)\Bigr) \dtau\\
    &   = \int_0^1 F\bigl(\set{v \in V \mid \bmx(v) \geq \tau }\bigr) \dtau = \int_0^1 f(\thr_\tau(\bmx)) \dtau,
    \end{align*}
    where in the last equality, we used the fact that $f(\bmone_S) = F(S)$ for $S \subseteq V$.
\end{proof}

%Given a vector, we can compute a set of vertices with a guarantee on its conductance using the following lemma.
Next, we provide two rounding methods, one for the case $\bmx \in {[0,1]}^V$ and the other for the case $\bmx \in {[-1,0]}^V$.
\begin{lemma}\label{lem:weak-rounding-positive}
  Let $F:\set{0,1}^V \to \bbR^E$ be a submodular transformation and $f:\bbR^V \to \bbR^E$ be its \Lovasz extension.
  For any $\bmx \in {[0,1]}^V$, there exists a set $\emptyset \subsetneq S \subseteq \supp(\bmx)$ such that
  \[
    \frac{\cut_F(S)}{\vol_F(S)} \leq \frac{\sum\limits_{e \in E} f_e(\bmx)   }{\sum\limits_{v \in V} \bmd_F(v) \bmx(v)}.
  \]
  Moreover, we can compute such a set $S$ in $O(n\log n+nm)$ time.
\end{lemma}
\begin{proof}
  By Lemma~\ref{lem:integral-form}, we have
  \[
    \frac{\int_0^1 \sum\limits_{e \in E}f_e\bigl(\thr_\tau(\bmx)\bigr) \dtau}{\int_0^1 \sum\limits_{v \in V} \bmd_F(v) \thr_\tau(\bmx(v)) \dtau}
    =
    \frac{\sum\limits_{e \in E}f_e(\bmx) }{\sum\limits_{v \in V} \bmd_F(v) \bmx(v)}.
  \]
  Therefore, there exists $\tau^* \in [0,1]$ such that
  \[
    \frac{ \sum\limits_{e \in E}f_e(\thr_{\tau^*}(\bmx)) }{ \sum\limits_{v \in V} \bmd_F(v) \thr_{\tau^*}(\bmx(v))}
    \leq
    \frac{\sum\limits_{e \in E}f_e(\bmx) }{\sum\limits_{v \in V} \bmd_F(v) \bmx(v)}.
  \]
  Let $S$ be the support of the vector $\thr_{\tau^*}(\bmx)$.
  Note that we can always choose $S$ to be non-empty.
  Since $\thr_{\tau^*}(\bmx)$ is a $\set{0,1}$-vector, we have $\thr_{\tau^*}(\bmx) = \bmone_S$.
  Then, we have
  \begin{align*}
    \frac{ \sum\limits_{e \in E}f_e(\thr_{\tau^*}(\bmx)) }{ \sum\limits_{v \in V} \bmd_F(v) \thr_{\tau^*}(\bmx(v))}
    & =
    \frac{ \sum\limits_{e \in E}F_e(S)}{ \sum\limits_{v \in S} \bmd_F(v)}
    = \frac{\cut_F(S)}{\vol_F(S)}.
  \end{align*}
  Therefore, we have
  \[
    \frac{\cut_F(S)}{\vol_F(S)} \leq \frac{\sum\limits_{e \in E}f_e(\bmx) }{\sum\limits_{v \in V} \bmd_F(v) \bmx(v)}
    \quad
    \text{and}
    \quad
    \emptyset \subsetneq S \subseteq \supp(\bmx).
  \]

  We can find this set $S$ as follows.
  First, let $v_1,\ldots,v_n$ be the ordering of $V$ such that $\bmx(v_1) \geq \cdots \geq \bmx(v_n)$.
  Then, we consider sets of the form $\{v_1,\ldots,v_k\}$ for $k \in [n]$ and then return the set with the smallest conductance.
  The running time of this algorithm is $O(n\log n + nm)$.
\end{proof}

\begin{corollary}\label{cor:weak-rounding-negative}
  Let $F:\set{0,1}^V \to \bbR^E$ be a submodular transformation with $F(V) = \bmzero$ and $f:\bbR^V \to \bbR^E$ be its \Lovasz extension.
  For any $\bmx \in {[-1,0]}^V$, there exists a set $\emptyset \subsetneq S \subseteq \supp(\bmx)$ such that
  \[
    \frac{\cut_F(V\setminus S)}{\vol_F(S)} \leq -\frac{\sum\limits_{e \in E} f_e(\bmx)   }{\sum\limits_{v \in V} \bmd_F(v) \bmx(v)}.
  \]
  Moreover, we can compute such a set $S$ in $O(n\log n+nm)$ time.
\end{corollary}
\begin{proof}
  Define a submodular transformation $F':\set{0,1}^V \to \bbR^E$ as $F'(S) = F(V \setminus S)$, and let $f':\bbR^V \to \bbR^E$ be its \Lovasz extension.
%  Note that, for each $e \in E$, $f'_e$ is the \Lovasz extension of a submodular function $F'_e:\set{0,1}^V \to \bbR$ defined as $F'_e(S) = F_e(V \setminus S)\;(S \subseteq V)$, and
  Then, we have $f'(\bmz) = f(\bmone - \bmz) = f(-\bmz)$ for any $\bmz \in \bbR^V$ by Proposition~\ref{pro:shift-invariant}.

  We apply Lemma~\ref{lem:weak-rounding-positive} on $F'$ and $-\bmx$.
  Then, we obtain a set $\emptyset \subsetneq S \subseteq \supp(\bmx)$ such that
  \[
    \frac{\cut_{f}(V \setminus S)}{\vol_{f}(S)}
    =
    \frac{\sum\limits_{e \in E} F'_e(S)}{\vol_{f}(S)}
    = \frac{\cut_{f'}(S)}{\vol_{f'}(S)}
    \leq \frac{\sum\limits_{e \in E} f'_e(-\bmx)   }{-\sum\limits_{v \in V} \bmd_{F'}(v)\bmx(v)}
    = \frac{\sum\limits_{e \in E} f_e(\bmx)}{-\sum\limits_{v \in V} \bmd_F(v)\bmx(v)}. \qedhere
  \]

\end{proof}

\subsubsection{Proof of Theorem~\ref{the:intro-cheeger}}\label{subsubsec:proof-cheeger}
We start proving Theorem~\ref{the:intro-cheeger}.
To this end, we need several auxiliary lemmas.
For a vector $\bmx \in \bbR^V$, we define $\bmx_+ \in \bbR^V$ and $\bmx_- \in \bbR^V$ as
\[
  \bmx_+(v) = \begin{cases}
  \bmx(v) & \text{if } \bmx(v) \geq 0,\\
  0 & \text{otherwise},
  \end{cases}
  \quad \text{and} \quad
  \bmx_-(v) = \begin{cases}
  \bmx(v) & \text{if } \bmx(v) \leq 0,\\
  0 & \text{otherwise},
  \end{cases}
\]

\begin{lemma}\label{lem:reverse-triangle-inequality}
  Let $f:\bbR^V \to \bbR$ be the \Lovasz extension of a submodular function $F:\set{0,1}^V \to \bbR$.
  If $f$ is non-negative, then we have
  \[
    {f(\bmx_+)}^2 + {f(\bmx_-)}^2 \leq {f(\bmx)}^2
  \]
\end{lemma}
\begin{proof}
  Recall that
  \[
    f(\bmx) = \max_{\bmw \in B(F)}\langle \bmw,\bmx\rangle
  \]
  Let $\bmw^* \in \bbR^V$ be the maximizer of this maximization problem.
  Then, by Lemma~\ref{lem:extreme-point}, we can calculate $\bmw^*$ as follows:
  First, let $v_1,\ldots,v_n$ be an arbitrary ordering of indices in $V$, such that $\bmx(v_1) \geq \bmx(i_2) \geq \cdots \geq \bmx(v_n)$.
  Now, we obtain $\bmw^*(v_k) = F( v_k \mid \set{v_1,\ldots,v_{k-1}})$ for each $k \in [n]$.

  The value of $f(\bmx_+)$ and $f(\bmx_-)$ can also be determined by the following maximization problems:
  \[
    f(\bmx_+) = \max_{\bmw \in B(F)}\langle \bmw,\bmx_+\rangle
    \quad
    \text{and}
    \quad
    f(\bmx_-) = \max_{\bmw \in B(F)}\langle \bmw,\bmx_-\rangle
  \]
  Let $\bmw_+$ and $\bmw_-$ be the maximizers for $f(\bmx_+)$ and $f(\bmx_-)$, respectively.
  Then, as we can use the same ordering $v_1,\ldots,v_n$ to determine $\bmw_+$ and $\bmw_-$, we can assume $\bmw_+ = \bmw_- = \bmw^*$.
  Now, we have
  \begin{align*}
    {f(\bmx)}^2
    & =
    \langle \bmw^*,\bmx\rangle^2
    =
    \langle \bmw^*,\bmx_++\bmx_-\rangle^2\\
    & =
    \langle \bmw^*,\bmx_+\rangle^2
    +
    \langle \bmw^*,\bmx_-\rangle^2
    + 2 \langle \bmw^*,\bmx_+\rangle \langle \bmw^*,\bmx_-\rangle\\
    & =
    {f(\bmx_+)}^2 + {f(\bmx_-)}^2 + 2f(\bmx_+)f(\bmx_-) \\
    & \geq
    {f(\bmx_+)}^2 + {f(\bmx_-)}^2,
  \end{align*}
  where we used the non-negativity in the inequality.
\end{proof}

The last component we use for proving Theorem~\ref{the:intro-cheeger} is the following equivalent definition of the \Lovasz extension:
\begin{lemma}[See, e.g., Definition~3.1 of~\cite{Bach:2013wf}]\label{lem:lovasz-extension-sum}
  Let $F:\set{0,1}^V \to \bbR$ be a submodular function and $f:\bbR^V \to \bbR$ be its \Lovasz extension.
  For $\bmx \in \bbR^V$, let $v_1,\ldots,v_n$ be an ordering of $V$, such that $\bmx(v_1) \geq \bmx(v_2) \geq \cdots \geq \bmx(v_n)$.
  Let $S_k = \set{v_1,\ldots,v_k}\;(k \in \set{0,\ldots,n})$.
  Then, we have
  \[
    f(\bmx) = \sum_{k \in [n-1]} F(S_k)(\bmx(v_k)-\bmx(v_{k+1})) + F(V)\bmx(v_n).
  \]
  In particular when $F(V)=0$, we have
  \[
    f(\bmx) = \sum_{k \in [n-1]} F(S_k)(\bmx(v_k)-\bmx(v_{k+1})).
  \]
\end{lemma}

\begin{lemma}\label{lem:strong-rounding}
  Let $F:\set{0,1}^V \to \bbR^E$ be a non-negative submodular transformation and $F(V) = \bmzero$, and let $\bmx \in \bbR^V$ be a vector $\langle \bmx, D_F^{1/2}\bmone \rangle = 0$.
  Then, there exists a set $\emptyset \subsetneq S \subsetneq V$ such that
  \[
    \phi_F(S) \leq 2\sqrt{\caR_F(\bmx)}.
  \]
  Moreover, we can find such a set $S$ in $O(n\log n + nm)$ time.
\end{lemma}
\begin{proof}
  Let $\tilde{\bmx} = D_F^{-1/2}\bmx $.
  Note that we have assumed $\bmd_F(v)$ is positive for every $v \in V$.
%  Then, we have $\sum_{v \in V}\bmd_F(v)\tilde{\bmx}(v) = \sum_{v \in V}\sqrt{\bmd_F(v)} \bmx(v) = \langle\bmx,\bmmu_G\rangle \cdot \sqrt{\vol_F(V)} = 0$.
  Then, we have
%  Since $\langle \bmx, \bmmu_G \rangle = 0$, we have
  \begin{align*}
    \caR_F(\bmx)
    & =
    \frac{\sum\limits_{e \in E} {f_e(\tilde{\bmx})}^2 }{\sum\limits_{v \in V} \bmd_F(v) {\tilde{\bmx}(v)}^2 },
  \end{align*}
  where $f:\bbR^V \to \bbR^E$ is the \Lovasz extension of $F$.
  Let $\tilde{\bmy} = \tilde{\bmx} + c \bmone$ for some appropriate $c \in \bbR$ such that $\vol_F(\supp(\tilde{\bmy}_+)) \leq \vol_F(V)/2$ and $ \vol_F(\supp(\tilde{\bmy}_-)) \leq \vol_F(V)/2$ hold.
  Let $\bmy = D_F^{1/2}\tilde{\bmy}$.
  Then, as $f_e(\tilde{\bmx}) = f_e(\tilde{\bmy})$ by Proposition~\ref{pro:shift-invariant} and $\|D_F^{1/2}\tilde{\bmy}\|_2 \geq \|D_F^{1/2}\tilde{\bmx}\|_2$ by the Pythagorean theorem, we have
  \begin{align}
    & \caR_F(\bmx) = \frac{\sum\limits_{e \in E} {f_e(\tilde{\bmx})}^2 }{\sum_{v \in V}\bmd_F(v) {\tilde{\bmx}(v)}^2}
    \geq \frac{\sum\limits_{e \in E} {f_e(\tilde{\bmy})}^2 }{ \sum\limits_{v \in V} \bmd_F(v) {\tilde{\bmy}(v)}^2}
    \geq
    \frac{\sum\limits_{e \in E} {f_e(\tilde{\bmy}_+)}^2 + \sum\limits_{e \in E} {f_e(\tilde{\bmy}_-)}^2  }{ \sum\limits_{v \in V} \bmd_F(v) {\tilde{\bmy}_+(v)}^2 + \sum\limits_{v \in V} \bmd_F(v) {\tilde{\bmy}_-(v)}^2} \tag{By Lemma~\ref{lem:reverse-triangle-inequality}} \nonumber \\
    \geq &
    \min\left\{
    \frac{\sum\limits_{e \in E} {f_e(\tilde{\bmy}_+)}^2}{ \sum\limits_{v \in V} \bmd_F(v) {\tilde{\bmy}_+(v)}^2},
    \frac{\sum\limits_{e \in E} {f_e(\tilde{\bmy}_-)}^2}{ \sum\limits_{v \in V} \bmd_F(v) {\tilde{\bmy}_-(v)}^2}
    \right\}. \label{eq:cheeger-two-cases}
  \end{align}

  Suppose the term for $\tilde{\bmy}_+$ achieves the minimum in~\eqref{eq:cheeger-two-cases}.
  Let $\tilde{\bmy}_+^2 \in \bbR^V$ be the vector defined as $\tilde{\bmy}_+^2(v) = {\tilde{\bmy}_+(v)}^2$ for each $v \in V$.
  Let $v_1,\ldots,v_n$ be the ordering of $V$, such that $\tilde{\bmy}_+^2(v_1) \geq \cdots \geq \tilde{\bmy}_+^2(v_n)$.
  For each $e \in E$, we take the subsequence $v_{e,1},\ldots,v_{e,n_e}$ of this ordering consisting of elements relevant to $F_e$, preserving the order.
  Note that $\supp(F_e) = \set{v_{e,1},\ldots,v_{e,n_e}}$ and that the ordering $v_{e,1},\ldots,v_{e,n_e}$ can be used to compute $f_e(\tilde{\bmy}_+^2)$ as well as $f_e(\tilde{\bmy}_+)$.
  As $F_e(V) = F_e(\set{v_{e,1},\ldots,v_{e,n_e}}) = 0$ for every $e \in E$, we have
  \begin{align}
    & \sum_{e \in E} f_e(\tilde{\bmy}_+^2)
    = \sum_{e \in E}\sum_{k \in [n_e-1]} F_e(\set{v_{e,1},\ldots,v_{e,k}})({\tilde{\bmy}(v_{e,k})}^2-\tilde{\bmy}{(v_{e,k+1})}^2) \tag{By Lemma~\ref{lem:lovasz-extension-sum}} \\
    & = \sum_{e \in E}\sum_{k\in [n_e-1]} F_e(\set{v_{e,1},\ldots,v_{e,k}})\bigl(\tilde{\bmy}_+(v_{e,k})-\tilde{\bmy}_+(v_{e,k+1})\bigr)\tilde{\bmy}_+(v_{e,k}) \nonumber \\
    & \quad + \sum_{e \in E}\sum_{k\in [n_e-1]} F_e(\set{v_{e,1},\ldots,v_{e,k}})\bigl(\tilde{\bmy}_+(v_{e,k})-\tilde{\bmy}_+(v_{e,k+1})\bigr)\tilde{\bmy}_+(v_{e,k+1}).
    \label{eq:strong-rounding-1}
  \end{align}
  We now analyze the first term.
  \begin{align}
    & \sum_{e \in E}\sum_{k\in [n_e-1]} F_e(\set{v_{e,1},\ldots,v_{e,k}})\bigl(\tilde{\bmy}_+(v_{e,k})-\tilde{\bmy}_+(v_{e,k+1})\bigr)\tilde{\bmy}_+(v_{e,k}) \nonumber \\
    & \leq
    \sqrt{\sum_{e \in E}\sum_{k\in [n_e-1]} {F_e(\set{v_{e,1},\ldots,v_{e,k}})}^2{(\tilde{\bmy}_+(v_{e,k})-\tilde{\bmy}_+(v_{e,k+1}))}^2}  \sqrt{\sum_{e \in E}\sum_{k\in [n_e-1]}  {\tilde{\bmy}_+(v_{e,k})}^2} \tag{By Cauchy-Schwarz} \\
    & \leq
    \sqrt{\sum_{e \in E}{\Bigl(\sum_{k\in [n_e-1]} F_e(\set{v_{e,1},\ldots,v_{e,k}}) \bigl(\tilde{\bmy}_+(v_{e,k})-\tilde{\bmy}_+(v_{e,k+1})\bigr) \Bigr)}^2}  \sqrt{\sum_{e \in E}\sum_{k\in [n_e-1]}  {\tilde{\bmy}_+(v_{e,k})}^2} \nonumber \\
    & \leq \sqrt{\sum_{e \in E}{f_e(\tilde{\bmy}_+)}^2}  \sqrt{\sum_{v \in V}\bmd_F(v) {\tilde{\bmy}_+(v)}^2}. \label{eq:strong-rounding-2}
  \end{align}
  In the second inequality, we used the fact that $F_e$ is non-negative for every $e \in E$.

  Similarly, we have
  \begin{align}
    \sum_{e \in E}\sum_{k\in [n_e-1]} F_e(\set{v_{e,1},\ldots,v_{e,k}})(\tilde{\bmy}_+(v_{e,k})-\tilde{\bmy}_+(v_{e,k+1}))\tilde{\bmy}_+(v_{e,k+1})
    \leq
    \sqrt{\sum_{e \in E}{f_e(\tilde{\bmy}_+)}^2} \sqrt{\sum_{v \in V}\bmd_F(v) {\tilde{\bmy}_+(v)}^2}. \label{eq:strong-rounding-3}
  \end{align}
  Combining~\eqref{eq:strong-rounding-1},~\eqref{eq:strong-rounding-2},~\eqref{eq:strong-rounding-3}, for $\bmy_+ = D_F^{1/2}\tilde{\bmy}_+$ we have
  \begin{align*}
    \frac{\sum\limits_{e \in E} f_e(\tilde{\bmy}_+^2)}{\sum\limits_{v \in V}\bmd_F(v){\tilde{\bmy}_+(v)}^2 }
    \leq 2\sqrt{\frac{\sum\limits_{e \in E} {f_e(\tilde{\bmy}_+)}^2}{\sum\limits_{v \in V}\bmd_F(v){\tilde{\bmy}_+(v)}^2}} \leq 2\sqrt{\caR_F(\bmy_+)} \leq 2\sqrt{\caR_F(\bmx)}.
  \end{align*}
  Now, we apply Lemma~\ref{lem:weak-rounding-positive} on $\tilde{\bmy}_+^2$.
  Then, we obtain a set $\emptyset \subsetneq S \subseteq \supp(\tilde{\bmy}_+^2)$ with $\vol_F(S) \leq \vol_F(\supp(\tilde{\bmy}_+^2)) = \vol_F(\supp(\tilde{\bmy}_+)) \leq \vol_F(V)/2$.
  Moreover, we have $\cut_F(S)/\vol_F(S) \leq 2\sqrt{\caR_F(\bmx)}$, which means $\phi_F(S) \leq 2\sqrt{\caR_F(\bmx)}$.

  Now, we consider the case that the term for $\tilde{\bmy}_-$ achieves the minimum in~\eqref{eq:cheeger-two-cases}.
  This time, we define $\tilde{\bmy}_-^2 \in \bbR^V$ as the vector such that $\tilde{\bmy}_-^2(v) = {\tilde{\bmy}_-(v)}^2$ for each $v \in V$.
  By an argument similar to the previous case, we can show that
  \begin{align*}
    \frac{\sum\limits_{e \in E} f_e(-\tilde{\bmy}_-^2)}{\sum\limits_{v \in V}\bmd_F(v){\tilde{\bmy}_-(v)}^2 }
    \leq 2\sqrt{\frac{\sum\limits_{e \in E} {f_e(\tilde{\bmy}_-)}^2}{\sum\limits_{v \in V}\bmd_F(v){\tilde{\bmy}_-(v)}^2}} \leq 2\sqrt{\caR_F(\bmy_-)} \leq 2\sqrt{\caR_F(\bmx)}.
  \end{align*}
  Here, we apply Corollary~\ref{cor:weak-rounding-negative} on $-\tilde{\bmy}_-^2$.
  Then, we obtain a set $\emptyset \subsetneq S \subseteq \supp(\tilde{\bmy}_-^2)$ with $\vol_F(S) \leq \vol_F(\supp(\tilde{\bmy}_-^2)) = \vol_F(\supp(\tilde{\bmy}_-)) \leq \vol_F(V)/2$.
  Moreover,  we have $\cut_F(V \setminus S)/\vol_F(S) \leq 2\sqrt{\caR_F(\bmx)}$, which means $\phi_F(S) \leq 2\sqrt{\caR_F(\bmx)}$.

  In both cases, we have $\phi_F(S) \leq 2\sqrt{\caR_F(\bmx)}$.
\end{proof}

\begin{proof}[Proof of the right inequality of Theorem~\ref{the:intro-cheeger}]
%  For each $e \in E$, we have the following.
%  From the assumption that $F(S) \in [0,1]^E$ for $S \subseteq V$, the submodular function $F_e$ is non-negative.
  For each $e \in E$, as $F_e(V) = 0$, we have $\bmzero \in B(F_e)$.
  It follows that $f_e$ is non-negative because $f_e(\bmx) = \max_{\bmw \in B(F_e)}\langle \bmw,\bmx\rangle \geq \langle \bmzero, \bmx\rangle = 0$.

  Now, we obtain $\phi_F \leq 2\sqrt{\caR_F(\bmx)}$ by invoking Lemma~\ref{lem:strong-rounding} with the eigenvector $\bmx \in \bbR^V$ corresponding to $\lambda_F$.
  The theorem follows because $\caR_F(\bmx) = \lambda_F$ by Theorem~\ref{the:normalized-eigenvalues}.
\end{proof}
%!TEX root=./main.tex

\section{Covering Number of Base Polytopes}\label{sec:covering}

For a set $S \subseteq \bbR^V$ and $\epsilon > 0$, we say that a set of points $C$ in $S$ is an \emph{$\epsilon$-cover} of $S$ if, for any $\bmx \in S$, there exists a point $\bmp \in C$ with $\|\bmx-\bmp\|_2 \leq \epsilon$.
The \emph{$\epsilon$-covering number} of $S$, denoted by $N(\epsilon,S)$, is the smallest size of an $\epsilon$-cover of $S$.
In this section, we show that the $\epsilon \|B(F)\|_H$-covering number of the base polytope $B(F)$ of a submodular function $F\colon\set{0,1}^V \to \bbR$ is small and provides an efficient method to construct such a cover.

The following lemma states that the base polytope of a submodular function is contained in a small $\ell_1$-ball.
\begin{lemma}\label{lem:l1-radius-of-base-polytope}
  Let $F\colon\set{0,1}^V \to \bbR_+$ be a non-negative submodular function.
  Then, we have
  \[
    \max_{\bmw \in B(F)}\|\bmw\|_1 \leq 2\|F\|_\infty.
  \]
\end{lemma}
\begin{proof}
  As $B(F)$ is a convex polytope, the maximum $\ell_1$-norm of a point in $B(F)$ is attained at an extreme point $\bmw^*$ of $B(F)$.
%  Let $\bmw^* \in B(F)$ be a extreme point of $B(F)$ that attains the maximum $\ell_1$-norm.
  By Lemma~\ref{lem:extreme-point}, there exists an ordering $v_1,\ldots,v_n$ of $V$, such that $\bmw^*(v_k) =  F(v_k\mid S_{k-1})\;(k\in [n])$, where $S_k = \set{v_1,\ldots,v_k}$.
%  We also note that $\sum_{i \in [n]} \bmw^*(i) = 0$ holds by $F([n])=0$.

  We now lower bound $\|F\|_\infty$ by using $\|\bmw^*\|_1$.
  Let $v^+_1,\ldots,v^+_{n^+}$ be the sequence obtained from the ordering $v_1,\ldots,v_n$ by extracting $v_k$'s such that $\bmw^*(v_k) > 0$, preserving the order, and let $S^+_k = \set{v^+_1,\ldots,v^+_k}\;(k \in [n^+])$.
  Then based on the submodularity, for any $k \in [n^+]$, we have $F(v^+_k \mid S^+_{k-1}) \geq F(v^+_k\mid S_{k'-1}) = \bmw^*(v^+_k)$, where $k' \geq k$ is such that $v^+_k = v_{k'}$.
  This means that
  \[
    F(S^+_{n^+}) = \sum_{k \in [n^+]} F(v^+_k \mid S^+_{k-1}) \geq \sum_{k \in [n^+]}\bmw^*(v^+_k) \geq \frac{1}{2}\|\bmw^*\|_1,
  \]
  where we used the fact that $\bmw^*(V)=f(V) \geq 0$ in the last inequality.
  Then, we have $\|\bmw^*\|_1 \leq 2\|F\|_\infty$.
\end{proof}
The above lemma suggests that, when $\|F\|_\infty\leq 1/2$, the base polytope is contained in the $\ell_1$-ball $B_1^V := \set{\bmx \in \bbR^V \mid \|\bmx\|_1 \leq 1}$.
The following covering number of $B_1^V$ is known to be obtained by using Maurey's empirical method (see, e.g.,~\cite{Pisier:1999ux})
\begin{lemma}\label{lem:covering-number-of-l1-ball}
  For every $\epsilon >0$, we have
  \[
    N(\epsilon, B_1^V) \leq U_{\ref{lem:covering-number-of-l1-ball}}(\epsilon, B_1^V):={\Bigl(1+2\epsilon^2 n\Bigr)}^{1/\epsilon^2},
    %\min\Bigl\{\Bigl(1+\frac{1}{2\epsilon^2 n}\Bigr)^{2n}, \Bigl(1+2\epsilon^2 n\Bigr)^{1/\epsilon^2}\Bigr\}.
  \]
  where $n = |V|$.
  Moreover, we can compute an $\epsilon$-cover of $B_1^V$ of size $U_{\ref{lem:covering-number-of-l1-ball}}(\epsilon, B_1^V)$ in $O(nU_{\ref{lem:covering-number-of-l1-ball}}(\epsilon, B_1^V))$ time.
\end{lemma}
This lemma states that the $\epsilon$-covering number of the $B_1^V$ is polynomial in $n$ (as long as $\epsilon$ is constant), which will be crucial when bounding the time complexity and the approximation ratio of our algorithms for approximating eigenvalues in Sections~\ref{sec:symmetric} and~\ref{sec:general}.
In contrast, the $\epsilon$-covering number of the $\ell_2$-ball $B_2^V := \set{\bmx \in \bbR^V \mid \|\bmx\|_2 \leq 1}$ is exponential in $n$ (see, e.g.,~\cite{Pisier:1999ux}).

Lemmas~\ref{lem:l1-radius-of-base-polytope} and~\ref{lem:covering-number-of-l1-ball} implies that we can compute a polynomial-size set $P$ of points in $\bbR^V$ such that any point in the base polytope $B(F)$ of a submodular function $F\colon\set{0,1}^V \to \bbR$ has a close point in $P$.
Obtaining an $\epsilon$-cover of $B(F)$ from $P$ requires us to eliminate the points outside of $B(F)$.
To this end, we use Wolfe's algorithm~\cite{Wolfe:1976dg}, which computes the minimum $\ell_2$-norm point in a polytope.
%For $r \geq 0$, we define $rB_1^V = \set{r\bmx \mid \bmx\in B_1^V} $ as the $\ell_1$-ball of radius $r$, and for $\bmp\in \bbR^V$ and $r \geq 0$, we define $\bmp + rB_1^V = \set{\bmp + \bmx \in \bbR^V\mid \|\bmx\|_1 \leq r}$ as the $\ell_1$-ball of radius $r$ centered at $\bmp$.
The following theoretical guarantee is known for Wolfe's algorithm:
\begin{lemma}[\cite{Chakrabarty:2014uy}]\label{lem:Chakrabarty}
  Let $F\colon\set{0,1}^V \to \bbR$ be a submodular function, and let $\bmp \in \bbR^V$ and $r > 0$.
  Wolfe's algorithm computes a point $\bmw^* \in B(F) \cap (\bmp + rB_1^V)$ such that $\|\bmw^*\|_2^2 \leq \min_{\bmw \in B(F) \cap (\bmp + rB_1^V)}\|\bmw\|_2^2 + 2\epsilon^2$ in $O\bigl(n^4\|B(F) \cap (\bmp + rB_1^V)\|_H^2/\epsilon^2\bigr)$ time, where $n = |V|$.
\end{lemma}
We remark that~\cite{Chakrabarty:2014uy} considers the case that the given polytope is $B(F)$ instead of $B(F) \cap (\bmp + rB_1^V)$.
However, their argument relies only on the fact that the given polytope is convex and we can solve a linear programming over the polytope, which is true for $B(F) \cap (\bmp + rB_1^V)$.

\begin{algorithm}[t!]
  \caption{Construction of an $\epsilon$-cover of the base polytope of a submodular function.}\label{alg:net-of-base-polytope}
  \begin{algorithmic}[1]
  \Require{a submodular function $F\colon\set{0,1}^V \to [0,1]$, $r \geq 0$, and $\epsilon>0$.}
  \State{Construct an $(\epsilon/3)$-cover $P$ of $rB_1^V$, where we identify $V$ with $[n]$.}
  \State{$C \leftarrow \emptyset$.}
  \For{each $\bmp \in P$}
    \State{Define $F_{\bmp}:\set{0,1}^V \to \bbR$ so that $F_\bmp(S) = F(S) - \bmp(S)\;(S \subseteq V)$.}
    \State{Run Wolfe's algorithm on $B(F_\bmp) \cap (-\bmp + rB_1^V)$ and $\epsilon/3$, and let $\bmw_\bmp$ be the returned vector.}
    \If{$\|\bmw_\bmp\|_2 \leq 2\epsilon/3$}
      \State{$C \leftarrow C \cup \set{\bmp + \bmw_\bmp}$.}
    \EndIf{}
  \EndFor{}
  \State{\Return{$C$.}}
  \end{algorithmic}
\end{algorithm}

%Then, by Lemma~\ref{lem:Chakrabarty}, we can compute $\bmw_\bmp \in B(F_\bmp)$ such that $\|\bmw_\bmp\|_2^2 \leq \|\bmw_\bmp^*\|_2^2 + \epsilon^2$ in polynomial time, where $\bmw_\bmp^* = \argmin_{\bmw \in B(F_\bmp)}\|\bmw\|_2^2$.

Now, we show that we can construct a small cover for a base polytope restricted to a small $\ell_1$-ball.
\begin{lemma}\label{lem:covering-base-polytope}
  Let $F\colon\set{0,1}^V \to \bbR$ be a non-negative submodular function.
  For every $\epsilon > 0$, we can construct an $\epsilon$-cover $C$ of $B(F) \cap rB_1^V$ of size $O(U_{\ref{lem:covering-number-of-l1-ball}}(\frac{\epsilon}{3r}, B_1^V))$ in $O(r^2n^4U_{\ref{lem:covering-number-of-l1-ball}}(\frac{\epsilon}{3r}, B_1^V)/\epsilon^2)$ time, where $n = |V|$.
%  \begin{itemize}
%  \item[(i)] For every $\bmw \in B(F)$, there exists a point $\bmp \in P$ such that $\|\bmw - \bmp\|_2 \leq \epsilon$.
%  \item[(ii)] For every $\bmp \in P$, there exists a point $\bmw \in B(F)$ such that $\|\bmw - \bmp\|_2 \leq \epsilon$.
%  \end{itemize}
\end{lemma}
\begin{proof}
  Our algorithm for constructing an $\epsilon$-cover $C$ is summarized in Algorithm~\ref{alg:net-of-base-polytope}.
  It first constructs an $(\epsilon/3)$-cover $P$ of $rB_1^V$.
  Then, for each $\bmp \in P$, we compute a minimum-norm point $\bmw_p$ in $B(F_\bmp) \cap (-\bmp + rB_1^V)$ for $F_\bmp = F-\bmp$ by running Wolfe's algorithm with an error parameter $\epsilon /3$.
  Then, if $\|\bmw_\bmp\|_2$ is sufficiently small, or more specifically, $\|\bmw_\bmp\|_2 \leq 2\epsilon/3$, then we add $\bmp + \bmw_p$ to $C$.
  Note that $\bmp+\bmw_p$ belongs to $B(F) \cap rB_1^V$ as $B(F_\bmp) = \set{\bmw - \bmp \mid \bmw \in B(F)}$.
  Hence, we need to check that any point in $B(F) \cap rB_1^V$ has a close point in the constructed set $C$.

  For every $\bmw \in B(F) \cap rB_1^V$, there exists a point $\bmp \in P$ such that $\|\bmw - \bmp\|_2 \leq \epsilon/3$.
  Then, by Lemma~\ref{lem:Chakrabarty}, we have
  \begin{align*}
    & \|\bmw_\bmp\|_2^2 \leq \argmin_{\bmw' \in B(F_\bmp) \cap (-\bmp+rB_1^V)}\|\bmw'\|_2^2 + 2{\Bigl(\frac{\epsilon}{3}\Bigr)}^2
    =  \argmin_{\bmw' \in B(F) \cap (-\bmp+rB_1^V)}\|\bmw'-\bmp\|_2^2 + \frac{2\epsilon^2}{9} \\
    & \leq \|\bmw-\bmp\|_2^2 + \frac{2\epsilon^2}{9}     \leq \frac{\epsilon^2}{3} \leq {\Bigl(\frac{2\epsilon}{3}\Bigr)}^2.
  \end{align*}
  Hence, $\bmp + \bmw_\bmp \in B(F) \cap rB_1^V$ will be added to $C$.
  Note that
  \[
    \|\bmw- (\bmp + \bmw_\bmp)\|_2
    \leq
    \|\bmw- \bmp\|_2 + \|\bmw_\bmp\|_2
    \leq
    \epsilon,
  \]
  which implies the returned set $C$ is an $\epsilon$-cover of $B(F) \cap rB_1^V$.

  Now, we analyze the time complexity of the algorithm.
  By Lemma~\ref{lem:covering-number-of-l1-ball}, we need $O(n U_{\ref{lem:covering-number-of-l1-ball}}(\epsilon/3, rB_1^V)) = O(n U_{\ref{lem:covering-number-of-l1-ball}}(\frac{\epsilon}{3r}, B_1^V))$ time to compute an $(\epsilon/3)$-cover $P$ of $rB_1^V$.
  For each point $\bmp \in P$, we run Wolfe's algorithm.
  We have
  \begin{align*}
    & \max_{\bmw \in B(F_\bmp) \cap (-\bmp + rB_1^V)}\|\bmw\|_2
    \leq \max_{\bmw \in B(F_\bmp) \cap (-\bmp+rB_1^V)}\|\bmw\|_1
    = \max_{\bmw \in B(F) \cap rB_1^V}\|\bmw - \bmp\|_1 \\
    & \leq \max_{\bmw \in B(F) \cap rB_1^V}\|\bmw\|_1 + \|\bmp\|_1 \leq 2r.
  \end{align*}
  Hence, the running time of Wolfe's algorithm is $O(r^2n^4/\epsilon^2)$ by Lemma~\ref{lem:Chakrabarty}.
  Then, the total running time is $O(r^2n^4 U_{\ref{lem:covering-number-of-l1-ball}}(\frac{\epsilon}{3r},B_1^V) /\epsilon^2)$.
\end{proof}

\begin{theorem}\label{the:covering-base-polytope}
  Let $F\colon\set{0,1}^V \to \bbR$ be a non-negative submodular function.
  For every $\epsilon > 0$, we can construct an $\epsilon\|B(F)\|_H$-cover $C$ of $B(F)$ of size $O( \log n \cdot U_{\ref{lem:covering-number-of-l1-ball}}(\epsilon/6,B_1^V))$ in $O(n^4  \log n \cdot U_{\ref{lem:covering-number-of-l1-ball}}(\epsilon/6,B_1^V)  /\epsilon^2)$ time, where $n = |V|$.
\end{theorem}
\begin{proof}
  Let $K = \max_{v \in V}F(\set{v})$.
  Then, it is easy to check $K \leq \|B(F)\|_H \leq nK$.
  We define $r_i = 2^i K$ for $i \in \set{0,\ldots,L}$, where $L=\lceil \log_2 n\rceil$.
  For each $i \in \set{0,\ldots,L}$, we construct an $\epsilon/2$-cover $C_i$ by invoking Lemma~\ref{lem:covering-base-polytope} on $B(F)/r_i \cap B_1^V$, and then we return the union $C := \bigcup_{i=0}^L r_i C_i$.
  The size of $C$ and the time complexity for constructing $C$ are as claimed.
%  is $O( U_{\ref{lem:covering-number-of-l1-ball}}(\epsilon/6,B_1^V)\log (|V|\|F\|_\infty/\epsilon))$ and the time  is $O(|V|^4  U_{\ref{lem:covering-number-of-l1-ball}}(\epsilon/6,B_1^V)\log (|V|\|F\|_\infty/\epsilon) /\epsilon^2)$

  Now, we show that $C$ is an $\epsilon \|B(F)\|_H$-cover of $B(F)$.
  Let $\bmw \in B(F)$ be an arbitrary vector in the base polytope.
  If $\|\bmw\|_2 \leq r_0$, then there is a point $\bmp \in C_0$ such that $\|\bmw/r_0-\bmp\|_2 \leq \epsilon/2$, which means that $r_0\bmp \in C_0 \subseteq C$ satisfies $\|\bmw - r_0\bmp\|_2 \leq r_0\epsilon/2 \leq \epsilon K \leq \epsilon \|B(F)\|_H$.
  Otherwise, let $i \in \set{1,\ldots,L}$ be such that $r_{i-1} < \|\bmw\|_2 \leq r_i$.
  Such $i$ always exists because $r_0 < \|\bmw \|_2 \leq nK$.
  Then, there exists a point $\bmp \in C_i$ such that $\|\bmw/r_i - \bmp\|_2 \leq \epsilon/2$, which means that $r_i\bmp \in r_i C_i \subseteq C$ satisfies $\|\bmw - r_i\bmp\|_2 \leq \epsilon r_i/2 \leq \epsilon r_{i-1} \leq \epsilon \|\bmw\|_2 \leq \epsilon \|B(F)\|_H$.
\end{proof}

One may think that the idea of approximating base polytopes by $\ell_1$-balls is too naive because base polytopes may have a rich structure derived from submodularity and we only have to cover extreme points instead of the whole base polytope in our applications in Sections~\ref{sec:symmetric} and~\ref{sec:general}.
However, we show in Appendix~\ref{apx:covering} that we cannot significantly improve the bound for some base polytope even if we only have to cover its extreme points.

%!TEX root=./main.tex

\section{Approximating the Smallest Non-trivial Eigenvalue in the Symmetric Case}\label{sec:symmetric}

In this section, we prove Theorem~\ref{the:symmetric-eigenvalue}, that is, we provide a polynomial-time algorithm that approximates the smallest non-trivial eigenvalue of the normalized Laplacian of a symmetric submodular transformation to within a factor of $O(\log n)$ and a small additive error.
We explain our SDP relaxation and rounding method in Section~\ref{subsec:symmetric-relaxation} and then provide an approximation guarantee in Section~\ref{subsec:symmetric-analysis}.

\subsection{SDP relaxation and rounding}\label{subsec:symmetric-relaxation}
Our algorithm is based on SDP relaxation, and our SDP formulation is based on the following simple observation, which exploits the symmetry:
\begin{proposition}\label{pro:symmetric-base-polytope}
  For a symmetric submodular function $F\colon\set{0,1}^V \to \bbR$, we have $B(F) = -B(F)$, that is, $-\bmw \in B(F)$ for every $\bmw \in B(F)$.
\end{proposition}
\begin{proof}
  As $B(F)$ is a convex polytope, it suffices to check whether $-\bmw \in B(F)$ holds for each extreme point $\bmw$ of $B(F)$.

  Let $\bmw \in B(F)$ be an extreme point of $B(F)$.
  By Lemma~\ref{lem:extreme-point}, there exists an ordering $v_1,\ldots,v_n$ of $V$, where $n=n$, such that $\bmw(v_k) = F(v_k \mid \set{v_1,\ldots,v_{k-1}})\;(k \in [n])$.
  Consider the ordering $v'_1,\ldots,v'_n$ of $V$ such that $v'_k = v_{n-k+1}\;(k \in [n])$.
  Again by Lemma~\ref{lem:extreme-point}, the vector $\bmw' \in \bbR^V$ with $\bmw'(v'_k) = F(v'_k \mid  \set{v'_1,\ldots,v'_{k-1}})$ is an extreme point of $B(F)$.
  For every $k \in [n]$, we have
  \begin{align*}
    & \bmw'(v_{n-k+1})=
    \bmw'(v'_k) = F(v'_k \mid \set{v'_1,\ldots,v'_{k-1}}) \\
    & = F(\set{v_n,\ldots,v_{n-k+1}}) - F(\set{v_n,\ldots,v_{n-k+2}}) \\
    & =
    F(\set{v_1,\ldots,v_{n-k}}) - F(\set{v_1,\ldots,v_{n-k+1}}) \tag{By symmetry}\\
    & = -\bmw(v_{n-k+1}).
  \end{align*}
  Hence, we have $-\bmw \in B(F)$.
\end{proof}
Then, we can rephrase ${f(\bmx)}^2$ as follows:
\begin{corollary}\label{cor:symmetric-f-squared}
  For the \Lovasz extension $f\colon\bbR^V \to \bbR$ of a symmetric submodular function $F\colon\set{0,1}^V\to \bbR$, we have
  \[
    {f(\bmx)}^2 = \max_{\bmw \in B(F)}\langle \bmw,\bmx\rangle^2
  \]
  for every $\bmx \in \bbR^V$.
\end{corollary}
\begin{proof}
  Let $\bmw^* \in \argmax_{\bmw \in B(F)} \langle \bmw,\bmx\rangle^2$.
  By Proposition~\ref{pro:symmetric-base-polytope}, we can also assume that $\bmw^* \in \argmax_{\bmw \in B(F)} \langle \bmw,\bmx\rangle$; otherwise, we can replace $\bmw^*$ with $-\bmw^* \in B(F)$ to achieve this.
  Then, we have
  \[
    {f(\bmx)}^2
    = {\Bigl(\max_{\bmw \in B(F)}\langle \bmw,\bmx\rangle\Bigr)}^2
    = \langle \bmw^*,\bmx \rangle^2
    =
    \max_{\bmw \in B(F)}\langle \bmw,\bmx\rangle^2.
    \qedhere
  \]
\end{proof}

By Theorem~\ref{the:normalized-eigenvalues} and Corollary~\ref{cor:symmetric-f-squared}, the smallest non-trivial eigenvalue of the normalized Laplacian $\caL_F$ of a symmetric submodular transformation $F\colon\set{0,1}^V \to \bbR^E$ is the minimum of
\[
  \sum_{e \in E}\max_{\bmw \in B(F_e)}\langle \bmw,D_F^{-1/2}\bmx\rangle^2
\]
subject to $\|\bmx\|_2^2 = 1$ and $\langle \bmx,D_F^{1/2}\bmone\rangle= 0$.
By replacing $\bmx$ with $D_F^{1/2}\bmx$, the minimum can be written as follows:
\begin{align}
  \begin{array}{lll}
  \text{minimize} & \displaystyle \sum_{e \in E} \eta_e^2,\\
  \text{subject to} & \displaystyle \langle \bmw, \bmx\rangle^2 \leq \eta_e^2 & \forall e \in E, \forall \bmw \in B(F_e),\\
%  & \displaystyle \sum_{v \in V}\bmw(v) x_v \leq \eta_e & \forall \bmw \in B(F_e)\\
  & \displaystyle \sum_{v \in V}\bmd_F(v){\bmx(v)}^2 = 1,\\
  & \displaystyle \sum_{v \in V}\bmd_F(v)\bmx(v) = 0.
  \end{array}\label{eq:original-symmetric}
\end{align}
Now, we consider an SDP relaxation of~\eqref{eq:original-symmetric}.
To this end, we introduce vectors $\bmeta_e \in \bbR^N\;(e \in E)$ and $\bmx_v \in \bbR^N\;(v \in V)$ that are supposed to represent $\eta_e\;(e \in E)$ and $x_v \;(v \in V)$, respectively, where $N \geq n$ is a sufficiently large integer.
%We also define a special unit vector $\bmone$ that corresponds to the value $1$.
Then, for a matrix $X = {(\bmx_v)}_{v \in V} \in \bbR^{N \times V}$, our SDP relaxation is the following:
\begin{align}
  \begin{array}{llll}
    \SDP(F) :=  & \text{minimize} & \displaystyle \sum_{e \in E} \|\bmeta_e\|_2^2, \\
    & \text{subject to} & \displaystyle  \|X\bmw \|_2^2 \leq \|\bmeta_e\|_2^2 & \forall e \in E, \bmw \in B(F_e), \\
    & & \displaystyle \sum_{v \in V}\bmd_F(v)\|\bmx_v\|_2^2 = 1,\\
    & & \displaystyle \sum_{v \in V}\bmd_F(v)\bmx_v = 0.
  \end{array}\label{eq:relaxed-symmetric}
\end{align}
The value $\|X\bmw\|_2^2 = \|\sum_{v \in V}\bmw(v)\bmx_v\|_2^2$ is supposed to represent the value $\langle \bmw,\bmx\rangle^2$ in~\eqref{eq:original-symmetric}.

Unfortunately, for each $e \in E$, there are infinitely many choices for $\bmw \in B(F_e)$, and hence we cannot efficiently write down  SDP~\eqref{eq:relaxed-symmetric}.
One observation is that we only have to consider extreme points of $B(F_e)$ because the maximum of $\|X\bmw\|_2^2$ over the base polytope $B(F_e)$ is attained at its extreme point.
However, we are still prevented from efficiently writing down SDP~\eqref{eq:relaxed-symmetric} because the number of extreme points of a base polytope can be $n!$ in general.
(Note that we can bypass this obstacle when the number of extreme points in each $B(F_e)$ is small.)

% We note that, given vectors $\bmx_1,\ldots,\bmx_n$, we can check whether they satisfies the constraint $\sum_{v \in V}\bmw(v) \langle \bmone, \bmx_v\rangle \leq \langle \bmone, \bmeta_e\rangle \;(e \in E, \bmw \in B(F_e))$ in polynomial time.
% Indeed, for each $e \in E$, the vector $\bmw \in B(F_e)$ that maximizes the left hand side can be obtained via Edmonds' greedy algorithm.
% Hence, the SDP~\eqref{eq:relaxed-symmetric-problem} can be solved in polynomial time using the ellipsoid method.
To address the above-mentioned problem, we consider replacing base polytopes $B(F_e)$ by their $\epsilon \|B(F_e)\|_H$-covers, where $\epsilon > 0$ is an error parameter.
For each $e \in E$, let $C_e$ be the $\epsilon \|B(F_e)\|_H$-cover of $B(F_e)$ given in Theorem~\ref{the:covering-base-polytope}.
We consider the following SDP obtained from SDP~\eqref{eq:relaxed-symmetric} by replacing $B(F_e)$ with $C_e$ for each $e \in E$:
\begin{align}
  \begin{array}{llll}
    \SDP_\epsilon(F) := & \text{minimize} & \displaystyle \sum_{e \in E} \|\bmeta_e\|_2^2, \\
    & \text{subject to} & \displaystyle  \|X\bmw\|_2^2 \leq \|\bmeta_e\|^2 & \forall e \in E, \bmw \in C_e, \\
    & & \displaystyle \sum_{v \in V}\bmd_F(v)\|\bmx_v\|_2^2 = 1,\\
    & & \displaystyle \sum_{v \in V}\bmd_F(v)\bmx_v = 0.
  \end{array}\label{eq:eps-covering-symmetric}
\end{align}
As $C_e \subseteq B(F_e)$, it is clear that $\SDP_\epsilon(F) \leq \SDP(F)$, and hence $\SDP_\epsilon(F)$ is a relaxation of~\eqref{eq:original-symmetric}.
Moreover, as the size of $C_e$ is polynomial (as long as $\epsilon$ is constant), we can solve SDP~\eqref{eq:eps-covering-symmetric} in polynomial time.

After solving SDP~\eqref{eq:eps-covering-symmetric}, we sample $\bmg \in \bbR^N$ from the standard normal distribution $\caN(0,I_N)$ and then we round the SDP solution to a vector $\bmz \in \bbR^V$ with $\bmz(v) = \langle \bmx_v,\bmg\rangle\;(v \in V)$.
Our algorithm is summarized in Algorithm~\ref{alg:symmetric-eigenvalue}.

\begin{algorithm}[t!]
  \caption{Approximation of the smallest non-trivial eigenvalue of the normalized Laplacian of a symmetric submodular transformation.}\label{alg:symmetric-eigenvalue}
  \begin{algorithmic}[1]
  \Require{a symmetric submodular transformation $F\colon\set{0,1}^V \to \bbR^E$ and $\epsilon > 0$.}
  \State{Solve SDP~\eqref{eq:eps-covering-symmetric}.}
  %\Comment{Round the obtained solution: }
  \State{Let $\bmg \in \bbR^N$ be a random vector sampled from the standard normal distribution $\caN(0,I_N)$.}
  \State{Define $\bmz\in \bbR^V$ as $\bmz(v) = \langle \bmx_v, \bmg\rangle$ for each $v \in V$.}
  \State{\Return $D_F^{1/2}\bmz$.}
  \end{algorithmic}
\end{algorithm}

\subsection{Analysis}\label{subsec:symmetric-analysis}
Now, we provide an approximation guarantee of Algorithm~\ref{alg:symmetric-eigenvalue}.
The following lemma is useful to analyze the error caused by replacing $B(F_e)$ with $C_e$.
\begin{lemma}\label{lem:error-by-eps-covering-symmetric}
  Let $F\colon\set{0,1}^V \to \bbR$ be a submodular function and let $C \subseteq B(F)$ be an $\epsilon$-cover of $B(F)$ for $\epsilon > 0$.
  Then, for any vector $\bmx \in \bbR^V$, we have
  \[
    \max_{\bmw \in B(F)}\langle \bmw, \bmx \rangle^2
    \leq \max_{\bmw \in C}\langle \bmw, \bmx \rangle^2 + 2\epsilon \|\bmx|_{\supp(F)}\|_2^2 \cdot \|B(F)\|_H.
  \]
\end{lemma}
\begin{proof}
  Let $\bmw^*$ be the maximizer of $\max_{\bmw \in B(F)}\langle \bmw, \bmx \rangle^2$.
  Then, there exists $\bmw' \in C$ such that $\|\bmw^*-\bmw'\|_2 \leq \epsilon$.
  By using the fact that $\bmw(v) = 0$ for every $\bmw \in B(F)$ and $v \in V \setminus \supp(F)$, we have
  \begin{align*}
    & \max_{\bmw \in B(F)}\langle \bmw, \bmx \rangle^2 - \max_{\bmw \in C}\langle \bmw, \bmx \rangle^2
    = \langle \bmw^*, \bmx \rangle^2 - \langle \bmw', \bmx \rangle^2 \\
    & = \langle \bmw^* - \bmw',\bmx\rangle \cdot \langle \bmw^* + \bmw',\bmx\rangle
    = \langle \bmw^* - \bmw',\bmx|_{\supp(F)} \rangle \cdot \langle \bmw^* + \bmw',\bmx|_{\supp(F)}\rangle    \\
    & \leq \epsilon \|\bmx|_{\supp(F)}\|_2 \cdot 2 \|B(F)\|_H \cdot \|\bmx|_{\supp(F)}\|_2 \\
    & = 2\epsilon \|\bmx|_{\supp(F)}\|_2^2 \cdot \|B(F)\|_H. \qedhere
  \end{align*}
\end{proof}

% \begin{theorem}
%   Let $f\colon\bbR^n \to \bbR^m$ be a symmetric submodular transformation \ynote{with $\|F_e\|_\infty \leq 1\;(j\in [m])$.}.
%   Then, the output $\bmz$ of Algorithm~\ref{alg:symmetric-eigenvalue} on $f$ and $\epsilon>0$ satisfies the following with probability at least $1/48$.
%   \begin{itemize}
%     \item[(v)] $\bmz$ is a non-zero vector.
%     \item[(ii)] We have $\sum_{v \in V}\bmz(v) = 0$.
%     \item[(iii)] We have \[
%       \caR_F(\bmz) \leq O\Bigl(\frac{\log (1+\epsilon^2 n)}{\epsilon^2}\lambda + \epsilon\Bigr),
%     \]
%     where $\lambda \geq 0$ is the smallest non-trivial eigenvalue of the Laplacian associated with $f$.
%   \end{itemize}
%   Moreover, the running time is $(1+2\epsilon^2 n)^{\poly(1/\epsilon)}$.
% \end{theorem}
%In what follows, we fix a symmetric submodular transformation $f:\bbR^n \to [0,1]^m$ and $\epsilon > 0$, and let $\bmz \in \bbR^n$ be the output of Algorithm~\ref{alg:symmetric-eigenvalue}.

\begin{lemma}\label{lem:approximation-guarantee-symmetric}
  Let $\bmz \in \bbR^V$ be the output of Algorithm~\ref{alg:symmetric-eigenvalue} on a symmetric submodular transformation $F\colon\set{0,1}^V \to \bbR^E$ and $\epsilon > 0$.
  Then, we have
  \[
     \caR_F(D_F^{1/2}\bmz) = O\Bigl(\frac{\log n}{\epsilon^2}\lambda_F + \epsilon \max_{e \in E}\|B(F_e)\|_H^2 \Bigr),
  \]
  with a probability of at least $1/24$,
  where $\lambda_F \geq 0$ is the smallest non-trivial eigenvalue of $\caL_F$.
\end{lemma}
\begin{proof}
%  We note that $\max_{\bmw \in B(F_e)}\|\bmw\|_2^2 \leq \sum_{v \in V}F_e(\set{v})^2$ because $\bmw(v)\;(v\in V)$ is the marginal gain of $F_e$ when adding $v$ to some subset $S \subseteq V \setminus \set{v}$, which is at most $F_e(v)$.
%  Also, we have $\max_{\bmw \in B(F_e)}\|\bmw\|_2^2 \leq \max_{\bmw \in B(F_e)}\|\bmw\|_1^2 \leq 4\|F_e\|_\infty^2$ by Lemma~\ref{lem:l1-radius-of-base-polytope}.
  For the expected numerator of $\caR_F(D_F^{1/2}\bmz)$, we have
  \begin{align}
    & \E_{\bmz}\Bigl[\sum_{e \in E}{f_e(\bmz)}^2\Bigr]
    = \E_{\bmz}\Bigl[\sum_{e \in E} \max_{\bmw \in B(F_e)}\langle \bmw, \bmz \rangle^2\Bigr] \nonumber \\
    & = \E_{\bmz}\Bigl[\sum_{e \in E}\max_{\bmw \in C_e}\langle \bmw, \bmz \rangle^2 + O\Bigl(\epsilon \sum_{e \in E} \|\bmz|_{\supp(F_e)}\|_2^2 \cdot \|B(F_e)\|_H^2  \Bigr)\Bigr] \tag{By Lemma~\ref{lem:error-by-eps-covering-symmetric}}\\
    & = \E_{\bmz}\Bigl[\sum_{e \in E}\max_{\bmw \in C_e}\langle \bmw, \bmz \rangle^2 + O\Bigl(\epsilon \max_{e \in E}\|B(F_e)\|_H^2 \cdot \sum_{v \in V}\bmd_F(v){\bmz(v)}^2\Bigr)\Bigr].
    \label{eq:numerator-intermediate-symmetric}
  \end{align}
  First, we analyze the first term in the expectation of~\eqref{eq:numerator-intermediate-symmetric}.
  For $\bmw \in C_e$, $\langle \bmw, \bmz \rangle$ is a normal distribution with mean $0$ and variance $\|X\bmw\|_2^2$.
  Hence, Proposition~\ref{pro:fact-8.6}, which bounds the maximum of squared normal random variables, and Theorem~\ref{the:covering-base-polytope} imply that
  \begin{align}
    \E_{\bmz}\Bigl[\sum_{e \in E}\max_{\bmw \in C_e}\langle \bmw, \bmz \rangle^2\Bigr] \leq 4 \log \max_{e \in E}{|C_e|} \cdot \sum_{e \in E}\max_{\bmw \in C_e}\|X\bmw\|_2^2 = \frac{K\log n}{\epsilon^2} \cdot \SDP_\epsilon(F), \label{eq:numerator-symmetric-first}
  \end{align}
  for some constant $K \in \bbR_+$.

  Next, we analyze the second term in the expectation of~\eqref{eq:numerator-intermediate-symmetric}.
  Using the linearity of expectation, we obtain
  \begin{align}
    \E_\bmz\Bigl[ \sum_{v \in V}\bmd_F(v){\bmz(v)}^2 \Bigr] = \sum_{v \in V}\bmd_F(v)\E_\bmz\Bigl[\langle \bmx_v,\bmg\rangle^2 \Bigr]
    = \sum_{v \in V}\bmd_F(v)\|\bmx_v\|_2^2 = 1. \label{eq:numerator-symmetric-second}
  \end{align}
  From~\eqref{eq:numerator-symmetric-first} and~\eqref{eq:numerator-symmetric-second}, by Markov's inequality, we have
  \begin{align}
    \Pr\Bigl[\sum_{e \in E}{f_e(\bmz)}^2 \leq \frac{24K\log n}{\epsilon^2}\cdot \SDP_\epsilon(F) + 24\epsilon \max_{e \in E}\|B(F_e)\|_H^2 \Bigr] \geq 1-\frac{1}{24}.
    \label{eq:numerator-symmetric}
  \end{align}
  Now, we analyze the denominator of $\caR_F(D_F^{1/2}\bmz)$.
  By Proposition~\ref{pro:fact-8.7}, we have
  \begin{align}
    \Pr\Bigl[ \sum_{v \in V}\bmd_F(v){\bmz(v)}^2 \geq \frac{1}{2}\Bigr] \geq \frac{1}{12}. \label{eq:denominator-symmetric}
  \end{align}

  From~\eqref{eq:numerator-symmetric} and~\eqref{eq:denominator-symmetric}, by the union bound, we have
  \[
    \Pr\Bigl[\caR_F(D_F^{1/2}\bmz) \leq \frac{48K\log n}{\epsilon^2}\cdot \SDP_\epsilon(F) + 48\epsilon \max_{e \in E}\|B(F_e)\|_H^2 \Bigr] \geq \frac{1}{24}.
    \qedhere
  \]
\end{proof}

\begin{proof}[Proof of Theorem~\ref{the:symmetric-eigenvalue}]
  Let $\bmz \in \bbR^V$ be the output of Algorithm~\ref{alg:symmetric-eigenvalue} on $F$ and $\epsilon > 0$.
  Because of the constraint $\sum_{v \in V}\bmd_F(v)\bmx_v = \bmzero$, we have $\langle D_F^{1/2}\bmz,D_F^{1/2}\bmone\rangle = \sum_{v \in V}\bmd_F(v)\langle \bmx_v,\bmg\rangle = \langle \bmzero, \bmg\rangle = 0$.
  Hence $\bmz$ is always feasible.
  The approximation guarantee is given by Lemma~\ref{lem:approximation-guarantee-symmetric}.
  The total time complexity is dominated by the time complexity for solving SDP~\eqref{eq:eps-covering-symmetric}, which is ${\poly(nm)}^{\poly(1/\epsilon)}$.

  Note that we can augment the success probability to $9/10$ by running this algorithm a constant number of times and by outputting the vector with the smallest Rayleigh quotient.

  When the number of extreme points of the base polytope of each $F_e$ is bounded by $N$, we can directly solve the optimization problem obtained from~\eqref{eq:relaxed-symmetric} by replacing each $B(F_e)$ with the set of its extreme points in $\poly(nmN)$ time.
  The same analysis goes through, and we get an approximation ratio of $O(\log N)$ because we do not have the second term in~\eqref{eq:numerator-intermediate-symmetric} and the number of points to be considered in~\eqref{eq:numerator-symmetric-first} is $N$ instead of $|C_e|$.
\end{proof}

% By Lemma~\ref{lem:l1-radius-of-base-polytope}, we have $B(F) \subseteq 2B_1^n$.
% Hence, the $\epsilon$-cover $N$ of $2B_1^n$ is an $\epsilon$-cover of $B(F)$.
% However, $N$ may be redundant in the sense it may contain a point $\bmp \in N$ such that $\bmp$ is far from $B(F)$.
% In order to eliminate such a point, we solve the minimum norm point problem on a submodular function $F_\bmp:\set{0,1}^V \to \bbR$ defined so that $F_\bmp(S) = F(S) - \bmp(S)\;(S \subseteq [n])$.
%!TEX root=./main.tex

\section{Approximating the Smallest Non-trivial Eigenvalue in the General Case}\label{sec:general}
In this section, we prove Theorem~\ref{the:general-eigenvalue}, that is, we provide a polynomial-time algorithm that approximates the smallest non-trivial eigenvalue of the normalized Laplacian of a general submodular transformation to within a factor of $O(\log^2n+ \log n\log m)$ and a small additive error.
We explain our SDP relaxation and rounding method in Section~\ref{subsec:general-relaxation} and then provide an approximation guarantee in Section~\ref{subsec:general-analysis}.

For a technical reason, we assume that the input submodular transformation $F\colon\set{0,1}^V \to \bbR^E$ satisfies $F(S) \in {[0,1/100]}^E$ (instead of ${[0,1]}^E$) for $S \subseteq V$.
This can be obtained by dividing the input function by $100 \max_{e\in E}\|F_e\|_\infty$, which preserves the approximation guarantee.

\subsection{SDP relaxation and rounding}\label{subsec:general-relaxation}

Our SDP formulation is based on the following observation:
\begin{proposition}\label{pro:submodular-squared}
  Let $f\colon\bbR^V \to \bbR$ be the \Lovasz extension of a submodular function $F\colon\set{0,1}^V\to \bbR$.
  If $f$ is non-negative, then we have
  \[
    {f(\bmx)}^2 = \frac{1}{2}\max_{\bmw \in B(F)}\Bigl(\langle \bmw,\bmx\rangle^2 +\langle \bmw,\bmx\rangle|\langle \bmw,\bmx\rangle|\Bigr).
  \]
\end{proposition}
\begin{proof}
  We have
  \begin{align*}
    & {f(\bmx)}^2 = {\Bigl(\max_{\bmw \in B(F)}\langle \bmw,\bmx\rangle\Bigr)}^2
    = \max_{\bmw \in B(F)}{\max\bigl\{\langle \bmw,\bmx\rangle,0\bigr\}}^2 \\
    & = \max_{\bmw \in B(F)}{\Bigl(\frac{\langle \bmw,\bmx\rangle+|\langle \bmw,\bmx\rangle|}{2}\Bigr)}^2
    = \frac{1}{2}\max_{\bmw \in B(F)}\Bigl(\langle \bmw,\bmx\rangle^2 +\langle \bmw,\bmx\rangle|\langle \bmw,\bmx\rangle|\Bigr).
  \end{align*}
  Here we have used the non-negativity of $f$ in the second equality.
\end{proof}

By Theorem~\ref{the:normalized-eigenvalues} and Proposition~\ref{pro:submodular-squared}, the smallest non-trivial eigenvalue of the normalized Laplacian $\caL_F$ of a submodular transformation $F\colon\set{0,1}^V \to \bbR^E$ is the minimum of
\[
  \frac{1}{2}\sum_{e \in E} \max_{\bmw \in B(F_e)} \Bigl(\langle \bmw,D_F^{-1/2}\bmx\rangle^2 +\langle \bmw,D_F^{-1/2}\bmx\rangle|\langle \bmw,D_F^{-1/2}\bmx\rangle|\Bigr)
\]
subject to $\|\bmx\|_2^2 = 1$ and $\langle \bmx,D_F^{1/2}\bmone\rangle = 0$.
By replacing $\bmx$ with $D_F^{1/2}\bmx$, the minimum can be written as follows:
\begin{align}
  \begin{array}{lll}
  \text{minimize} & \displaystyle \frac{1}{2}\sum_{e \in E} \eta_e^2,\\
  \text{subject to} & \displaystyle \langle \bmw,\bmx\rangle^2 + \langle \bmw,\bmx\rangle|\langle \bmw,\bmx\rangle| \leq \eta_e^2 & \forall e \in E, \forall \bmw \in B(F_e),\\
  & \displaystyle \sum_{v \in V}\bmd_F(v){\bmx(v)}^2 = 1,\\
  & \displaystyle \sum_{v \in V}\bmd_F(v)\bmx(v) = 0.
  \end{array}\label{eq:original-general}
\end{align}

To derive an SDP relaxation, we introduce vectors $\bmeta_e \in \bbR^N\;(e \in E)$ and $\bmx_v \in \bbR^N\; (v \in V)$ that are supposed to represent $\eta_e\;(e \in E)$ and $\bmx(v) \;(v \in V)$, respectively, where $N \geq n$ is a sufficiently large integer.
In addition, for each $e \in E$ and $\bmw \in B(F_e)$, we introduce vectors $\bmv_{|\langle\bmw,\bmx \rangle|} \in \bbR^N \; (e \in E, \bmw \in B(F_e)) $ that are supposed to represent $|\langle\bmw,\bmx \rangle|$.
Then, for a matrix $X = {(\bmx_v)}_{v \in V} \in \bbR^{N \times V}$, our SDP relaxation is the following:
\begin{align}
  \begin{array}{llll}
    \SDP(f) :=  & \text{minimize} & \displaystyle \frac{1}{2}\sum_{e \in E} \|\bmeta_e\|_2^2, \\
    & \text{subject to} & \displaystyle \|X\bmw\|_2^2 + \langle X\bmw, \bmv_{|\langle\bmw,\bmx \rangle|}\rangle \leq \|\bmeta_e\|_2^2 & \forall e \in E, \bmw \in B(F_e), \\
    & & \|\bmv_{|\langle\bmw,\bmx \rangle|}\|_2^2 = \|X\bmw\|_2^2 & \forall e \in E, \bmw \in B(F_e),\\
    & & \langle \bmv_{|\langle\bmw,\bmx \rangle|}, \bmv_1\rangle \geq \|\bmv_{|\langle\bmw,\bmx \rangle|}\|_2^2 & \forall e \in E, \bmw \in B(F_e),\\
    & & \displaystyle \sum_{v \in V}\bmd_F(v)\|\bmx_v\|_2^2 = 1,\\
    & & \displaystyle \sum_{v \in V}\bmd_v(v)\bmx_v = 0.
  \end{array}\label{eq:relaxed-general}
\end{align}
As in the symmetric case, the value $\|X\bmw\|_2^2 = \|\sum_{v \in V}\bmw(v)\bmx_v\|_2^2$ is supposed to represent the value $\langle \bmw,\bmx\rangle^2$ in~\eqref{eq:original-general}.
The vector $\bmv_1 \in \bbR^N$ is a fixed unit vector that represents the value of one.
The constraint $\|\bmv_{|\langle\bmw,\bmx \rangle|}\|_2^2 = \|X\bmw\|_2^2$ is supposed to represent $|\langle\bmw,\bmx \rangle|^2 = \langle\bmw,\bmx \rangle^2$.
The constraint $\langle \bmv_{|\langle\bmw,\bmx \rangle|}, \bmv_1\rangle \geq \|\bmv_{|\langle\bmw,\bmx \rangle|}\|_2^2$ is supposed to represent $|\langle\bmw,\bmx \rangle| \geq |\langle\bmw,\bmx \rangle|^2$, which is valid because $|\langle \bmw,\bmx\rangle| \leq \|\bmw\|_1 \cdot\max |\bmx(v)| \leq 2/100 \leq 1$ by Lemma~\ref{lem:l1-radius-of-base-polytope} and the assumption that $\|F_e\|_\infty \leq 1/100$ for every $e \in E$ discussed in the beginning of Section~\ref{sec:general}.

We cannot efficiently solve the SDP relaxation~\eqref{eq:relaxed-general} because the numbers of the vectors $\bmv_{|\langle \bmw,\bmx\rangle|}$ and constraints are uncountably many.
We avoid this problem, as in the symmetric case, by replacing each $B(F_e)\;(e \in E)$ with its $\epsilon$-cover $C_e$ provided in Theorem~\ref{the:covering-base-polytope}:
\begin{align}
  \begin{array}{llll}
    \SDP_\epsilon(f) :=  & \text{minimize} & \displaystyle \frac{1}{2}\sum_{e \in E} \|\bmeta_e\|_2^2, \\
    & \text{subject to} & \displaystyle \|X\bmw\|_2^2 + \langle X\bmw, \bmv_{|\langle\bmw,\bmx \rangle|}\rangle \leq \|\bmeta_e\|_2^2 & \forall e \in E, \bmw \in C_e, \\
    & & \|\bmv_{|\langle\bmw,\bmx \rangle|}\|_2^2 = \|X\bmw\|_2^2 & \forall e \in E, \bmw \in C_e,\\
    & & \langle \bmv_{|\langle\bmw,\bmx \rangle|}, \bmv_1\rangle \geq \|\bmv_{|\langle\bmw,\bmx \rangle|}\|_2^2 & \forall e \in E, \bmw \in C_e\\
    & & \displaystyle \sum_{v \in V}\bmd_F(v)\|\bmx_v\|_2^2 = 1,\\
    & & \displaystyle \sum_{v \in V}\bmd_F(v)\bmx_v = 0.
  \end{array}\label{eq:eps-covering-general}
\end{align}
As $C_e \subseteq B(F_e)$, it is clear that $\SDP_\epsilon(f) \leq \SDP(f)$, and hence $\SDP_\epsilon(f)$ is a relaxation of~\eqref{eq:original-general}.
Further, as the size of $C_e$ is polynomial (as long as $\epsilon$ is constant), we can solve SDP~\eqref{eq:eps-covering-general} in polynomial time.

% \begin{proof}
%   $\bmv = X\bmw$ and $|\bmv| = \bmv_{|\langle\bmw,\bmx \rangle|}$.
%   We assume that $\bmv_1 = (1,0)$, $\bmv = (r \cos \theta, r \sin \theta)$, $|\bmv| = (r \cos \tau, r \sin \tau)$.
%   We know that $r \cos \tau \geq r^2$, which must hold with equality.
%   Hence, $\cos \tau = r$.
%   \begin{align*}
%     \eta_e^2
%     = \frac{r^2}{2} \Bigl(1 + \cos\theta \cos \tau + \sin\theta \sin \tau\Bigr)
%     = \frac{r^2}{2} \Bigl(1 + r\cos\theta + \sin\theta \sin \tau\Bigr)
%     \geq \frac{r^2}{2} \Bigl(1 + r\cos\theta - \sqrt{1-r^2}|\sin\theta| \Bigr)
%   \end{align*}
% \end{proof}

After solving SDP~\eqref{eq:eps-covering-general}, we sample $\bmg \in \bbR^V$ from the standard normal distribution $\caN(0,I_N)$ and then define $\bmz_+ \in \bbR^V$ as $\bmz_+(v) = \langle \bmx_v,\bmv_1 \rangle+ \delta \langle P_{\bmv_1^\bot}\bmx_v,\bmg\rangle\;(v \in V)$ and $\bmz_- \in \bbR^V$ as $\bmz_-(v) = \langle \bmx_v,\bmv_1 \rangle- \delta \langle P_{\bmv_1^\bot}\bmx_v,\bmg\rangle\;(v \in V)$.
Here, $\delta = O\bigl(1/\sqrt{\log (n^{1/\epsilon^2}m)}\bigr)$ and $P_{\bmv_1^\bot}$ is the projection matrix to the subspace orthogonal to $\bmv_1$.
Then, we return the one with the smaller Rayleigh quotient.
Intuitively, this rounding procedure places more importance on the direction $\bmv_1$ than on other directions.
Our algorithm is summarized in Algorithm~\ref{alg:general-eigenvalue}.

\begin{algorithm}[t!]
  \caption{Approximation of the smallest non-trivial eigenvalue of the normalized Laplacian of a general submodular transformation.}\label{alg:general-eigenvalue}
  \begin{algorithmic}[1]
  \Require{a submodular transformation $F\colon\set{0,1}^V \to \bbR^E$ and $\epsilon > 0$.}
  \State{Solve the SDP~\eqref{eq:eps-covering-general}.}
  %\Comment{Round the obtained solution: }
  \State{Let $\delta = \Theta(1/\sqrt{\log (n^{1/\epsilon^2}m)})$. }
  \State{Let $\bmg \in \bbR^N$ be a random vector sampled vector the standard normal distribution $\caN(0,I_N)$.}
  \State{Define $\bmz_+\in \bbR^V$ as $\bmz_+(v) = \langle \bmx_v,\bmv_1\rangle + \delta \langle P_{\bmv_1^\bot}\bmx_v, \bmg\rangle$ for each $v \in V$.}
  \State{Define $\bmz_-\in \bbR^V$ as $\bmz_-(v) = \langle \bmx_v,\bmv_1\rangle - \delta \langle P_{\bmv_1^\bot}\bmx_v, \bmg\rangle$ for each $v \in V$.}
  \If{$\caR_F(D_F^{1/2}\bmz_+) \leq \caR_F(D_F^{1/2}\bmz_-)$}
    \State{\Return $D_F^{1/2}\bmz_+$.}
  \Else
    \State{\Return $D_F^{1/2}\bmz_-$.}
  \EndIf
  \end{algorithmic}
\end{algorithm}

\subsection{Analysis}\label{subsec:general-analysis}
Now, we provide an approximation guarantee of Algorithm~\ref{alg:general-eigenvalue}.

\subsubsection{Denominator of Rayleigh quotients}
We analyze the maximum denominator of $\caR_F(D_F^{1/2}\bmz_+)$ and $\caR_F(D_F^{1/2}\bmz_-)$.

\begin{lemma}\label{lem:denominator-general}
  Let $\bmz_+,\bmz_- \in \bbR^V$ be the vectors obtained in Algorithm~\ref{alg:general-eigenvalue} on a submodular transformation $F\colon\set{0,1}^V \to \bbR^E$ (and some $\epsilon > 0$).
  Then, we have
  \[
    \frac{\delta^2}{2} \leq \max\Bigl\{\E_{\bmg}\Bigl[\sum_{v \in V}\bmd_F(v){\bmz_+(v)}^2\Bigr],
    \E_{\bmg}\Bigl[\sum_{v \in V}\bmd_F(v){\bmz_-(v)}^2\Bigr]\Bigr\} \leq 24+10\delta.
  \]
  with a probability of at least $1/50$.
\end{lemma}
\begin{proof}
  For the later convenience, we define $\alpha = \sqrt{\sum_{v \in V}\bmd_F(v)\langle \bmx_v,\bmv_1\rangle^2}$ and $\beta = \sqrt{\sum_{v \in V}\bmd_F(v)\|P_{\bmv_1^\bot}\bmx_v\|_2^2}$.
  We have
  \begin{align}
    & \max\Bigl\{\sum_{v \in V}\bmd_F(v){\bmz_+(v)}^2,\sum_{v \in V}\bmd_F(v){\bmz_-(v)}^2\Bigr\}
    =
    \max_{\sigma \in \set{-1,1}}\sum_{v \in V}\bmd_F(v){\Bigl(\langle \bmx_v, \bmv_1\rangle + \sigma\delta \langle P_{\bmv_1^\bot}\bmx_v, \bmg\rangle\Bigr)}^2 \nonumber \\
    & =
    \sum_{v \in V}\bmd_F(v)\langle \bmx_v, \bmv_1\rangle^2
    + \delta^2 \sum_{v \in V}\bmd_F(v)\langle P_{\bmv_1^\bot}\bmx_v, \bmg\rangle^2
    + \delta \Bigl|\sum_{v \in V}\bmd_F(v) \langle \bmx_v, \bmv_1\rangle \langle P_{\bmv_1^\bot}\bmx_v, \bmg\rangle\Bigr| \nonumber \\
    & =
    \alpha^2
    + \delta^2 \sum_{v \in V}\bmd_F(v)\langle P_{\bmv_1^\bot}\bmx_v, \bmg\rangle^2
    + \delta \Bigl|\sum_{v \in V}\bmd_F(v) \langle \bmx_v, \bmv_1\rangle \langle P_{\bmv_1^\bot}\bmx_v, \bmg\rangle\Bigr|
     \label{eq:denominator-general-decomposition}
  \end{align}

  For the second term of~\eqref{eq:denominator-general-decomposition}, as $\E\limits_\bmg \langle P_{\bmv_1^\bot}\bmx_v, \bmg\rangle^2 = \| P_{\bmv_1^\bot}\bmx_v\|_2^2$, by Proposition~\ref{pro:fact-8.7}, we have
  \begin{align}
    \Pr_{\bmg}\Bigl[\sum_{v \in V}\bmd_F(v)\langle P_{\bmv_1^\bot}\bmx_v, \bmg\rangle^2 \geq \frac{1}{2}\beta^2 \Bigr] \geq \frac{1}{12}. \label{eq:denominator-general-decomposition-1}
  \end{align}
  By Markov's inequality, we have
  \begin{align}
    \Pr_{\bmg}\Bigl[\sum_{v \in V}\bmd_F(v)\langle P_{\bmv_1^\bot}\bmx_v, \bmg\rangle^2 \leq 24\beta^2 \Bigr] \geq 1-\frac{1}{24}.
    \label{eq:denominator-general-decomposition-2}
  \end{align}
  For the third term of~\eqref{eq:denominator-general-decomposition}, by Mill's inequality, we have
  \begin{align*}
    & \Pr\Bigl[\Bigl| \sum_{v \in V}\bmd_F(v) \langle \bmx_v, \bmv_1\rangle \langle P_{\bmv_1^\bot}\bmx_v, \bmg\rangle\Bigr| \geq t\Bigr]
    =
    \Pr\Bigl[\Bigl|  \Bigl\langle \sum_{v \in V}\bmd_F(v) \langle \bmx_v, \bmv_1\rangle P_{\bmv_1^\bot}\bmx_v, \bmg\Bigr\rangle\Bigr|\geq t\Bigr]
    \leq \frac{1}{t}\sqrt{\frac{2}{\pi}} \exp\Bigl(-\frac{t^2}{2\sigma^2}\Bigr),
  \end{align*}
  where $\sigma = \|\sum_{v \in V}\bmd_F(v) \langle \bmx_v, \bmv_1\rangle P_{\bmv_1^\bot}\bmx_v\|_2$.
  Note that $\sigma \leq  \alpha \beta$ by the vector version of the Cauchy-Schwarz inequality (see Lemma~\ref{lem:vector-cauchy-schwarz}).
  Hence, by setting $t = 10 \alpha \beta \log \frac{1}{\alpha \beta}$, where we regard $t=0$ when $\alpha\beta=0$, we have
  \begin{align}
    \Pr\Bigl[\Bigl| \sum_{v \in V}\bmd_F(v) \langle \bmx_v, \bmv_1\rangle \langle P_{\bmv_1^\bot}\bmx_v, \bmg\rangle\Bigr| \geq 10 \alpha \beta\log\frac{1}{\alpha \beta}\Bigr]
    \leq
    \frac{1}{10 \alpha \beta \log\frac{1}{\alpha \beta}} \sqrt{\frac{2}{\pi}}\exp\Bigl(-5 \log \frac{1}{\alpha \beta}\Bigr)\leq \frac{1}{100}.
    \label{eq:denominator-general-decomposition-3}
  \end{align}
  By the union bound on~\eqref{eq:denominator-general-decomposition-1},~\eqref{eq:denominator-general-decomposition-2}, and~\eqref{eq:denominator-general-decomposition-3}, with a probability of at least $1/50$, we have  $\beta^2/2 \leq \sum_{v \in V}\bmd_F(v)\langle P_{\bmv_1^\bot}\bmx_v,\bmg\rangle^2 \leq 24\beta^2$ and $|\sum_{v \in V}\bmd_F(v)\langle \bmx_v,\bmv_1\rangle \langle P_{\bmv_1^\bot}\bmx_v,\bmg\rangle| \leq 10 \alpha \beta \log\frac{1}{\alpha \beta}$.
  In what follows, we assume this happened.

  For the upper bound, from~\eqref{eq:denominator-general-decomposition} and our assumptions, we have
  \begin{align*}
    & \max\Bigl\{\sum_{v \in V}\bmd_F(v){\bmz_+(v)}^2,\sum_{v \in V}\bmd_F(v){\bmz_-(v)}^2\Bigr\}
    \leq \alpha^2 + 24\delta^2 \beta^2 + 10 \delta \alpha \beta \log\frac{1}{\alpha \beta}
    \leq 24 + 10\delta,
  \end{align*}
  where we used the fact that $\alpha^2+\beta^2 = 1$ and the maximum of $\alpha \beta \log(1/\alpha \beta)$ subject to $\alpha^2+\beta^2=1$ is $\log(2)/2 \leq 1$.

  For the lower bound, from~\eqref{eq:denominator-general-decomposition} and our assumptions, we have
  \[
    \max\Bigl\{\sum_{v \in V}\bmd_F(v){\bmz_+(v)}^2,\sum_{v \in V}\bmd_F(v){\bmz_-(v)}^2\Bigr\} \geq \alpha^2 + \frac{\delta^2 \beta^2}{2}\geq \frac{\delta^2}{2},
  \]
  where we used the fact that $\alpha^2+\beta^2 = 1$.
  Note that the third term of~\eqref{eq:denominator-general-decomposition} does not appear because we take the maximum of $\sum_{v \in V}\bmd_F(v){\bmz_+(v)}^2$ and $\sum_{v \in V}\bmd_F(v){\bmz_-(v)}^2$.
  % \begin{align*}
  %   & \Bigl| \sum_{v \in V}\bmd_F(v) \langle \bmx_v, \bmv_1\rangle \langle P_{\bmv_1^\bot}\bmx_v, \bmg\rangle\Bigr|
  %   \leq
  %   200 \sqrt{\log n}\sum_{v \in V}\bmd_F(v) |\langle \bmx_v, \bmv_1\rangle| \|P_{\bmv_1^\bot}\bmx_v\|_2 \\
  %   & \leq
  %   200 \sqrt{\log n}\sqrt{\sum_{v \in V}\bmd_F(v) \langle \bmx_v, \bmv_1\rangle^2} \sqrt{\sum_{v \in V}\bmd_F(v) \|P_{\bmv_1^\bot}\bmx_v\|_2^2} \tag{By Cauchy-Schwarz}\\
  %   & \leq 200\sqrt{\log n}\sum_{v \in V}\bmd_F(v)\Bigl( \langle \bmx_v, \bmv_1\rangle^2+ \|P_{\bmv_1^\bot}\bmx_v\|_2^2\Bigr) \tag{by AM-GM inequality} \\
  %   & = 200\sqrt{\log n}\sum_{v \in V}\bmd_F(v)\|\bmx_v\|_2^2
  %   = 200 \sqrt{\log n}.
  % \end{align*}
\end{proof}

\subsubsection{Numerator of Rayleigh quotients}
Next, we analyze the maximum numerator of $\caR_F(D_F^{1/2}\bmz_+)$ and $\caR_F(D_F^{1/2}\bmz_-)$.

The following lemma is useful to bound the error that occurred by replacing the base polytope with its $\epsilon$-cover.
\begin{lemma}\label{lem:error-by-eps-covering-general}
  Let $F\colon\set{0,1}^V \to \bbR$ be a submodular function and let $C \subseteq B(F)$ be an $\epsilon$-cover of $B(F)$ for $\epsilon > 0$.
  Then, we have
  \[
    \max_{\bmw \in B(F)}\max\bigset{\langle \bmw,\bmx\rangle,0}^2 \leq
    \max_{\bmw \in C}\max\bigset{\langle \bmw,\bmx\rangle,0}^2 + \max\bigset{\epsilon^2, 2\epsilon\max_{\bmw \in B(F)}\|\bmw\|_2} \cdot \|\bmx|_{\supp(F)}\|_2^2
  \]
  for any $\bmx \in \bbR^V$.
\end{lemma}
\begin{proof}
  Because $\bmw(v) = 0$ for every $v \in V \setminus \supp(F)$, it suffices to show the inequality for which $\bmx$ is replaced with $\bmx|_{\supp(F)}$.

  Let $\bmw^*$ be the maximizer of $\max_{\bmw \in B(F)}\max\bigset{\langle \bmw,\bmx|_{\supp(F)}\rangle,0}^2$.
  If $\langle \bmw^*,\bmx|_{\supp(F)}\rangle < 0$, then the inequality clearly holds.
  Hence, we assume $\langle \bmw^*,\bmx|_{\supp(F)}\rangle \geq 0$.

  From the definition of $\epsilon$-cover, there exists $\bmw' \in C$ with $\|\bmw^*-\bmw'\|_2 \leq \epsilon$.
  Our goal is showing that $\langle \bmw^*,\bmx|_{\supp(F)}\rangle^2 \leq
   \max\bigset{\langle \bmw',\bmx|_{\supp(F)}\rangle,0}^2 + \max\bigset{\epsilon^2, 2\epsilon \max_{\bmw \in B(F)}\|\bmw\|_2} \cdot \|\bmx|_{\supp(F)}\|_2^2$.

  If $\langle \bmw',\bmx|_{\supp(F)}\rangle < 0 $, then we have
  \begin{align*}
    \langle \bmw^*, \bmx|_{\supp(F)} \rangle = \langle \bmw', \bmx|_{\supp(F)} \rangle + \langle \bmw^*-\bmw', \bmx|_{\supp(F)} \rangle
    < \epsilon \|\bmx|_{\supp(F)}\|_2,
  \end{align*}
  which implies $\langle \bmw^*,\bmx|_{\supp(F)}\rangle^2 \leq \epsilon^2 \|\bmx|_{\supp(F)}\|_2^2$.

  Otherwise, we have
  \begin{align*}
    & \langle \bmw^*,\bmx|_{\supp(F)}\rangle^2 -
    \max\bigset{\langle \bmw',\bmx|_{\supp(F)}\rangle,0}^2
    = \langle \bmw^*, \bmx|_{\supp(F)} \rangle^2 - \langle \bmw', \bmx|_{\supp(F)} \rangle^2 \\
    & = \langle \bmw^* - \bmw',\bmx|_{\supp(F)}\rangle \cdot \langle \bmw^* + \bmw',\bmx|_{\supp(F)}\rangle
    \leq \epsilon \|\bmx|_{\supp(F)}\|_2 \cdot 2\max_{\bmw \in B(F)}\|\bmw\|_2 \cdot \|\bmx|_{\supp(F)}\|_2 \\
    & = 2\epsilon \|\bmx|_{\supp(F)}\|_2^2 \max_{\bmw \in B(F)}\|\bmw\|_2. \qedhere
  \end{align*}
\end{proof}

In what follows, we fix a submodular transformation $F\colon\set{0,1}^V \to \bbR^E$ and $\epsilon > 0$, and let $X = (\bmx_v) \in \bbR^{N \times V}$ be the SDP solution and let $\bmz_+,\bmz_- \in \bbR^V$ be the vectors obtained by rounding $X$.
Now, we divide $\bmw \in C_e\;(e \in E)$ into two classes by the value of $\langle X\bmw,\bmv_1\rangle$.
\begin{align*}
  W_e^+ &= \Bigl\{ \bmw \in C_e \mid \langle X\bmw, \bmv_1\rangle > -\frac{1}{2}\Bigr\},
  \quad W_e^- = C_e \setminus W_e^+,\\
  W^+ & =  \bigcup_{e \in E}W_e^+, \quad
  W^- =  \bigcup_{e \in E}W_e^-.
%  J_- & = \Bigl\{ e \in E \mid \max_{\bmw \in C_e}\sum_{v \in V}\bmw(v)\langle \bmx_v, \bmv_1\rangle \leq -\frac{1}{2}\Bigr\},\\
%  J_+ & = [m] \setminus I_-.
\end{align*}
We will see that, although $\bmw \in W_e^-\;(e \in E)$ makes no contribution to the SDP value, it also does not contribute to in $f_e(\bmz^+)$ and $f_e(\bmz^-)$, and hence no loss is incurred for such $\bmw$ by rounding.
On the other hand, although $\bmw \in W_e^+\;(e \in E)$ may make a large contribution to the SDP value, we can specify its lower bound by using $\|X\bmw\|_2^2$ (instead of $\max\set{\langle X\bmw,\bmv_1\rangle,0}^2$), and hence we can use an argument similar to the symmetric case.

First, we analyze the contribution of $\bmw \in W^-$.
\begin{lemma}\label{lem:numerator-general-small-bias}
%  Let $\bmx_v \in \bbR^N\;(v \in V)$ and $\bmz_+,\bmz_- \in \bbR^n$ be the vectors obtained in Algorithm~\ref{alg:general-eigenvalue} on a submodular transformation $f\colon\bbR^n \to \bbR^E$ (and some $\epsilon > 0$).
  With a probability of at least $99/100$, we have
  \[
    \max\Bigl\{\langle \bmw, \bmz_+\rangle,\langle \bmw, \bmz_-\rangle\Bigr\}\leq 0
  \]
  for every $\bmw \in W^-$.
\end{lemma}
\begin{proof}
  We have
  \begin{align}
    & \max_{\bmw \in W^-}\max\Bigl\{\langle \bmw, \bmz_+\rangle,\langle \bmw, \bmz_-\rangle\Bigr\}
    = \max_{\bmw \in W^-}\max_{\sigma \in \set{-1,1}}\sum_{v \in V}\bmw(v)\Bigl(\langle \bmx_v,\bmv_1\rangle + \sigma\delta \langle P_{\bmv_1^\bot}\bmx_v,\bmg \rangle \Bigr) \nonumber \\
    & \leq -\frac{1}{2} + \delta \max_{\bmw \in W^-}\Bigl|\Bigl\langle \sum_{v \in V}\bmw(v) P_{\bmv_1^\bot}\bmx_v,\bmg\Bigr\rangle\Bigr|. \label{eq:W^--1}
  \end{align}
  By Lemma~\ref{pro:fact-8.6} and Markov's inequality, with a probability of $99/100$, we have
  \begin{align}
    & \delta \max_{\bmw \in W^-}\Bigl|\Bigl\langle \sum_{v \in V}\bmw(v) P_{\bmv_1^\bot}\bmx_v,\bmg\Bigr\rangle\Bigr|
    \leq 100\delta \sqrt{\log 2\sum_{e \in E}|C_e|} \max_{\bmw \in W^-} \Bigl\|\sum_{v \in V}\bmw(v) P_{\bmv_1^\bot}\bmx_v\Bigr\|_2 \nonumber \\
    & \leq 100\delta \sqrt{\log 2\sum_{e \in E}|C_e|} \max_{\bmw \in W^-} \sum_{v \in V}\bmw(v) \Bigl\| P_{\bmv_1^\bot}\bmx_v\Bigr\|_2. \label{eq:W^--2}
  \end{align}
  Note that $\|\bmx_v\|_2 \leq 1$ for any $v \in V$ and $\|\bmw\|_1 \leq 2/100=1/50$ for any $\bmw \in C_e \subseteq B(F_e)$ from Lemma~\ref{lem:l1-radius-of-base-polytope}.
  Hence, we have~$\eqref{eq:W^--2} \leq 1/50$ by choosing the hidden constant in $\delta$ to be sufficiently small.
  Then by~\eqref{eq:W^--1}, we have
  \[
    \max_{\bmw \in W^-}\max\Bigl\{\langle \bmw, \bmz_+\rangle,\max_{\bmw \in C_e}\langle \bmw, \bmz_-\rangle\Bigr\} \leq -\frac{1}{2} + \frac{1}{50} \leq 0
  \]
  with a probability of at least $99/100$.
\end{proof}

We next show that we can bound the SDP value from below by using $\|X\bmw\|_2^2$ for $\bmw \in W_e^+$.
\begin{lemma}\label{lem:numerator-general-large-bias-objective}
  For every $e \in E$, we have
  \[
    \max_{\bmw \in W_e^+} \|X\bmw\|_2^2 \leq 2\|\bmeta_e\|_2^2.
  \]
\end{lemma}
\begin{proof}
  Take an arbitrary vector $\bmw$ in $W_e^+$, and let $\theta \in [0,\pi]$ be the angle between $\bmv_{|\langle\bmw,\bmx \rangle|}$ and $X\bmw$.
  Then, we have
  \begin{align}
    \|\bmeta_e\|_2^2
    = \|X \bmw\|_2^2 + \langle X \bmw,\bmv_{|\langle\bmw,\bmx \rangle|} \rangle
    = (1+ \cos \theta)\|X \bmw\|_2^2.\label{eq:numerator-general-large-bias}
  \end{align}
  Hence, we want to provide a lower bound for $\cos \theta$.

  Let $\theta'$ be the angle between $\bmv_{|\langle\bmw,\bmx \rangle|}$ and $\bmv_1$, and let $\theta''$ be the angle between $X\bmw$ and $\bmv_1$.
  From the constraints in~\eqref{eq:eps-covering-general}, we have $\langle \bmv_{|\langle\bmw,\bmx \rangle|}, \bmv_1\rangle \geq \|\bmv_{|\langle\bmw,\bmx \rangle|}\|_2^2 = \|X\bmw\|_2^2 $, which implies that $\cos \theta' \geq \|X\bmw\|_2$.
  On the other hand, as $\bmw \in W_e^+$, we have $\cos \theta'' \geq \max\set{-1/(2\|X\bmw\|_2),-1} = -\min\set{1/(2\|X\bmw\|_2),1}$.

  We note that $\|X\bmw\|_2 \leq \sum_{v \in V}\bmw(v) \|\bmx_v\|_2 \leq 2/100 = 1/50$ by Lemma~\ref{lem:l1-radius-of-base-polytope}.
  Then, we have
  \begin{align*}
    & \cos \theta \geq \cos (\theta'+\theta'') = \cos\theta' \cos\theta'' - \sin\theta' \sin \theta'' \\
    & \geq -\|X\bmw\|_2 \cdot \min\Bigset{\frac{1}{2\|X\bmw\|_2},1} - \sqrt{1 - \|X\bmw\|_2^2} \sqrt{1 - \min\Bigset{\frac{1}{4\|X\bmw\|_2^2},1}}\\
    & = -\min\Bigset{\frac{1}{2},\|X\bmw\|_2} - \sqrt{1 + \min\Bigset{\frac{1}{4},\|X\bmw\|_2^2} - \|X\bmw\|_2^2 - \min\Bigset{\frac{1}{4\|X\bmw\|_2^2},1}} \\
    & = -\frac{1}{50} - \sqrt{1 + \frac{1}{2500} - 0 -1} = -\frac{1}{25}.
  \end{align*}
  Then, we have $\eqref{eq:numerator-general-large-bias} \geq 24/25\cdot \|X\bmw\|_2^2$, and the claim holds.
\end{proof}

Now, we show that the maximum numerator of $\caR_F(\bmz_+)$ and $\caR_F(\bmz_-)$ is roughly at most $O(\log n)$ times the SDP value with a certain probability.
We start with the following:
\begin{lemma}\label{lem:numerator-general-large-bias-rounding}
  We have
  \[
    \E\Bigl[\max\Bigl\{\sum_{e \in E}\max_{\bmw \in W_e^+}\langle \bmw, \bmz_+\rangle^2,\sum_{e \in E}\max_{\bmw \in W_e^+}\langle \bmw, \bmz_-\rangle^2 \Bigr\}\Bigr] = O\Bigl(\frac{\log n}{\epsilon^2}\SDP_\epsilon(f)\Bigr).
  \]
\end{lemma}
\begin{proof}
  We have
  \begin{align*}
    & \E\Bigl[\max\Bigl\{ \sum_{e \in E}\max_{\bmw \in W_e^+}\langle \bmw, \bmz_+\rangle^2,\sum_{e \in E}\max_{\bmw \in W_e^+}\langle \bmw, \bmz_-\rangle^2\Bigr\}\Bigr]
    \leq \E\Bigl[\sum_{e \in E}\max_{\bmw \in W_e^+}\max\Bigl\{\langle \bmw, \bmz_+\rangle^2,\langle \bmw, \bmz_-\rangle^2\Bigr\}\Bigr] \\
    & =
    \E\Bigl[\sum_{e \in E}\max_{\bmw \in W_e^+}\max_{\sigma \in \set{-1,1}}  {\Bigl(\sum_{v \in V}\bmw(v)\Bigl(\langle \bmx_v,\bmv_1\rangle + \sigma \delta \langle P_{\bmv_1^\bot}\bmx_v, \bmg\rangle \Bigr)\Bigr)}^2\Bigr] \\
    & =
    \E\Bigl[\sum_{e \in E}\max_{\bmw \in W_e^+}\max_{\sigma \in \set{-1,1}} {\Bigl(\langle X\bmw,\bmv_1\rangle + \sigma \delta \langle P_{\bmv_1^\bot}X\bmw , \bmg\rangle \Bigr)}^2 \Bigr] \\
    & \leq 4 \max_{e \in E}\log(2|C_e|) \sum_{e \in E} \max_{\bmw \in W_e^+} \Bigl\{\langle X\bmw,\bmv_1\rangle^2 + \delta^2 \langle P_{\bmv_1^\bot}X\bmw , \bmg\rangle^2\Bigr\} \tag{By Proposition~\ref{pro:fact-8.6-biased}}\\
    & \leq 4 \max_{e \in E}\log(2|C_e|) \sum_{e \in E} \max_{\bmw \in W_e^+} \| X\bmw\|_2^2 \\
    & \leq 16 \max_{e \in E}\log (2|C_e|) \sum_{e \in E} \|\bmeta_e\|_2^2 \tag{By Lemma~\ref{lem:numerator-general-large-bias-objective}}\\
    & = O\Bigl(\frac{\log n}{\epsilon^2}\SDP_\epsilon(f)\Bigr) . \qedhere
  \end{align*}

\end{proof}

\begin{lemma}\label{lem:numerator-general}
  Let $\bmz_+,\bmz_- \in \bbR^V$ be the vectors obtained in Algorithm~\ref{alg:general-eigenvalue} on a submodular transformation $F\colon\set{0,1}^V \to \bbR^E$ (and some $\epsilon > 0$).
  Then, we have
  \[
    \max\Bigl\{\sum_{e \in E}{f_e(\bmz_+)}^2,\sum_{e \in E}{f_e(\bmz_-)}^2 \Bigr\} = O\Bigl(\frac{\log n}{\epsilon^2}\SDP_\epsilon(f) + \epsilon \max_{e \in E}\|B(F_e)\|_H^2 \Bigr).
  \]
  with a probability of at least $1/100$.
\end{lemma}
\begin{proof}
  We only show the bound for $\bmz_+$ as the analysis for $\bmz_-$ is the same.
  We have
  \begin{align*}
    & \sum_{e \in E}{f_e(\bmz_+)}^2
    = \sum_{e \in E} \max_{\bmw \in B(F_e)}{\max\bigl\{\langle \bmw, \bmz_+ \rangle,0\bigr\}}^2 \\
    & = \sum_{e \in E}\max_{\bmw \in C_e}{\max\bigl\{\langle \bmw, \bmz_+ \rangle,0\bigr\}}^2 + O(\epsilon \sum_{e \in E} \|\bmz_+|_{\supp(f_e)}\|_2^2 \cdot \|B(F_e)\|_H^2) \tag{By Lemma~\ref{lem:error-by-eps-covering-general}}\\
    & = \sum_{e \in E}\max_{\bmw \in C_e}{\max\bigl\{\langle \bmw, \bmz_+ \rangle,0\bigr\}}^2 + O(\epsilon \max_{e \in E}\|B(F_e)\|_H^2 \cdot \sum_{v \in V}\bmd_F(v){\bmz_+(v)}^2).
  \end{align*}
  The second term is bounded by $O(\epsilon)$ with a probability of at least $1/50$ by Lemma~\ref{lem:denominator-general}.

  Now, we analyze the first term.
  Note that
  \begin{align*}
    & \sum_{e \in E}\max_{\bmw \in C_e}{\max\bigl\{\langle \bmw, \bmz_+ \rangle,0\bigr\}}^2\\
    & \leq
    \sum_{e \in E}\max_{\bmw \in W_e^+}{\max\bigl\{\langle \bmw, \bmz_+ \rangle,0\bigr\}}^2
    +
    \sum_{e \in E}\max_{\bmw \in W_e^-}{\max\bigl\{\langle \bmw, \bmz_+ \rangle,0\bigr\}}^2 \\
    & \leq
    \sum_{e \in E}\max_{\bmw \in W_e^+}\max\langle \bmw, \bmz_+ \rangle^2
    +
    \sum_{e \in E}\max_{\bmw \in W_e^-}{\max\bigl\{\langle \bmw, \bmz_+ \rangle,0\bigr\}}^2 \\
    & = O\Bigl(\frac{\log n}{\epsilon^2}\SDP_\epsilon(f)\Bigr) \tag{By Lemmas~\ref{lem:numerator-general-small-bias} and~\ref{lem:numerator-general-large-bias-rounding}}
  \end{align*}
  with a probability of at least $99/100$.

  By the union bound, we have the claim.
\end{proof}

\subsubsection{Consolidation of results}

\begin{proof}[Proof of Theorem~\ref{the:general-eigenvalue}]
  Let $\bmz_+,\bmz_- \in \bbR^V$ be the output of Algorithm~\ref{alg:symmetric-eigenvalue} on $f$ and $\epsilon > 0$.
  As with the proof of Theorem~\ref{the:symmetric-eigenvalue}, we can show that both $\bmz_+$ and $\bmz_-$ are feasible.
  By considering the one with the larger denominator in the Rayleigh quotient, we have the desired approximation guarantee by combining Lemmas~\ref{lem:denominator-general} and~\ref{lem:numerator-general}.
  The total time complexity is dominated by the time complexity for solving SDP~\eqref{eq:eps-covering-general}, which is ${\poly(nm)}^{\poly(1/\epsilon)}$.

  Note that we can augment the success probability to $9/10$ by running this algorithm a constant number of times and by outputting the vector with the minimum Rayleigh quotient.

  When the number of extreme points of the base polytope of each $F_e$ is bounded by $N$, we can directly solve the optimization problem obtained from~\eqref{eq:relaxed-general} by replacing each $B(F_e)$ with the set of its extreme points in $\poly(nmN)$ time.
  We can choose $\delta = O(1/\sqrt{\log(mN)})$ to make Lemma~\ref{lem:numerator-general-small-bias} goes through, and the bound claimed in Lemma~\ref{lem:numerator-general} becomes $O(\log N \cdot \SDP(f))$ because we do not need the second term and the number of points to be considered in the proof becomes $N$ instead of $|C_e|$.
  Hence, the approximation ratio is now $O(\log N \cdot \log (mN))  = O(\log^2 N + \log m \log N)$.
\end{proof}

%!TEX root=./main.tex

\section{Non-trivial Eigenvalues of Submodular Laplacians}\label{sec:eigen}

In this section, we prove Theorem~\ref{the:eigenvalues}.
We omit the proof of Theorem~\ref{the:normalized-eigenvalues} as it is obtained by replacing $L_F$ and $R_F$ by $\caL_F$ and $\caL_F$ in the proof of Theorem~\ref{the:eigenvalues}.

% In Section~\ref{subsec:heat-equation}, we introduce a diffusion process associated with a submodular Laplacian.
%As with the Laplacian for an undirected graph, a submodular Laplacian has a trivial eigenvector of $\bmone$, that is, the all-one vector, with the corresponding eigenvalue $0$.
% In Section~\ref{subsec:eigenvalue}, we use this diffusion process to show that the submodular Laplacian has a non-trivial eigenvector.

% \subsection{Diffusion process}\label{subsec:heat-equation}

We study the non-trivial eigenpairs of a submodular Laplacian by considering a diffusion process defined as follows:
\begin{definition}\label{def:diffusion}
  Let $F:\set{0,1}^V \to \bbR^E$ be a submodular transformation with $F(V) = \bmzero$.
  We stipulate that the time evolution of $\bmx \in \bbR^V$ obeys the following equation:
  \begin{align}
    \frac{\dx}{\dt} \in -L_F(\bmx). \label{eq:heat-equation}
  \end{align}
  The initial condition is given by an arbitrary vector $\bmx_0 \in \bbR^V$.
  Let $\bmx_t$ denote $\bmx$ at time $t \in \bbR_+$.
\end{definition}
The process with the Laplacian of an undirected graph is referred to as the \emph{heat equation} in the literature~\cite{Chung:2007ep,Kloster:2014wq}.
We can show that this process has a (unique) solution using the theory of diffusion inclusion~\cite{Aubin:2012wi} or the theory of monotone operators and evolution equations~\cite{miyadera1992nonlinear}.
See~\cite{Ikeda:2018tu} for more details.

Next, we show that the Laplacian $L_F$ of a submodular transformation $F:\set{0,1}^V \to \bbR^E$ with $F(V) = \bmzero$ has a non-trivial eigenpair and each non-trivial eigenpair $(\gamma,\bmz)$ satisfies $\gamma = R_F(\bmz)$.

Our strategy is to observe the value of the Rayleigh quotient in the diffusion process~\eqref{eq:heat-equation}.
%In this section, we adopt $W$ given by Lemma~\ref{lem:infinitesimally-valid} at each time $t$ to choose $\frac{\dx}{\dt}$ from $-L_F(\bmx)$.
%Now, we analyze how $R_F(\bmx)$ changes under the diffusion process~\eqref{eq:heat-equation}.
Note that, as we choose one vector from $L_F(\bmx)$ at each time $t$, we can represent $L_F$ at time $t$ as a matrix and denote it by $L_t \in \bbR^{V \times V}$.
In this section, the norm $\|\cdot\|$ always represents the $\ell_2$-norm.
% Although $R_U(\bmx)$ was shown in $t$ to be non-increasing~\cite{Louis:2014tg}, we give an alternative proof to obtain an explicit form of the derivative of $\caR_U(\bmx)$.
%Although the following is given in~\cite{Louis:2014tg}, we show its proof for completeness.
\begin{lemma}\label{lem:dx^2/dt}
  We have
  \[
    \frac{\mathrm{d}\|\bmx\|^2}{\dt}  = -2R_F(\bmx)\|\bmx\|^2.
  \]
\end{lemma}
\begin{proof}
  \[
    \frac{\mathrm{d}\|\bmx\|^2}{\dt} = 2\Bigl\langle \bmx, \frac{\dx}{\dt}\Bigr\rangle = -2 \langle \bmx, L_t \bmx\rangle = -2R_F(\bmx)\|\bmx\|^2. \qedhere
  \]
\end{proof}

%  Then, for any $t \in \bbR_+$, there exists $a > 0$ such that $\caL_\tau = \caL_{\tau'}$ for any $\tau,\tau' \in (t,t+a]$,
% and we define $\caL^+_t$ as one of such $\caL_\tau$'s. Note that $\caL^+_t$ is a linear Markov operator. We define $\caL^-_t$ analogously.
% Let $\frac{\mathrm{d}^+f}{\dt}$ be the right derivative of $f$, that is, $\frac{\mathrm{d}^+f}{\dt}\bigr|_{t=t_0} = \lim\limits_{t \to t_0+}\frac{f(t)-f(t_0)}{t-t_0}$. Similarly, let $\frac{\mathrm{d}^-f}{\dt}$ be the left derivative of $f$.
We define $\overline{\bmx} = \bmx/\|\bmx\|$.
Then, we have the following.
\begin{lemma}\label{lem:dR/dt}
  We have
  \[
    \frac{\mathrm{d}\overline{\bmx}}{\dt} = R_F(\overline{\bmx}) \overline{\bmx} -L_t \overline{\bmx}
    \quad \text{and} \quad
    \frac{\mathrm{d}R_F(\bmx)}{\dt} = 2 ({R_F(\overline{\bmx})}^2 -\|L_t \overline{\bmx}\|^2).
  \]
\end{lemma}
\begin{proof}
  From Lemma~\ref{lem:dx^2/dt},
  we have $\frac{\mathrm{d}\|\bmx\|}{\dt} = -R_F(\bmx)\|\bmx\|$.
  Then, we have
  \begin{align*}
    \frac{\mathrm{d}\overline{\bmx}}{\dt}
    & =
    \frac{ \|\bmx\| \frac{\dx}{\dt} - \bmx \frac{\mathrm{d}\|\bmx\|}{\dt}}{ \|\bmx\|^2}
    =
    \frac{ \bmx R_F(\bmx)\|\bmx\|-L_t \bmx \|\bmx\|}{\|\bmx\|^2} \\
    & =
    \frac{R_F(\bmx)\bmx }{\|\bmx\|} - \frac{L_t \bmx}{\|\bmx\|}
    =
    R_F(\overline{\bmx}) \overline{\bmx} -L_t \overline{\bmx}.
  \end{align*}
  Let $Q_F(\bmx) = \langle \bmx, L_F(\bmx)\rangle$ and $Q_t(\bmx) = \langle \bmx, L_t \bmx\rangle$.
  Note that $Q_F(\bmx)$ does not depend on the choice of $W \in \prod_{e \in E}\partial f_e(\bmx)$ used in Definition~\ref{def:submodular-Laplacian}.
  Hence, we have
  \begin{align*}
    & \frac{\mathrm{d}R_F(\bmx)}{\dt}
    = \frac{\mathrm{d}Q_F(\overline{\bmx})}{\dt}
    = \frac{\mathrm{d}Q_t(\overline{\bmx})}{\dt}
    = \Bigl\langle  \frac{\mathrm{d}Q_t(\overline{\bmx})}{\mathrm{d}\overline{\bmx}}, \frac{\mathrm{d}\overline{\bmx}}{\dt} \Bigr\rangle\\
    & =
    2 \bigl\langle L_t \overline{\bmx},  R_F(\overline{\bmx})\overline{\bmx}-L_t \overline{\bmx}  \bigr \rangle
%    & =
%    2 \bigl\langle L_t \overline{\bmx}, ( R_F(\overline{\bmx})\overline{\bmx}-L_t \overline{\bmx} ) \dt\bigr\rangle\\
    =
    2 ({R_F(\overline{\bmx})}^2 - \|L_t \overline{\bmx}\|^2). \qedhere
  \end{align*}
\end{proof}
\begin{corollary}\label{cor:non-increasing}
  $R_F(\bmx)$ is non-increasing in $t$.
\end{corollary}
\begin{proof}
  Note that
  %\[
    $\| L_t \overline{\bmx}\|
        \geq \langle \overline{\bmx}, L_t\overline{\bmx}\rangle
        =R_F(\overline{\bmx})$.
  %\]
  The inequality holds because $\overline{\bmx}$ is a unit vector. From Lemma~\ref{lem:dR/dt}, we have $\frac{\mathrm{d}R_F(\bmx)}{\dt} \leq 0$.
  Since $R_F$ is a continuous function of $t$, we have the desired result.
\end{proof}

\begin{theorem}\label{the:converge-to-eigenvalue}
  Suppose that we initiate a simulation of the diffusion process~\eqref{eq:heat-equation} with a non-zero vector $\bmx_0 \bot \bmone$.
  Then, as $t \to \infty$, $\bmx$ and $R_F(\bmx)$ converge to some $\bmz \in \bbR^V$ and $\gamma \in \bbR_+$, respectively, such that
  \[
    \bmz \bot \bmone, \quad
    L_F(\bmz) \ni \gamma \bmz, \quad \text{and} \quad \gamma = R_F(\bmx).
  \]
\end{theorem}
\begin{proof}
  Note that $R_F(\bmx)$ is bounded from below by 0 from Lemma~\ref{lem:positive-semidefiniteness}.
  Since $R_F(\bmx)$ is non-increasing from Corollary~\ref{cor:non-increasing}, $R_F(\bmx)$ converges to some non-negative value as $t \to \infty$.

  Let $\gamma \in \bbR_+$ be the limit.
  We have $\lim\limits_{t \to \infty}\|L_t \overline{\bmx}\|^2 = \gamma^2$ by Lemma~\ref{lem:dR/dt}.
  It follows that $\lim\limits_{t \to \infty} \bigl\langle L_t \overline{\bmx} - \gamma\overline{\bmx}, L_t \overline{\bmx}\bigr\rangle = 0$. Since $\lim\limits_{t \to \infty} \| L_t \overline{\bmx}\| = \gamma$ and $\|\overline{\bmx}\| = 1$, we must have $\lim\limits_{t \to \infty}L_t \overline{\bmx} - \gamma\overline{\bmx} = \bmzero$ or $\lim\limits_{t \to \infty}L_t \overline{\bmx} = \bmzero$.
  However, the latter implies that $\gamma = 0$.
  Hence, we have $\lim\limits_{t \to \infty}L_t \overline{\bmx} - \gamma\overline{\bmx} = \bmzero$ in both cases.
  In particular, this means $\lim\limits_{t \to \infty}\frac{\mathrm{d}\overline{\bmx}}{\dt} \to \bmzero$ by Lemma~\ref{lem:dR/dt}.
  As $\overline{\bmx}$ is bounded, $\overline{\bmx}$ converges to a vector $\bmz$, which is an eigenvector of $L_F$ with the eigenvalue $\gamma = R_F(\bmz)$.

  It is clear that $\bmz \bot \bmone$ because we always have $\dx \bot \bmone$ when we start the diffusion process with a vector $\bmx_0 \bot \bmone$.
\end{proof}
Theorem~\ref{the:converge-to-eigenvalue} immediately implies Theorem~\ref{the:eigenvalues}.

% Let $\lambda_F$ be the smallest eigenvalue of $ L_F$.
% The following lemma is helpful for understanding the relation between the eigenvalues of $ L_F$ and $L_F$.
% \begin{lemma}\label{lem:relation-between-eigenvalues}
%   It holds that
%   \[
%     \lambda_{U} = \inf_{\bmx \in U}R_F(\bmx),
%   \]
%   where $\lambda_F$ is the smallest eigenvalue of $L_F$.
%   In particular,
%   \[
%     \lambda_G = \inf_{\bmx \in \bmmu_G^\bot}R_F(\bmx).
%   \]
% \end{lemma}
% \begin{proof}
%   Let $\bmz$ be the eigenvector of $\Pi_{U}L_F$ associated with $\lambda_F$. Note that $\bmz \in U$ holds because $\Pi_{U}\caL_F(\bmz) = \lambda_F \bmz$. Then, we have
%   \[
%     \lambda_F
%     = \inf_{\bmx \in \bbR^V} \frac{\bmx^T \caL_F(\bmx)}{\bmx^T \bmx}
%     = \inf_{\bmx \in U} \frac{\bmx^T \caL_F(\bmx)}{\bmx^T \bmx}
%     = \inf_{\bmx \in U}\caR(\bmx).
%   \]
%   The second claim is immediate.
% \end{proof}

\section*{Acknowledgments}
The authors would like to thank Tasuku Soma for many useful discussions.

\bibliographystyle{habbrv}
\bibliography{main}

%!TEX root=main.tex

\appendix

\section{Facts on Normal Distributions}

We review several facts on normal distributions.
\begin{proposition}[Fact~6 of~\cite{Chan:2018eu}]\label{pro:fact-8.6}
  Suppose $X_1,X_2,\ldots,X_n$ are normal random variables that are not necessarily independent such that $\E[X_i] = 0\;(i\in [n])$ and $\E[X_i^2]  = \sigma_i^2\;(i \in [n])$.
  Then, we have $\E[\max\limits_{i\in[n]} X_i^2] \leq 4\sigma^2 \log n$ and $\E[\max\limits_{i\in[n]} X_i^2] \leq 2\sigma \sqrt{\log n}$, where $\sigma := \max\limits_{i \in [n]} \sigma_i$.
\end{proposition}

By slightly changing the proof of Proposition~\ref{pro:fact-8.6}, we can show a similar bound for biased normal random variables.
\begin{proposition}\label{pro:fact-8.6-biased}
  Suppose $X_1,X_2,\ldots,X_n$ are normal random variables that are not necessarily independent such that $\E[X_i] = \mu_i\;(i\in [n])$ and $\E[X_i^2]  = \sigma_i^2\;(i \in [n])$.
  Then, we have $\E[\max\limits_{i\in[n]} X_i^2] \leq 4\lambda^2 \log n$ and $\E[\max\limits_{i\in[n]} X_i^2] \leq 2\lambda \sqrt{\log n}$, where $\lambda := \max\limits_{i \in [n]} \sqrt{\mu_i^2 + \sigma_i^2}$.
\end{proposition}
\begin{proof}
  For $i \in [n]$, we write $X_i = \mu_i + \sigma_i Z_i$, where $Z_i$ has the standard normal distribution $\caN(0,1)$.
  Observe that, for any real numbers $x_1, x_2, \ldots, x_n$ and  positive integer $p$, we have $\max_{i \in [n]} x_i^2 \leq {(\sum_{i\in [n]} x_i^{2p})}^{1/p}$.
  Hence, we have
  \begin{align*}
    & \E \Bigl[\max_{i\in [n]} X_i^2\Bigr]
    \leq \E \Bigl[{\Bigl(\sum_{i \in [n]}X_i^{2p}\Bigr)}^{1/p}\Bigr]\\
    & \leq {\Bigl(\E\Bigl[ \sum_{i \in [n]}X_i^{2p}\Bigr]\Bigr)}^{1/p} \tag{by Jensen's Inequality because $t \mapsto t^{1/p}$ is concave}\\
    & \leq \lambda^2 {\Bigl(\E \Bigl[\sum_{i \in [n]}Z_i^{2p}\Bigr]\Bigr)}^{1/p}\\
    & = \lambda^2 {\Bigl( \sum_{i \in [n]}\frac{(2p)!}{p!2^p}\Bigr)}^{1/p} \tag{$\E[Z_i^{2p}] = \frac{(2p)!}{p!2^p}$.}\\
    & \leq \sigma^2 p d^{1/p} \tag{using $\frac{(2p)!}{p!} \leq(2p)^p$.}
  \end{align*}
  Selecting $p = \lceil \log n\rceil$ provides the first result $\E[\max_{i\in [n]} X_i^2] \leq  4\lambda^2 \log n$.
  Moreover, the inequality $\E[|X|] \leq \sqrt{\E X^2}$ immediately provides the second result.
\end{proof}

\begin{proposition}[Fact~8.7 of~\cite{Chan:2018eu}]\label{pro:fact-8.7}
  Let $X_1,\ldots , X_n$ be normal random variables that are not necessarily independent $\E[X_i] = 0\;(i\in [n])$ and $\E[\sum_{i \in [n]}X_i^2]  = 1\;(i \in [n])$.
  Then, we have
  \[
    \Pr\Bigl[\sum_{i \in [n]}X_i^2 \geq \frac{1}{2}\Bigr] \geq \frac{1}{12}.
  \]
\end{proposition}

\section{Inequalities}

The following vector version of the Cauchy-Schwarz inequality holds:
\begin{lemma}\label{lem:vector-cauchy-schwarz}
  For $\alpha_1,\ldots,\alpha_k \in \bbR$ and $\bmv_1,\ldots,\bmv_k \in \bbR^n$, we have
  \[
    \Bigl\|\sum_{i\in[k]}\alpha_i \bmv_k\Bigr\|_2^2
    \leq \sum_{i \in [k]}\alpha_i^2 \cdot \sum_{i \in [k]}\|\bmv_i\|_2^2.
  \]
\end{lemma}
\begin{proof}
  \begin{align*}
    & \Bigl\|\sum_{i\in[k]}\alpha_i \bmv_k\Bigr\|_2^2
    =
    \sum_{j \in [n]} {\Bigl(\sum_{i\in[k]}\alpha_i \bmv_k(j)\Bigr)}^2
    \leq
    \sum_{j \in [n]} {\Bigl(\sqrt{\sum_{i\in[k]}\alpha_i^2} \sqrt{\sum_{i \in [k]}{\bmv_k(j)}^2}\Bigr)}^2
    =
    \sum_{j \in [n]}\Bigl(\sum_{i\in[k]}\alpha_i^2\Bigr)\Bigl( \sum_{i \in [k]}{\bmv_k(j)}^2 \Bigr)\\
    & =
    \Bigl(\sum_{i\in[k]}\alpha_i^2 \Bigr)\Bigl(\sum_{j \in [n]} \sum_{i \in [k]}{\bmv_k(j)}^2\Bigr) = \sum_{i \in [k]}\alpha_i^2 \cdot \sum_{i \in [k]}\|\bmv_i\|_2^2. \qedhere
\end{align*}
\end{proof}

\section{Lower Bounds on the Number of Balls for Covering Extreme Points of a Base Polytope}\label{apx:covering}
In this section, we show that there exists a submodular function $F\colon2^V \to \bbR$ such that, to cover all the extreme points of the base polytope $B(F)$ using $\ell_2$-balls of radius $\epsilon \|B(F)\|_H$, we need almost as many number of balls as in Theorem~\ref{the:covering-base-polytope}.

Let $V=[n]$ be a finite set of $n$ elements.
For an integer $k \leq n/2$, consider a submodular function $F_k\colon\bbR^V\to \bbR$ with $F_k(S) = \min\set{|S|, k}\;(S \subseteq V)$, which is the rank function of a uniform matroid, and define $G_k\colon \bbR^V \to \bbR$ as its symmetrize version, that is, $G_k(S) = F_k(S) + F_k(V\setminus S)-F_k(V)\;(S \subseteq V)$.
It is easy to verify that the extreme points of $B(G_k)$ are of the form $\bmx \in \set{-1/2,0,1/2}^V$ with $\#\set{i \in [n] \mid \bmx(i)=-1/2}=\#\set{i \in [n] \mid \bmx(i)=1/2} = k$.
Note that the number of extreme points is $\binom{n}{k}\cdot \binom{n-k}{k} = \frac{n!}{k!k!(n-2k)!}$ and $\|B(F)\|_H = \sqrt{k/2}$.
Also, a ball of radius $r \in \bbR_+$ in $\ell_2$-norm can cover at most $\sum_{i=0}^{4r^2}\binom{n}{i}2^i$ extreme points.
Hence, to cover all the extreme points using balls of radius $\epsilon \|B(F)\|_H=\epsilon\sqrt{k/2}$, we need at least
\[
  N_k(\epsilon) := \frac{n!}{k!k!(n-2k)!\sum_{i=0}^{2\epsilon^2 k} \binom{n}{i}2^i}
\]
balls.
Noticing that
\[
  \sum_{i=0}^{r} \binom{n}{i}2^i \leq \sum_{i=0}^{r} \frac{{(2n)}^i}{i!} =\sum_{i=0}^{r} \frac{r^i}{i!} {\Bigl(\frac{2n}{r}\Bigr)}^i = e^r{\Bigl(\frac{2n}{r}\Bigr)}^r,
\]
we have by Stirling's formula that
\begin{align*}
  & \log N_k(r) \geq \log \frac{\sqrt{2\pi}n^{n+1/2}e^{-n}}{{(e k^{k+1/2}e^{-k})}^2e {(n-2k)}^{(n-2k)+1/2}e^{-(n-2k)} {(en/\epsilon^2 k)}^{2\epsilon^2 k}} \\
  & = \Omega\left(\left(n+\frac12\right)\log n-n - 2 \left(\left(k+\frac12\right)\log k-k\right) - \left(n-2k+\frac12\right)\log (n-2k) + (n-2k) -2\epsilon^2 k \log \frac{n}{\epsilon^2 k} \right) \\
%  & = \Omega\left(2k\log (n-2k) -2k - 2 \left(\left(k+\frac12\right)\log k-k\right)  -2\epsilon^2 k \log \frac{n}{2\epsilon^2 k} \right) \\
  & = \Omega\left(2k\log (n-2k)  - 2\left(k+\frac12\right)\log k  -2\epsilon^2 k \log \frac{n}{\epsilon^2 k} \right) \\
  & = \Omega\left(2k\log \frac{n-2k}{k} - \log k  -2\epsilon^2 k \log \frac{n}{\epsilon^2 k} \right).
\end{align*}
Then, for any small constant $\epsilon \in (0,1)$, by choosing $k = 1/\epsilon^2$, we have
$\log N_k(r) = \Omega\bigl(\log(\epsilon^2 n)/\epsilon^2\bigr)$.
Recalling that the logarithm of the number of balls required in Theorem~\ref{the:covering-base-polytope} is $O\bigl(\log U_{\ref{lem:covering-number-of-l1-ball}}(\epsilon/6,B_1^V)+ \log \log n\bigr) = O\bigl(\log(\epsilon^2 n)/\epsilon^2 + \log \log n\bigr)$, we can conclude that the bound in Theorem~\ref{the:covering-base-polytope} is almost tight.

\section{Expressing Deep Neural Network}\label{sec:dnn}

A typical feed-forward neural network $f\colon\bbR^n \to \bbR$ used in deep learning is of the following form:
\[
  f(\bmx) = W_L(\sigma_{L-1}(W_{L-1}(\cdots \sigma_2(W_2\sigma_1(W_1 \bmx))))),
\]
where $W_\ell\in \bbR^{d_\ell \times d_{\ell-1}}\;(\ell \in \set{1,\ldots,L})$ is a matrix with $d_0 = n$ and $d_L=1$ and $\sigma_\ell\colon\bbR^{d_\ell} \to \bbR^{d_\ell}\;(\ell \in \set{1,\ldots,L-1})$ is a rectified linear unit (ReLU), which applies the following operation coordinate-wise: $x \mapsto \max\set{x,0}$.
Then, in the regression setting with the $\ell_2$-norm loss, given training examples $(\bmx_1,y_1),\ldots,(\bmx_m,y_m) \in \bbR^n \times \bbR$, we aim to find $W_1,\ldots,W_L$ that minimizes the loss function $\sum_{i=1}^m\|f(\bmx_i) - y_i\|_2^2$.
As the loss function is non-convex, we cannot hope to obtain the global minimum in polynomial time, and hence we want to analyze the structure of local minima.
When ReLUs are not applied in a neural network, every local minimum is known to be a global minimum (under a plausible assumption)~\cite{Kawaguchi:2016ub}.
However, the proof heavily relies on elegant properties of linear transformations and it does not generalize to the case with ReLUs.

Note that the function $\max\set{x-y,0}$ is the \Lovasz extension of the cut function of the directed graph consisting of a single arc $(x,y)$.
Using this fact, we can express the feed-forward neural network $f$ as an iterated applications of \Lovasz transformations.
First, define $W'_\ell \in \bbR^{(d_\ell+1)\times d_{\ell-1}}$ as the matrix obtained from $W_\ell$ by adding the all-zero row vector.
Then, we define $\sigma'_\ell\colon\bbR^{d_\ell} \to \bbR^{d_\ell}$ as $\sigma'_\ell(\bmx)(i) = \max\set{\bmx(i)-\bmx(d_\ell+1),0}$.
Finally, we define $f'\colon\bbR^n \to \bbR$  as
\[
  f'(\bmx) = W_L(\sigma'_{L-1}(W'_{\ell-1}(\cdots \sigma'_2(W'_2\sigma'_1(W'_1 \bmx))))).
\]
We can observe that $\sigma'_\ell\;(\ell \in \set{1,\ldots,L-1})$ acts as a ReLU because the last element of the vector given to $\sigma'_{\ell}$ is always zero.
Hence, we have $f' \equiv f$.
This observation implies that we could deepen the understanding of deep learning by studying \Lovasz transformations.

Indeed, the smallest non-trivial eigenvalue of the Laplacian $L_F$ of a submodular transformation $F\colon\set{0,1}^V \to \bbR^E$ is equal to  $\min_{\bmx \bot \bmone} \|f(\bmx)\|_2/\|\bmx\|_2$ for the corresponding \Lovasz transformation $f\colon\bbR^V \to \bbR^E$, which can be regarded as the smallest non-trivial singular value of $f$. (The connection will become clear in Section~\ref{sec:Laplacian}.)
Hence, this work can be seen as the first step toward extending linear algebra to the algebra with \Lovasz transformations, or \emph{submodular algebra}.

\end{document}